\documentclass{article}

\usepackage[T1]{fontenc}
\usepackage{palatino}
\usepackage{fullpage}
\usepackage{graphicx}
\usepackage{amsmath}
\usepackage{amsfonts}
\usepackage{amssymb}
\usepackage[dvipsnames,svgnames,x11names,hyperref]{xcolor}
\definecolor{lkcol}{RGB}{130,70,200}
\usepackage[hidelinks,colorlinks=true,linkcolor=lkcol,citecolor=blue,linktocpage,pdfusetitle]{hyperref}
\usepackage{cite}
\usepackage{csquotes}
\usepackage{caption}
\usepackage{subcaption}
\usepackage{nicematrix}
\usepackage{mathtools}
\usepackage{tabularray}
\usepackage{float}
\usepackage{amsthm}
\usepackage{dsfont}
\usepackage{calc}
\usepackage{fancyhdr}
\usepackage{titling}
\usepackage{centernot}
\usepackage{orcidlink}
\usepackage[boxsize=5pt]{ytableau}
\usepackage{comment}
\usepackage{bbold}
\usepackage{soul}
\usepackage{titling}
\posttitle{\par\end{center}\vskip 2cm}

\numberwithin{equation}{section}

\newcommand{\IC}{\mathbb{C}}
\newcommand{\C}{\mathbb{C}}

\newcommand{\IQ}{\mathbb{Q}}
\newcommand{\Q}{\mathbb{Q}}
\newcommand{\IR}{\mathbb{R}}
\newcommand{\IZ}{\mathbb{Z}}

\newcommand{\Sk}{\mathrm{Sk}}

\newcommand\sfW{\mathsf{W}}
\newcommand\sfP{\mathsf{P}}
\newcommand\sfPsi{\mathsf{\Psi}}

\def\be{\begin{equation}}
\def\ee{\end{equation}}

\def\IR{{\mathbb{R}}}
\def\R{{\mathbb{R}}}

\def\IZ{{\mathbb{Z}}}
\def\IP{{\mathbb{P}}}
\def\IC{{\mathbb{C}}}
\def\C{{\mathbb{C}}}

\def\CL{\mathcal{L}}

\def\CO{\mathcal{O}}

\def\CN{{\mathcal{N}}}

\def\Li{\mathrm{Li}}

\def\disk{{\mathrm{disk}}}

\def\disk{{\rm{disk}}}

\def\bd{\boldsymbol{d}}

\def\gcd{{\mathrm{gcd}}}

\def\Trop{\Sigma^{\mathrm{Trop}}}

\newtheorem{theorem}{Theorem}
\newtheorem{remark}[theorem]{Remark}

\newtheorem{corollary}[theorem]{Corollary}

\newtheorem{proposition}[theorem]{Proposition}

\numberwithin{theorem}{section}

\usepackage{authblk}

\title{
Open strings on knot complements
}
\author[1,2,3]{Sachin Chauhan\orcidlink{0000-0003-1045-3428}\thanks{sachin.chauhan@math.uu.se}}
\author[1,2,4]{Tobias Ekholm\orcidlink{0000-0003-2060-8462}\thanks{tobias.ekholm@math.uu.se}}
\author[1,2,3]{Pietro Longhi\orcidlink{0000-0002-5558-0386}\thanks{pietro.longhi@physics.uu.se}}
\affil[1]{Centre for Geometry and Physics, Uppsala University, Box 516, 751 20 Uppsala, Sweden}
\affil[2]{ Department of Mathematics, Uppsala University, Box 480, 751 06 Uppsala, Sweden}
\affil[3]{Department of Physics and Astronomy, Uppsala University, Box 516, 751 20 Uppsala, Sweden}
\affil[4]{Institut Mittag-Leffler, Aurav 17, 182 60 Djursholm, Sweden}
\date{}                     
\fancypagestyle{firstpage}
{
    \fancyhead[L]{}    
    \fancyhead[R]{UUITP-32/25}
}

\begin{document}

\maketitle
\thispagestyle{firstpage}

\begin{abstract}
Using skein valued holomorphic curve counting techniques, we give a flow loop formula for the skein valued partition function of the Lagrangian knot complement of a fibered knot (of the $A$-model open topological strings with Lagrangian $A$-branes wrapping the complement) in the cotangent bundle of the three-sphere and in the resolved conifold. For torus knots we show that the partition function in the cotangent bundle localizes on two or three holomorphic annuli and give a corresponding generalized quiver structure for the partition function in the resolved conifold.

We connect the formula to the augmentation curve, the representation variety of the knot contact homology algebra of the knot,  generated by Reeb chords of its Legendrian conormal and with differential given by holomorphic disks interpolating between words of Reeb chords. The curve admits a quantization as a $q$-difference equation for the generating function of symmetrically colored HOMFLYPT-polynomials of the knot or, geometrically, for the $U(1)$-partition function of the knot conormal. For $(2,2p+1)$-torus knots we show that, after a change of variables, the partition function of the knot complement also satisfies this $q$-difference equation. This gives another geometrically defined coordinate chart for the $D$-module defined by the quantized augmentation polynomial. 

\end{abstract}

\newpage

\tableofcontents

\section{Introduction}


The topological $A$-model is a BPS sector of string theory that corresponds to worldsheets mapping to holomorphic curves in some K\"ahler target space $X$. The model can be enriched by introducing $A$-brane boundary conditions where worldsheets can end, so that one counts maps from some Riemann surface with boundary into $X$ with the boundary mapping to a Lagrangian $A$-brane. 

This subject has a, by now, long history in physics, where explicit computations were carried out using combinations of dualities (large $N$, 3d-3d, etc.) and mirror symmetry, see e.g., \cite{Aganagic:2000gs, Aganagic:2001nx, Aganagic:2003db, Ooguri:1999bv}. However, a general formulation of deformation invariant \emph{open Gromov-Witten theory} was missing until recently when \cite{Ekholm:2019yqp} showed that if holomorphic curves in a symplectic Calabi-Yau 3-fold with Maslov zero Lagrangian boundary are counted by their values in the HOMFLYPT skein module of the Lagrangian boundary condition then the curve count is deformation invariant.  
Skein valued curve counts recover the open string partition function and explain the connection to augmentation varieties from knot contact homology \cite{Aganagic:2013jpa} and their quantization (corresponding quantum curves) as well as corresponding phenomena for higher rank local systems on $A$-branes, see \cite{Ekholm:2018iso, Ekholm:2021osm, Ekholm:2020csl, Ekholm:2024ceb}.

In this paper we consider Lagrangian knot complements $M_K = S^3\setminus K$ of fibered knots $K$ in $T^\ast S^3$ and after transition to the resolved conifold in $X$, the total space of the bundle $\CO(-1)\oplus \CO(-1)\to\C P^1$. In physics language this means we consider the open string $A$-model in $X$ with an $A$-brane supported, at weak string coupling, on $M_K$.

We show that the open string partition function can be evaluated by a localization argument. Let $\beta$ be a generic  closed nowhere vanishing $1$-form on $M_K$. Let $\mathcal{F}$ be the set of closed flow loops of $\beta$. If $\gamma\in\mathcal{F}$ define the action of $\gamma$ to be $\mathfrak{a}(\gamma)=\int_\gamma\beta$. Then for any $\mathfrak{a}_0>0$ the subset $\{\gamma\in\mathcal{F}\colon \mathfrak{a}(\gamma)<\mathfrak{a_0}\}$ is finite. Shifting $M_K$ by $\beta$ in $T^\ast S^3$ gives a Lagrangian in $T^\ast S^3\setminus S^3$ and to each flow loop $\gamma\in \mathcal{F}$ we associate an (approximately) holomorphic annulus $A(\gamma)$ stretching between $S^3$ and $M_K$, lying above $\gamma$. We write $\partial A(\gamma)$ for the two boundary components of $A(\gamma)$ which are copies of $\gamma$ in $M_K$ and $S^3$ oriented as the boundary of $A(\gamma)$, i.e., in opposite directions

Fix orientations and spin structures on $S^3$ and on $M_K$. The orientations and spin structures together with the linearized return map of a flow along $\gamma$ determine the type of holomorphic annulus $A(\gamma)$. There are four types, see Section \ref{ssec:fibered and flow loops}, and associated to each type there is a skein valued partition function 
\[
\sfPsi_{\mathrm{ann}}^{(\sigma,\epsilon)} \in \Sk(S^1\times D^2)^{\otimes2},\quad (\sigma,\epsilon)=(\pm 1,\pm 1),
\]
see \eqref{eq:annulus partition function}. 

If $\gamma$ is a flow loop in $M_K$, and $\eta\in \Sk(S^1\times D^2)^{\otimes2}$ we write $\iota_{\partial A(\gamma)}(\eta)$ for the skein element in $M_K\sqcup S^3$ obtained by inserting $\eta$ in a tubular neighborhood around $\partial A(\gamma)$.  

\begin{theorem}\label{t:skeinvalued knot complement}
Let $M_K$ be the knot complement of a fibered knot $K$ with fixed orientation and spins structure. Pick a generic nowhere vansishing 1-form on $M_K$ and let $\mathcal{F}$ denote the set of flow loops.
   The skein valued partition function of $M_K$ is given by 
   \begin{equation}\label{eq:localization-formula-skein}
   \mathsf{Z}_{(T^\ast S^3,M_K)}=\prod_{\gamma\in\mathcal{F}} \iota_{\partial A(\gamma)}(\sfPsi^{(\sigma(\gamma),\epsilon(\gamma))}_{\mathrm{ann}}),
   \end{equation}
   where the product simply means, consider the union of all insertions in the skein of $M_K\sqcup S^3$. Note that $\mathsf{Z}_{(T^\ast S^3,M_K)}$ is naturally filtered by action and that below any given action, $\mathsf{Z}_{(X,M_K)}$ is a skein element which is represented by a finite linear combination of links.

   Furthermore, the skein valued partition function of $M_K$ in $X$ is obtained by taking the HOMFLYPT polynomials of the skein elements in $S^3$, writing $H_{S^3}$ for the HOMFLYPT polynomial:
   \[
   \mathsf{Z}_{(X,M_K)}=(H_{S^3}\otimes 1)\circ \mathsf{Z}_{(T^\ast S^3,M_K)}.
   \]
\end{theorem}
Physically, the skein valued partition function of $M_K$ corresponds to the path integral of Chern-Simons theory with unitary gauge group on $M_K$ of arbitrary rank, with Wilson line insertions arising from boundaries of open strings. In this regard, the fact that the entire partition function organizes into contributions from (generalized) annuli, implies that the open string partition function on $M_K$ admits a lift to M-theory with stacks of M5 branes on $M_K$ and $S^3$ and with M2 branes wrapping the annuli. Integrality structures of the open string free energy and corresponding quiver structures can be deduced from this picture, see Remark \ref{rmk:integrality}.

If we restrict to the $\mathfrak{gl}_1$ skein on $M_K$, the partition function takes a more concrete form. 
The positive $\mathfrak{gl}_1$ skein of $M_K$ is the power series ring generated by a meridian $\mu$ of $K\subset S^3$. 
Only Wilson lines in symmetric representations contribute, and $\mu$ corresponds to the expectation value for the fundamental representation wrapped along the meridian.
We write $\psi_{(X,M_K)}$ for the $\mathfrak{gl}_1$ partition function of $M_K$ in $X$. For $\boldsymbol{\gamma}=\{\gamma_1,\dots,\gamma_m\}\subset\mathcal{F}$, let $d_{\boldsymbol{\gamma}}=(d_1,\dots,d_m)$ denote a vector of multiplicities, multiplicity $d_j$ of the flow loop $\gamma_j$. For $\beta,\gamma\in\mathcal{F}$, let $\Gamma_{\beta\gamma}$ denote the linking matrix with entries the linking number between $\beta$ and $\gamma$ in $M_K$. (Here we define the linking number as the intersection number $\Gamma_{\beta\gamma}=\sigma\cdot\beta$, where $\sigma$ is an oriented $2$-chain with $\partial\sigma=\beta - (k\cdot\mu)$, where $[\beta]=k[\mu]\in H_1(M_K)$ and $k\cdot \mu$ denotes $k$ parallel copies of $\mu$ in $\partial M_K$.)

\begin{corollary}\label{c:U(1) knot complement}
If $X$ is the cotangent bundle of $S^3$ or the resolved conifold and if $M_K$ is the complement of a fibered knot $K$ then 
\be\label{eq:localization-formula-U(1)}
	\psi_{(X,M_K)}(a,q,\mu) 
	= \sum_{d_1,\dots, d_m\geq 0}
	H_{(d_1)^\vee,\dots, (d_m)^{\vee}}(a,q)
	\,
	q^{\bd^t \cdot\Gamma\cdot \bd}
	\,
	\prod_{j=1}^{m}(\sigma_j \epsilon_j q^{\beta_j} \mu^{m_j})^{d_j}
\ee
where $H_{(d_1)^\vee,\dots, (d_m)^{\vee}}(a,q)$ is the HOMFLYPT polynomial (unnormalized, including denominators) of the link with components $\gamma\subset S^3$ for all non-zero entries of $d_\gamma=(d_1,\dots, d_m)$ in representation $(d_j)^{\vee}$. 
Here $((d_j)^{\vee}, \sigma_j, \epsilon_j)$ are $((d_j)^t, -1, +1)$ for elliptic flow loops, $((d_j), +1, +1)$ for hyperbolic flow loops, and $((d_j), +1, -1)$ for negative hyperbolic flow loops. See \eqref{eq:flowloopsandannuli} and \eqref{eq:d-gamma-vee}. 
The powers $\beta_j$ correspond to the self linking of the flow loop $\gamma_j$ and $m_j$ is its homology class.

\end{corollary}

For general fibered knots both \eqref{eq:localization-formula-skein} and \eqref{eq:localization-formula-U(1)} are very complicated expressions. In fact, \cite{ghrist} shows that for 2-bridge knots that are not torus knots, every link type appears as a closed flow loop in the knot complement and that for other fibered knots the set of isotopy classes of flow loops is similarly infinite. 
This is interesting from the point of view of augmentation polynomials and their quantizations, that give polynomial recursion relations for the partition function in~\eqref{eq:localization-formula-U(1)} even though there are contributions from the HOMFLYPT polynomial of any knot, which indicates that there are substantial cancellations.

\begin{figure}[h!]
\begin{center}
\includegraphics[width=0.55\textwidth]{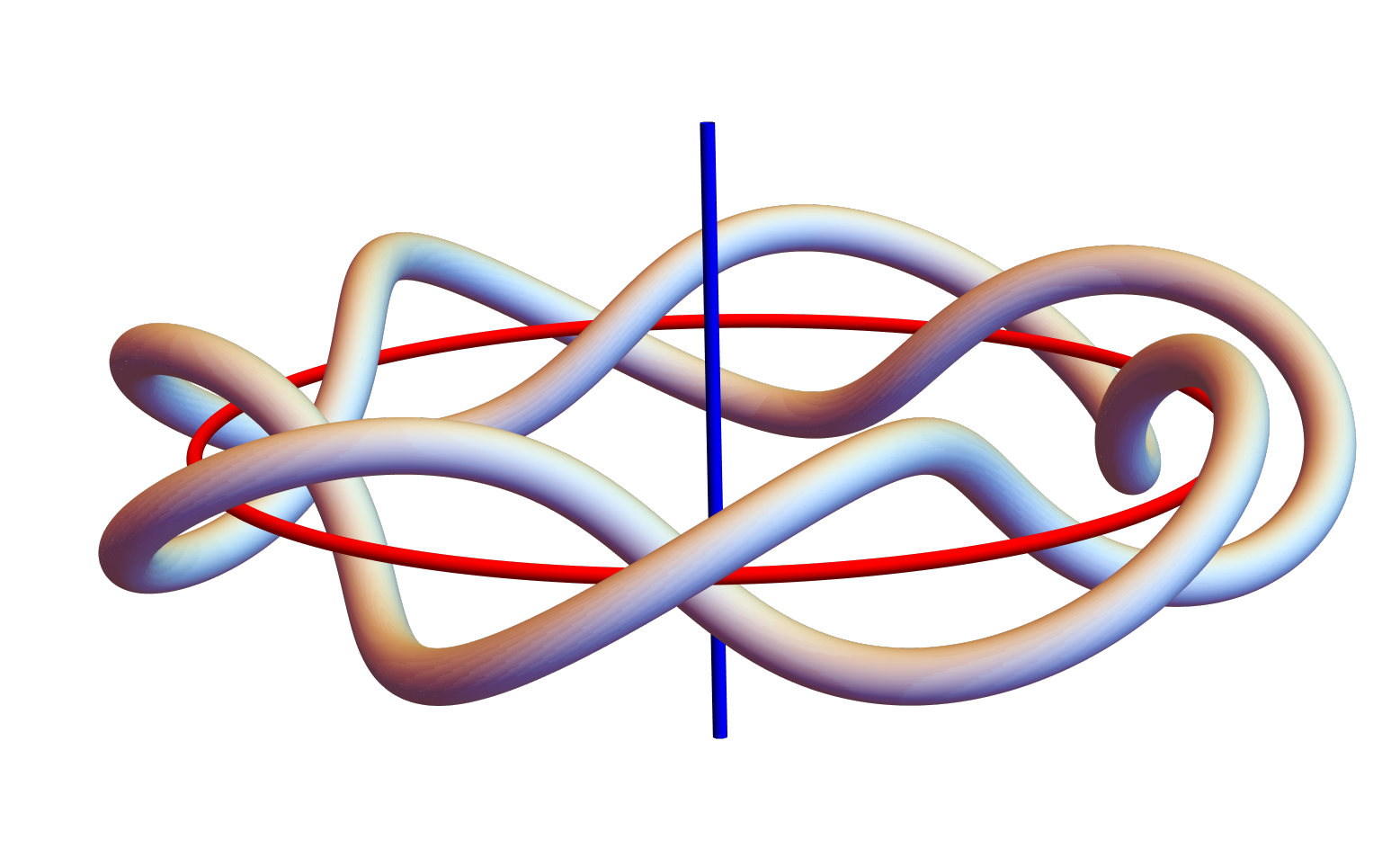}
\caption{The complement of a torus knot $T_{2,2p+1}$ with $p=3$, shown with opposite orientation (negative crossings), and the Hopf link formed by its flow loops.
}
\label{fig:Hopf-torus-intro}
\end{center}
\end{figure}

In this paper we focus on the simpler case of negative torus knots, where the partition functions localize on a finite number of 
flow loops. This gives explicit partition function formulas for Lagrangian knot complements. In the case of negative $(2,2p+1)$-torus knots, see Proposition \ref{prp:2,2p+1torusknot}, there are further cancellations among the flow loops and the remaining loops form a Hopf link, see Figure \ref{fig:Hopf-torus-intro}. 
This leads to an explicit formula \eqref{eq:gl1-formula-2-2p+1-knots} for the knot complement partition function. 
We further show how the formula leads to a generalized quiver form of the partition function \eqref{eq:ZMK-final-torus-knots}, compare \cite{Kucharski:2017ogk, Ekholm:2018eee}.
In Section \ref{sec:aug-curves}, we use the explicit formula to give new derivations of the augmentation polynomials of $(2,2p+1)$-torus knots and we also verify that the resulting $U(1)$-partition function satisfies (essentially, up to a certain shift) the same recursion relation as for the open topological string partition function for the knot conormal. 
This  gives a direct geometric interpretation to the wave function on the complement branch of the augmentation variety which is related to the $\hat Z$-invariant, \cite{Gukov:2017kmk, gukov2021two}, see Section~\ref{eq:relation-to-Zhat}.

 \subsubsection*{Organization of the paper}

Section \ref{sec:background} collects useful background material on skein-valued curve counting and open topological string theory. 
In Section \ref{sec:knot-complements} we prove Theorem \ref{t:skeinvalued knot complement} and provide key data for its application, namely an explicit description of flow loops on the complements of generic torus knots.
Section \ref{sec:evaluation} contains details on the explicit evaluation of knot complement partition functions, illustrating an application of Theorem \ref{t:skeinvalued knot complement} to the case of $(2,2p+1)$-torus knots in full detail.
In Section \ref{eq:quantum-branches} we discuss quantization schemes of augmentation varieties defined by curve counts for geometric branches.
Sections \ref{sec:unknot} and \ref{sec:trefoil} contain detailed discussions for the two simplest $T_{2,2p+1}$ knots, corresponding to the unknot ($p=0$) and the trefoil ($p=1$).
Finally, Section~\ref{eq:relation-to-Zhat} discusses the implications of our work for $\widehat Z$ invariants of knot complements.

\section{Background -- curve counts and topological strings}\label{sec:background}
In this section we briefly recall the relation between physical open topological string amplitudes and geometric counts of holomorphic curves. We start with geometry and then give the physical perspective. For this general discussion let $X$ denote a symplectic Calabi-Yau 3-fold and let $L\subset X$ be a Maslov zero Lagrangian. If $X$ is non-compact we assume that it has asymptotic contact boundary $Y=\partial X$, and if $L$ is non-compact we take it to have asymptotic Legendrian boundary $\Lambda=\partial L$, $\Lambda\subset Y$. 

\subsection{Skein valued curve counts}\label{ssec : skein valued counts}
From the geometric point of view we consider the moduli space $\mathcal{M}$ of holomorphic curves. A holomorphic curve is a map of a Riemann surface $S$ with boundary $\partial S$ into $(X,L)$, with everywhere $J$-complex linear differential for an almost complex structure $J$ on $X$:
\[
u\colon (S,\partial S)\to (X,L), \quad \bar\partial_J u = \tfrac12(du + J\circ du\circ j)=0.
\]
The moduli space $\mathcal{M}$ is obtained by identifying solutions that differ by pre-composition of holomorphic automorphisms of the domain. The Calabi-Yau and Maslov zero conditions imply that the formal dimension of $\mathcal{M}$ equals zero. If $L=\emptyset$ one defines the Gromov-Witten partition function
\begin{equation}\label{eq : closed GW count}
Z_{\mathrm{GW}}(X) = \sum_{(u,S)\in \mathcal{M}} w(u,S) g_s^{-\chi(S)} t^{[u(S)]},
\end{equation}
where $w(u,S)$ is the rational (orbifold) weight of $(u,S)$, $\chi(S)$ is the Euler characteristic of $S$, and $t^{[u(S)]}$ is the element represented by $u(S)$ in the group ring $\Q[H_2(X)]$. In order that the formal dimension of $\mathcal{M}$ equals its actual dimension, the Cauchy-Riemann equations must be perturbed coherently over the Deligne-Mumford space of all complex curves, compactified by adding nodal curves. Constructing such a perturbation scheme is an elaborate task and several different approaches have been developed, see e.g., \cite{hofer, fukaya, pardon}. To see that the count is independent of the choice of perturbations one observes that solution spaces over a generic 1-parameter family of perturbations give an oriented 1-dimensional cobordism between the zero-manifolds of the counts that constitute its boundary. The key observation here is that Gromov compactness for holomorphic curves implies that the parameterized moduli space is a 1-manifold up to bubbling where a node forms in a curve. The local model for nodal formation is the family of curves in $\C^2$ given by the equation $xy=\epsilon$, $\epsilon\in \C$. Here the node corresponds to $\epsilon=0$, hence has real codimension two in moduli and therefore does not appear in generic 1-parameter families.

Consider now instead the case of open curves. The moduli space of holomorphic curves  $\mathcal{M}$ still has formal dimension equal to zero and therefore the direct analogue of \eqref{eq : closed GW count} with relative homology replacing homology makes sense. However, its behavior under deformation is different in the open case. Here, except for the interior nodes discussed above, there are also boundary nodes. Boundary nodes are of two types: hyperbolic modeled on $xy=\epsilon$, and elliptic modeled on $x^2+y^2=\epsilon$, $\epsilon\in\R$. Note in particular that boundary nodes are codimension one phenomena and therefore cannot be avoided in generic 1-parameter families. In \cite{Ekholm:2019yqp} it was observed that the wall crossing phenomena can be identified with the HOMFLYPT skein relations in the skein module of the Lagrangian boundary condition $L$, $\Sk(L)$. This means that if curves are counted by the values of their boundaries in the skein module of $L$ the count remains invariant under deformations. In order that the wall crossing models apply, the perturbation needs to be small compared to local holomorphic area, and this is achieved by keeping area zero curves unperturbed. 
To still have an invariant count, zero area bubbles have to be excluded without the usual codimension two argument, see \cite{ekholm2022ghost}, for the relevant mechanism called ghost bubble censorship. 
The outcome is a deformation invariant count that is non-perturbative over area zero curves and that takes the form
\begin{equation}
    \mathsf{Z}_{(X,L)} = \sum_{(u,S)\in \mathcal{M}} w(u,S)\, t^{[u(S)]}\, z^{-\chi(S)}\, a^{\mathrm{lk}(u,L)}\,\langle u(\partial S)\rangle \ \in \ \Sk(L),
\end{equation}
where the skein module of $L$ is
\[
\Sk(L) = \Q[a^\pm,z^\pm]\{\text{framed links in $L$}\} / \text{skein relations},
\]
see \cite{Ekholm:2019yqp}.
In case $L=\emptyset$ the Gromov-Witten count is recovered as
\[
Z_{\mathrm{GW}}(X;g_s)=Z(X,\emptyset;z=e^{g_s/2}-e^{-g_s/2}),
\]
see \cite{ekholm2025countingbarecurves}.

The skein counting perspective also gives a geometric mechanism for conifold transition as follows. If $L\subset T^\ast S^3$ is a Lagrangian such that $L\cap S^3=\emptyset$ then we can consider the skein valued curve count $Z(T^\ast S^3, L\cup S^3)$.
We then deform the complex structure by so called SFT-stretching along the boundary of a tubular neighborhood of $S^3$ which is very small compared to the distance between $L$ and $S^3$. 
In general, in the limit of infinite SFT stretching, holomorphic curves converge to several level holomorphic buildings. 
These consist of holomorphic pieces in $\R\times \partial (T^\ast S^3)$, joined at Reeb orbits at positive and negative infinity, and a bottom level in $T^\ast S^3$ with Reeb orbits at only the positive end. 
In the case of $S^3$, any Reeb orbit at a negative end carries dimension $\le-2$, therefore there cannot be any bottom level and curves stretching between $L$ and $S^3$ leaves $S^3$, and eventually lie entirely in $\R\times\partial (T^\ast S^3)$. 
Thus, for sufficiently large stretching, there are then no curves with boundary components on $S^3$, and all contributions to $a$-powers in the invariant curve count come from the homology class of the fiber sphere. 
It can be shown \cite{Ekholm:2019yqp} that when a resolved conifold with $\C P^1$ of sufficiently small area is glued in at the negative end obtained by removing the zero section, the open curve count does not change. 
Then, by deformation invariance, it gives the curve count at any area. 
This demonstrates the relation between open curve counts in the conifold and HOMFLYPT polynomials predicted by Ooguri and Vafa \cite{Ooguri:1999bv}, see Section \ref{ssec:topstringandknots}.

\subsection{Skein-valued basic annuli}\label{ssec:annuli} 
In what follows it will be important to count curves of genus zero (and their all-genus multiple coverings) with boundaries on two different Lagrangians. The underlying curve is then an annulus and we will need the skein contributions of all its multiple covers. This contribution was determined in \cite{Ekholm:2021osm} by finding a skein valued recursion for the contribution and solving it. As mentioned there, the exact form of the partition function depends on choices of orientations and spin structures. Here it will be important to have the exact dependence.

We consider the model used in \cite{Ekholm:2021osm} and refer to there for details of the holomorphic curves discussed below. The ambient symplectic manifold is $(T^\ast S^1)\times\C^2$. The two Lagrangians are the zero section $L_0=S^1\times\R^2$ and small shift and rotation of $L_0$, $L_1=(S^1\times\{\xi\})\times e^{i\delta}\R^2$. Pick a spin structure and an orientation on $L_0$ and transport it to $L_1$ by the small rotation. The recursion relation is then determined by holomorphic curves of a $1$-dimensional Legendrian Hopf link in $\partial \C^2\approx S^3$. 
For both $L_0$ and $L_1$ there are two disks with only one positive puncture and one disk with one positive and two negative punctures, and we refer to the latter as triangles. 

The recursion relation is obtained by finding all curves with two positive punctures that can be attached to the triangles. The resulting skein recursion is, see \cite[Theorem 6.2]{Ekholm:2024ceb} 
\be\label{eq : recursion local model}
	\sfP^{(L_0)}_{1,0} - \sfP^{(L_0)}_{0,0} 
	+z\sum r(\mathsf{C}_k)\otimes \mathsf{C}_{-k} = 0,
\ee
where $\mathsf{C}_k$ is the curve in the solid torus of homology class $k$ and with $k-1$ positive crossings.  

We also consider other types of skein-valued recursion relations obtained by changes of orientations and of spin structures on $L_0$ or on $L_1$. Changing the spin structure on one of the Lagrangians alters the last term in the recursion by $z\sum r(\mathsf{C}_k)\otimes \mathsf{C}_{-k}\mapsto z\sum (-1)^k r(\mathsf{C}_k)\otimes \mathsf{C}_{-k}$, then changing the orientation of the second Lagrangian changes the last term by an overall sign and by $r(\mathsf{C}_k)\mapsto r(\mathsf{C}_{-k})$. This gives four recursions with the following solutions  
with values in the skein module of two disjoint solid tori $S^1\times D^2\sqcup S^1\times D^2$, $\Sk(S^1\times D^2)^{\otimes2}$:  
\be\label{eq:annulus partition function}
	\sfPsi_{\mathrm{ann}}^{(\sigma,\epsilon)}=\exp\left(\sigma\sum_{d\geq 1}  \frac{\epsilon^d}{d}\sfP_{d}\otimes \sfP_{-d}\right),
\ee
where $\sfP_{d}$ are power-sum generators of $\Sk(S^1\times D^2)^\pm$, in the positive or negative half of the HOMFLYPT skein of $S^1\times D^2$ for $d>0$ respectively $d<0$, see \cite{2014arXiv1410.0859M}, and where $\sigma,\epsilon \in \{-1,+1\}$ are signs.
The four types of annuli can be also expressed as follows, see  \cite{nakamura2024recursion}
\be\label{eq:annuli-explicit-skein}
\begin{array}{|c|c|c|} 
\hline
(\sigma,\epsilon) & \text{terminology} & \text{skein-valued partition function} \\
\hline\hline
    (1,+1) &\text{bosonic annulus}& \sum_\lambda  \sfW_{\lambda,\emptyset} \otimes \sfW_{\emptyset,\lambda}\\
    \hline
    (1,-1) &\text{fermionic annulus} & \sum_\lambda (-1)^{|\lambda|} \sfW_{\lambda,\emptyset} \otimes \sfW_{\emptyset, \lambda}\\
    \hline
    (-1,+1) &\text{bosonic anti-annulus} & \sum_\lambda (-1)^{|\lambda|} \sfW_{\lambda,\emptyset} \otimes \sfW_{\emptyset, \lambda^t} \\
    \hline
    (-1,-1) &\text{fermionic anti-annulus}& \sum_\lambda \sfW_{\lambda,\emptyset} \otimes \sfW_{\emptyset, \lambda^t} \\
    \hline
\end{array}
\ee
The solution to \eqref{eq : recursion local model} is the bosonic anti-annulus. 
Reversing orientation on one Lagrangian corresponds to changing $\sigma$, while changing spin structure corresponds to changing $\epsilon$.

In later sections, holomorphic annuli will be associated to flow loops, i.e., periodic solutions to vector fields dual to closed $1$-forms in a 3-manifold. To connect to the model annulus studied above, take $L_0$ as the manifold of the flow loop and $L_1$ as the graph of small multiple of the closed $1$-form in $T^\ast L_0$. The holomorphic annulus is then well approximated by the annulus which is obtained by shifting the flow loop in $L_0$ along the closed $1$-form. It has oriented boundary components along flow loop in both $L_0$ and its negative in $L_1$. 

Such flow loops and associated annuli were studied in \cite{dioagoekholm}. There are three distinct types of generic flow loops distinguished by the behavior of $1-\psi$, where $\psi$ is the linearization of the return map along the loop: elliptic, if both eigenvalues are either inside or outside the unit circle, and hyperbolic or negative hyperbolic otherwise. 
The difference between positive and negative hyperbolic is the sign of $\det(1-\psi)$ or, more geometrically, whether the return map does not (respectively does) permute the two stable and the two unstable directions, and to connect one to the other a change of trivialization corresponding to a change of spin structure is required.
For our purposes it will be essential to know which types of annuli correspond to which flow loops. 
To this end, note that the front of the model annulus Lagrangian shows that it corresponds to an elliptic flow loop. As mentioned above, a direct calculation of the recursion relation shows that the corresponding partition function is that of the bosonic anti-annulus, see \eqref{eq : recursion local model}. 
As shown in \cite{dioagoekholm}, reversing the orientation of a Lagrangian changes between elliptic and hyperbolic flow loops.
We thus have the following dictionary, where $o$ and $s$ denotes orientations and spin structures on the Lagrangians compared using the projection to the base.
\be\label{eq:flowloopsandannuli}
\begin{array}{|c||c|c|c|c|} 
\hline
\text{Flow loop} & (o,s)=(o,s) & (o,s)=(-o,s) & (o,s)=(o,-s) & (o,s)=(-o,-s)\\
\hline\hline
\text{elliptic} & \text{bosonic anti} & \text{bosonic} & \text{fermionic anti} & \text{fermionic}\\
\hline
\text{hyperbolic} & \text{bosonic} & \text{bosonic anti} & \text{fermionic} & \text{fermionic anti}\\
\hline
\text{negative hyperbolic} & \text{fermionic} & \text{fermionic anti} & \text{bosonic} & \text{bosonic anti}\\
\hline
\end{array}
\ee

\subsection{Topological string amplitudes and knot invariants}\label{ssec:topstringandknots}
Quantum knot invariants such as the colored HOMFLYPT polynomial coincide with expectation values of knotted Wilson lines in $SU(N)_k$ Chern-Simons theory on $S^3$ \cite{WittenJones},
and Chern-Simons theory provides an archetypal example of gauge/string duality, where the vacuum partition function on a 3-manifold $M$ and Wilson loop observables admit dual interpretations as $A$-model open topological string amplitudes in $T^*M$ \cite{Witten:1992fb}. 

This duality leads to the following: since the topological $A$-model is insensitive to variations of the complex structure, its amplitudes do not depend on the size of $M$, and in \cite{Gopakumar:1998ki} it was argued that the open $A$-model on $T^*S^3$ is dual to closed strings on the resolved conifold, by shrinking the $S^3$ to zero size and taking the small resolution of the resulting conifold singularity, with a central $\C P^1$ corresponding to the fiber $S^2$ linking $S^3$ before transition. The resulting geometry has a $B$-field flux $t = i N g_s = \frac{2\pi i\, N}{N+k}$ that coincides with the `t Hooft coupling of Chern-Simons theory on $S^3$. Duals of Wilson loop observables were found in \cite{Ooguri:1999bv}: the colored HOMFLYPT polynomials of links $K\subset S^3$ correspond to open string amplitudes in the resolved conifold $X$ of strings with boundary condition on the Lagrangian knot conormal $L_K$ with topology $S^1\times D^2$. 

The mathematical formulation of this latter duality was described in Section \ref{ssec : skein valued counts}. The connection to open string amplitudes is as follows. 
The open string partition function in the resolved conifold $X$ obtained by geometric transition (large $N$ limit) of a stack of $N$ branes on $S^3$, with $M$ branes on $L_K$ has the form
\begin{equation}\label{eq : open string part function}
\psi_{(X,L_K)} = \sum_{\lambda} c_\lambda(a,g_s) s_\lambda(x_1,\dots,x_M),  
\end{equation}
and is obtained from the skein valued partition function as follows.

The topology of $L_K$ is $S^1\times D^2$, and its HOMFLYPT skein, $\Sk(S^1\times D^2)$, has a basis $\{\sfW_{\lambda,\mu}\}$ labeled by pairs of partitions $\lambda$ and $\mu$ that are eigenvectors of the meridian operator \cite{2014arXiv1410.0859M}. 
Holomorphic curves ending on $L_K\subset X$ obtained by transition from $T^\ast S^3$ wind around $L_K$ only in the negative direction. Therefore their boundaries take values in the negative subspace $\Sk(S^1\times D^2)^-\subset\Sk(S^1\times D^2)$ generated by curves in non-positive homology class with respect to the positively oriented $S^1\times\{0\}$. A basis for the submodule $\Sk(S^1\times D^2)^-$ is given by $\{\overline \sfW_\lambda\}$, where $\overline  \sfW_\lambda:=\sfW_{\emptyset, \lambda^t}$.
Thus,
\be\label{eq:HOMFLYPT-generating-series}
	\mathsf{Z}_{L_K} = \sum_\lambda c_\lambda(a,z) \overline  \sfW_\lambda,
\ee
where it is a consequence of the skein proof of Ooguri-Vafa's conjecture that $c_\lambda(a,z)=H_\lambda(K)$, the  $\lambda$ colored HOMFLYPT polynomial of $K$ (with unreduced normalization). To get the open string partition function we change variables: $z=q-q^{-1}$ are related to open string parameters by $q=e^{g_s/2}$, $a=e^t = q^N$, $a_{L_K} = q^M$, and 
$\overline \sfW_{\lambda}=s_{\lambda^t}(x_1,\dots,x_M)$. Here the Schur functions $s_\lambda(x_1,\dots,x_M)$ give coordinates of boundary conditions for a stack of $M$ branes wrapping $L_K$.

\begin{remark}[Annuli from elliptic flow loops]\label{rmk:conormal-annuli} 
The coordinate conventions for solid tori we use here differ slightly from those used previously. We use coordinates adapted to identifications of nearby Lagrangians by small shifts. In previous works, the standard convention is that the annulus partition function of an annulus in the intersection $L_K\cap S^3$ is $\sum_\lambda \sfW_{\lambda,\emptyset}\otimes \sfW_{\lambda,\emptyset}$. 
As explained in Section \ref{ssec:annuli}, in adapted coordinates we get instead the `bosonic anti-annulus' 
\be\label{eq:twisted-anti-annulus}
	\sfPsi^{(-1,+1)}_{\mathrm{ann}}
	= \sum_\lambda (-1)^{|\lambda|} \sfW^{(S^3)}_{\lambda,\emptyset}\otimes \sfW^{(L_K)}_{\emptyset,\lambda^t} \,,
\ee
and we see that the standard convention is obtained from adapted coordinates by reversing the orientation of the central $S^1$ and reversing the orientation of $L_K$.  
\end{remark}

\subsection{Augmentation varieties as $A$-brane moduli spaces}
In this section, we consider the augmentation variety and its relation to physical moduli spaces of $A$-branes.
Certain asymptotic regions of the augmentation variety corresponding to `large volume' branches in open string moduli are understood to encode counts of holomorphic disks with boundaries on a Lagrangian submanifold $L$ of $T^*S^3$ (such interpretations are not generally known for all asymptotic regions). Here we will discuss compatibility of quantizations of distinct branches.

We recall the definition of the augmentation variety.
The Legendrian conormal $\Lambda_K\subset ST^\ast S^3$ of a knot $K\subset S^3$ is the ideal boundary of its Lagrangian conormal $L_K\subset T^\ast S^3$. Knot contact homology is the Chekanov-Eliashberg dg-algebra $CE^\ast(\Lambda_K)$, generated by Reeb chords and with differential that counts $\R$-families of punctured holomorphic disks in $\R\times ST^\ast S^3$ with boundary on $\R\times\Lambda_K$ \cite{2011arXiv1109.1542E}. 
The Chekanov Eliashberg dg-algebra has coefficients in $\C[x^\pm,y^\pm]$, identified with the group ring of the first homology of $\Lambda_K$. The augmentation variety of $K$ is the representation variety of $CE^\ast(\Lambda_K)$. It is cut out by the augmentation polynomial, $A_K(x,y)=0$. If $L$ is a Lagrangian filling of $\Lambda_K$ in the resolved conifold then the disk potential of $L$ gives a local branch of the augmentation curve, e.g., the conormal $L_K$ topologically fills the $y$-cycle and if $W_K(x)$ is the disk potential of $L_K$ then
\[
A_K(x, e^{\frac{\partial W_K}{\partial \log x}})=0.
\]
The disk potential is the semi-classical limit of the $U(1)$-skein valued partition function and from this perspective one expects that there is an operator valued lift $\hat A_K(\hat x,\hat y)$ that annihilates the $U(1)$-partition function 
\[
Z_{L_K}^{U(1)} = \sum_{d} H_{(d)}(K;(a,q)) x^d,
\]
where the coefficients are the symmetrically colored HOMFLYPT polynomials of $K$. It was shown in \cite{Garoufalidis:2016zhf} that such an operator polynomial exists for every knot.

From the physical perspective it is natural to think of the augmentation variety as a moduli space of $A$-branes as follows.  
In the limit of large volume for open string moduli, i.e. in a situation in which all open string instantons have a large area,
Lagrangian $A$-branes can be modeled by local systems over Lagrangian submanifolds \cite{Witten:1992fb}.
$A$-branes over conormal Lagrangians $L_K$ belong to this class. 
When the local system is Abelian (rank 1), such branes come in complex 1-parameter families, locally controlled by a real-valued modulus for $L_K$ and a periodic modulus for the local system (recall $b_1(L_K)=1$).

Varying the geometric modulus changes the area of curves, while varying the periodic modulus changes their boundary holonomies. Small perturbations within a large volume region do not change the overall partition function due to topological invariance under deformations. However finite changes in brane moduli can lead to probing regions of moduli space where worldsheet instantons have small area, and where the description of branes in terms of local systems over Lagrangian cycles is less accurate.

\begin{figure}[h!]
\begin{center}
\includegraphics[width=0.6\textwidth]{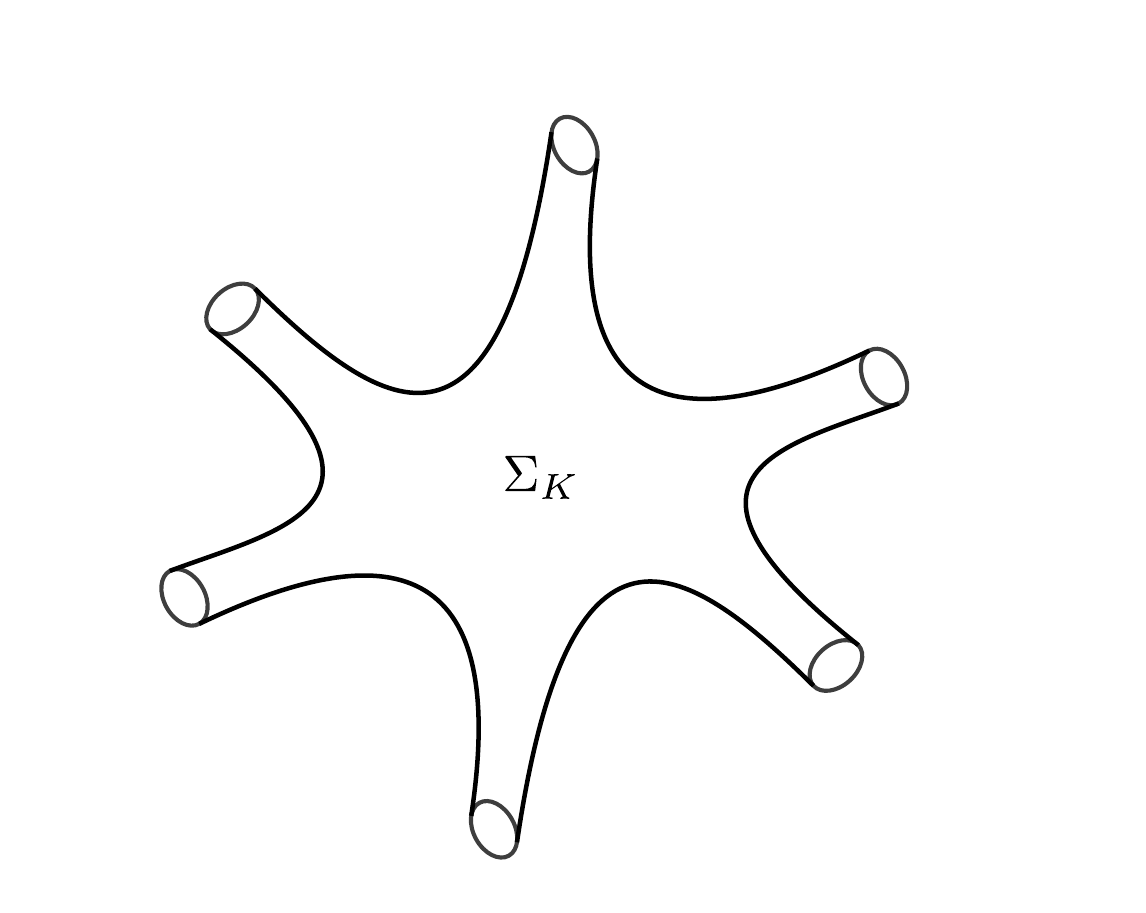}
\caption{Augmentation curves feature several large volume regions in correspondence of punctures.}
\label{fig:aug-curve}
\end{center}
\end{figure}

A global description of the moduli space of rank-1 conormal $A$-branes was proposed in \cite{Aganagic:2013jpa}, where it was identified with the augmentation variety of knot contact homology \cite{2011arXiv1109.1542E}. 
The augmentation curve has several asymptotic regions corresponding to punctures in $\IC^*\times \IC^*$, see Figure \ref{fig:aug-curve}. 
In suitably defined open string moduli, as defined below in Section \ref{sec:local-moduli}, each of these asymptotic regions corresponds to a large volume phase.
While it is not known whether every asymptotic region admits an interpretation as a patch of moduli space for a geometric $A$-brane (i.e. some Lagrangian $L_\alpha$ with a local system), it is nevertheless possible to define the formal counterpart of curve counts associated with each asymptotic region \cite{Ekholm:2021irc, Ekholm:2024ceb}.

For rank-1 local systems, all-genus curve counts in a large volume phase labeled by $\alpha$ define a certain quantization of $A_K$ to a $q$-difference operator $\hat A_{K,\alpha} \in \IZ[\hat x,\hat y,a]$, characterized by the open string partition function $\psi_\alpha$ via
\be\label{eq:branch-curve-counts}
	\hat A_{K,\alpha} \cdot \psi_\alpha=0\,.
\ee
The quantization of $A_K$ may depend on a choice of branch $\alpha$ as shown in \cite{Ekholm:2024ceb}. The dependence is mild and correspond to changes in auxiliary geometric data. We will see counterparts of this below when we study the branch of the augmentaion variety corresponding to the knot complement.

\section{The skein valued partition function on the complement of a fibered knot}\label{sec:knot-complements}

In this section we first prove Theorem \ref{t:skeinvalued knot complement} which expresses the skein valued partition function of the complement of a fibered knot in terms of flow loops. Then we determine flow loops for torus knot complements.

\subsection{Fibered knots and flow loops}\label{ssec:fibered and flow loops}
Let $K\subset S^3$ be a fibered knot. This means that that there is a fibration $\pi\colon S^3\setminus K\to S^1$. Let $\theta$ be a coordinate on $S^1$ and let $\alpha=\pi^\ast d\theta$, then $\alpha$ is a nowhere vanishing $1$-form on $S^3\setminus K$. After a small perturbation $\alpha$ has a discrete set of flow loops that is finite below any given action. Furthermore, the flow loops are generic in the following sense. Recall the correspondence between flow loops and holomorphic annuli, see \ref{eq:flowloopsandannuli}.

\begin{proof}[Proof of Theorem \ref{t:skeinvalued knot complement}]
    Let $\alpha$ denote the generic 1-form used to shift the complement of $K\subset S^3$ in $T^\ast S^3$ and consider the scaled 1-form $\epsilon\alpha$. Consider holomorphic curves in $T^\ast S^3$ with boundary in $S^3\cup M_K$ with boundary $\gamma$ in $M_K$ for which $\int_\gamma\alpha <\mathfrak{a}$ for some $\mathfrak{a}>0$. 

    We use the correspondence between flow graphs and holomorphic curves. First, \cite[Lemma A.3]{Ekholm:2025anq} shows that any sequence of holomorphic curves with action $\le \epsilon\mathfrak{a}$ converges to a flow graph. Second, since $M_K\subset T^\ast S^3$ is graphical over $S^3$ (and has only one sheet), the only finite energy flow graphs are flow loops. Consequently, all holomorphic curves are multiple covers of the basic annuli. Furthermore, by \cite[Lemma 3.5]{dioagoekholm} transverse flow loops correspond to transverse annuli. 
The contribution from all multiple covers of a rigid annulus in the skein of the solid tori that are tubular neighborhoods of its boundaries was computed in \cite{Ekholm:2021osm}, see also \cite[Section 6.2]{Ekholm:2024ceb}. The theorem follows.      
\end{proof}

Since the skein module of $S^3$ equals $\Q[a^\pm,q^{\pm}]$ we can replace the link components in $S^3$ in Theorem \ref{t:skeinvalued knot complement} by polynomials. This gives the following expression for the partition function of skein-valued curve counts for $M_K$. Let $\sfW_{\lambda}^{(\gamma)}\in \Sk(M_K)$ denote the insertion of $\sfW_\lambda \in \Sk(S^1\times D^2)$ in a neighbourhood of the flow loop $\gamma\subset M_K$, and let $\sfW_{\lambda_1}^{(\gamma_1)} \dots \sfW^{(\gamma_m)}_{\lambda_m}$ denote the simultaneous insertions along all flow loops.
The skein-valued partition function is given by the following explicit formula
\be\label{eq:localization-formula-final-skein}
\begin{split}
	\sfPsi_{(T^\ast S^3,M_K)}(a,q,\mu) 
	& =\prod_{\gamma\in\Gamma} \iota_{\partial A(\gamma)}(\sfPsi^{(\sigma(\gamma),\epsilon(\gamma))}_{\mathrm{ann}})
	\\
	&= \sum_{\lambda_1,\dots, \lambda_m}
	\prod_{j=1}^{m}
	(\epsilon_j \, \sigma_j)^{|\lambda_j|}
	\,
	H_{\lambda^\vee_1,\dots, \lambda_m^{\vee}}(a,q)
	\,
	\sfW_{\lambda_1}^{(\gamma_1)} \dots \sfW^{(\gamma_m)}_{\lambda_m}
\end{split}
\ee
where $\lambda_j^\vee$ is a partition determined by the type of the $j$-th annulus, which in turn is determined by the type of flow loop.

\subsection{Flow loops on the complement of torus knots}\label{sec:torus-knots-general}
Let $(p,q)$ be relatively prime integers and consider a $T_{p,q}$ torus knot. There is a well known Seifert fibration of the complement $S^3\setminus T_{p,q}$ that we will use to find flow loops. 
To see the Seifert fibration consider the Hopf-fibration $\pi\colon S^3\to S^2$ and the singular foliation of $S^2$ by latitude circles that degenerate over the poles. The preimage of the latitude foliation is a foliation of $S^3$ by tori with two 1-dimensional leaves that form a Hopf link, the preimage of the poles. Each torus is foliated by $T_{p,q}$ torus knots with leaf space $S^1$ and the fibers limit to a curve $p$-times around one of the Hopf link components and $q$-times around the other. Removing a neighborhood of one $T_{p,q}$-fiber we find that the complement fibers over a disk with two singular fibers. 

\begin{figure}[h!]
\begin{center}
\includegraphics[width=0.35\textwidth]{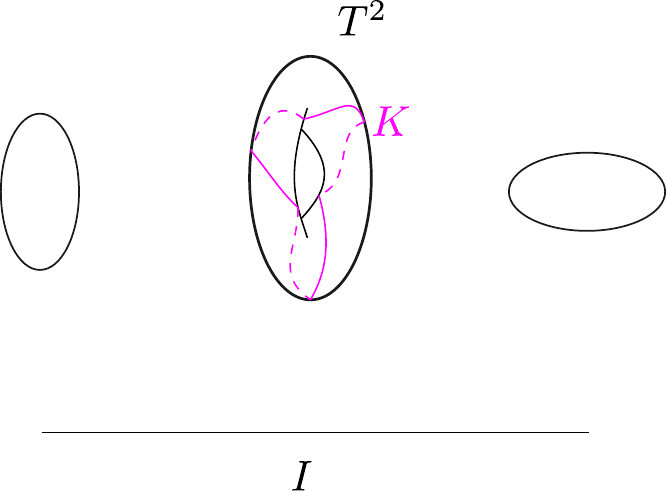}
\caption{Fibration of the 3-sphere by 2-tori, with a torus knot.}
\label{fig:S3-fibration}
\end{center}
\end{figure}

\begin{proposition}\label{prp:generaltorusknot}
Let $K=T_{p,q}$.
    There is a nowhere zero 1-form on the complement $M_K=S^3\setminus K$ of a $(p,q)$-torus knot with exactly three basic flow loops, the two braid axes of the torus knot and the torus knot itself. For the pull-back of the spin structure and orientation of $S^3$ to $M_{K}$ 
    the flow loops of the axes are both elliptic and support bosonic anti-annuli, while the torus knot flow loop is hyperbolic and supports a bosonic annulus, see \eqref{eq:flowloopsandannuli}. 
\end{proposition}

\begin{proof}
Perturb the 1-form along the fibers of the fibration $M_K\to D^2$ by a Morse function of the base that has two minima at the two singular fibers, one saddle point at a general fiber, and with outward gradient along the boundary. The critical points correspond to the flow loops. The characters of the flow loops are clear.
\end{proof}

\begin{figure*}[h!]
    \centering
    \begin{subfigure}[t]{0.5\textwidth}
        \centering
        \includegraphics[width=.8\textwidth]{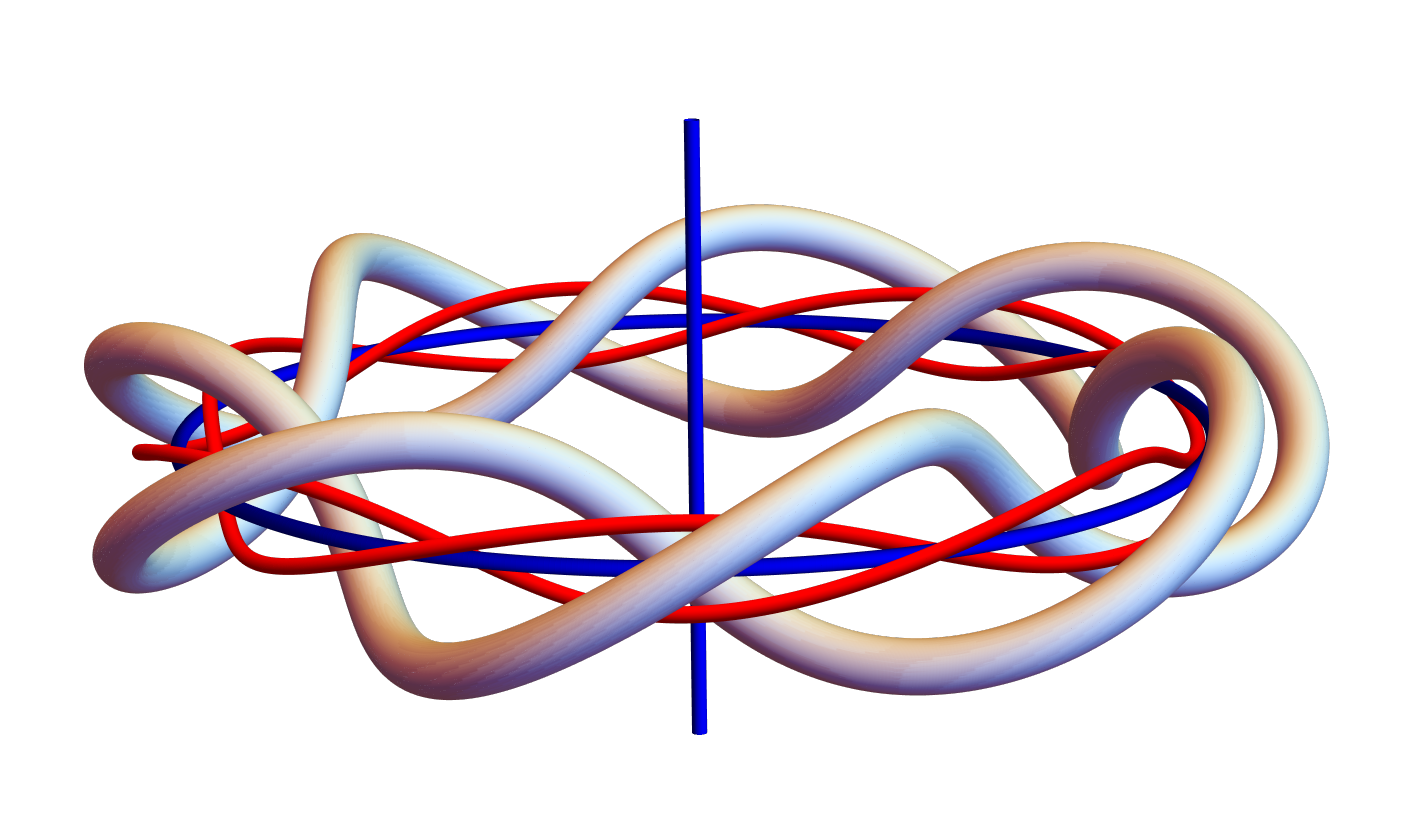}
        \caption{}
        \label{fig:T27-3loops}
    \end{subfigure}%
    ~ 
    \begin{subfigure}[t]{0.5\textwidth}
        \centering
        \includegraphics[width=.8\textwidth]{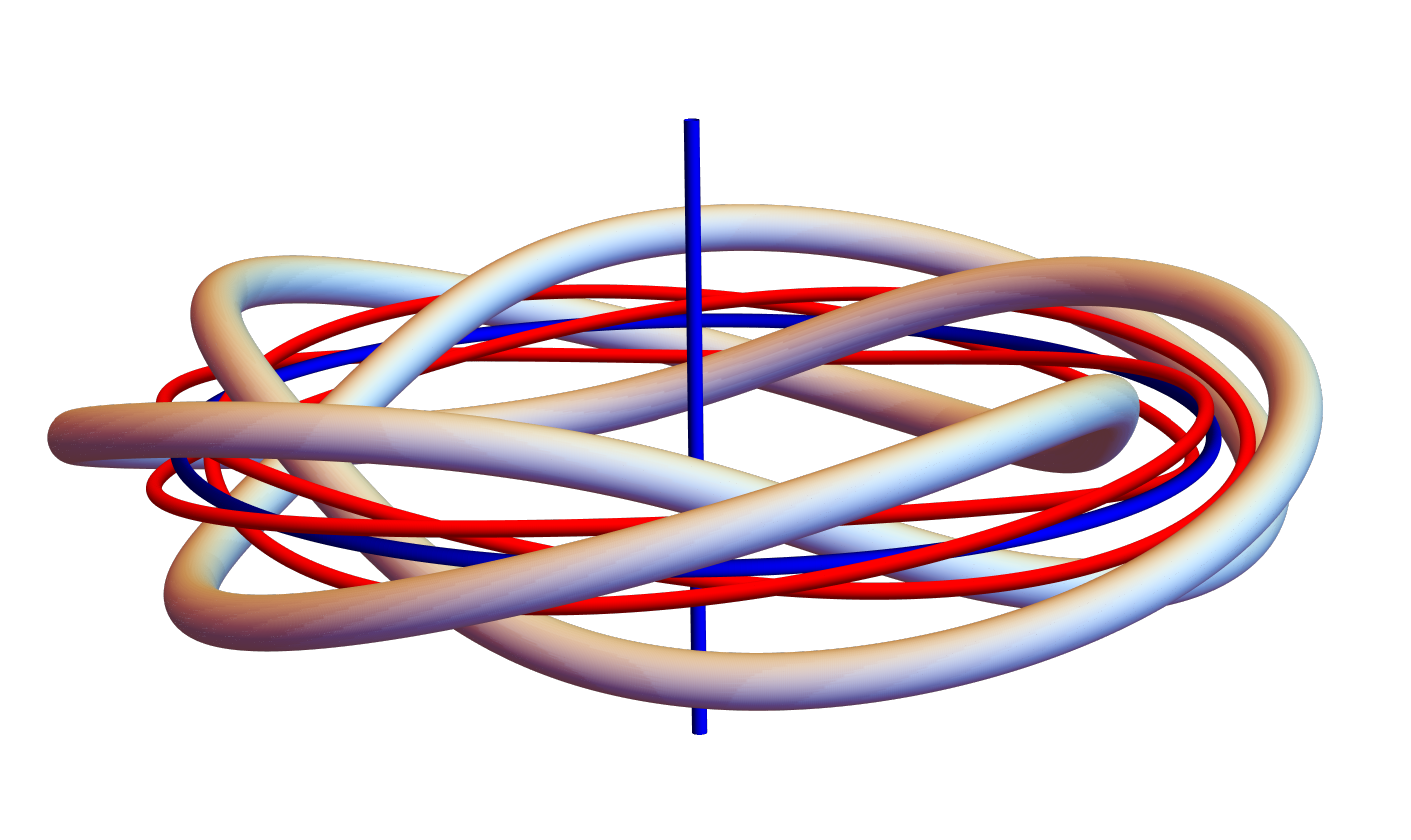}
        \caption{}
        \label{fig:T35-3loops}
    \end{subfigure}
    \caption{Examples of torus knot complements, respectively $T_{2,7}$ and $T_{3,5}$ and corresponding Morse flow loops.
    Blue loops are elliptic critical points (minima) of the Morse function, red loops are hyperbolic (saddles). 
    }
    \label{fig:3-flow-loops-torus-knots}
\end{figure*}

\subsection{Flow loops on the complement of $(2,2p+1)$-torus knots}\label{sec:torus-knots}
For $(2,2p+1)$-torus knots the flow loop structure can be further simplified as follows. 

\begin{proposition}\label{prp:2,2p+1torusknot}
    Let $K=T_{(2,2p+1)}$. Then $M_K$ admits a 1-form with two basic flow loops, an elliptic flow loop of degree $2$ and a negative hyperbolic flow loop of degree $2p+1$. 
    The corresponding annuli are: a bosonic anti-annulus for the degree $2$ loop and a fermionic annulus for the degree $(2p+1)$ loop, forming a negative Hopf link. 
\end{proposition}

\begin{proof}
We use a perturbing Morse function as in Proposition \ref{prp:generaltorusknot}. The Seifert surface $\Sigma$ is the branched degree $2(2p+1)$-cover $\pi\colon \Sigma\to D$ of the disk $D$ with two branch points $\zeta_1$ of ramification index $2$ and $\zeta_2$ of index $2p+1$. The former lifts on $\Sigma$ to $2p+1$ intersections of the Seifert surface with the horizontal elliptic flow loop in Figure \ref{fig:3-flow-loops-torus-knots}, while the latter lifts to $2$ intersections with the vertical elliptic flow loop. 
The Morse function on $\Sigma$ is pulled back from a Morse function on $D$ with two minima at $\zeta_1$ and $\zeta_2$, and a saddle point at a regular point $\zeta$. 
The pre-images of the critical points (i.e., the critical points of the pull back function) will be denoted as follows 
\be
\begin{split}
\pi^{-1}(\zeta_1) \ &= \ \{\zeta_1^{(1)},\zeta_1^{(2)},\dots,\zeta_1^{(2p)},\zeta_1^{(2p+1)}\},\\
\pi^{-1}(\zeta_2) \ &= \ \{\zeta_2^{(1)},\zeta_2^{(2)}\},\\
\pi^{-1}(\zeta) \ &= \ \{\zeta^{(1)},\zeta^{(2)},\dots,\zeta^{(2p+1)},\zeta^{(2p+2)},\dots,\zeta^{(4p+2)}\}.
\end{split}
\ee

Let us move the saddle $\zeta$ very close to $\zeta_1$. 
Then the pre-image of a small disk $\Delta$ around $\zeta_1$ consists of $2p+1$ small disks $\Delta^{(k)}\subset\Sigma$, $k=1,\dots,2p+1$, where we choose notation so that
\be\label{eq:crit-point-grouping}
	\{\zeta_1^{(k)},\zeta^{(k)},\zeta^{(k+2p+1)}\}\subset \Delta^{(k)},
\ee
so that the return map takes $\Delta^{(k)}$ to $\Delta^{(k+1)}$, where $1$ is identified with $2p+2$, 
and so that the unstable manifold of $\zeta^{(k)}$ has one flow line that connects with $\zeta_1^{(k)}$ and one that connects with $\zeta_2^{(1)}$, the unstable manifold of $\zeta^{(k+2p+1)}$ also has one flow lines that connects with $\zeta_1^{(k)}$ but the other connects with $\zeta_2^{(2)}$. Note that, with this notation, the return map on $\pi^{-1}(\zeta)$ is 

\be
        \zeta^{(k)}\mapsto
        \begin{cases}
                \zeta^{(k+1+2p+1)}, &\text{ if } 1\le k < 2p+1,\\
                \zeta^{(2p+2)}, &\text{ if } k =2p+1,\\
                \zeta^{(k+1-(2p+1))} &\text{ if } 2p+2\le k < 4p+2,\\
                \zeta^{(1)} , &\text{ if } k =4p+2.
        \end{cases}
\ee

Note that the component of the flow along $\Sigma$ across the boundary $\partial \Delta^{(k)}$ has four general position tangencies. Hence we can redefine the flow inside $\Delta^{(k)}$ to be a single saddle $\zeta^{(k)}$ at the center with unstable manifold components connecting to $\zeta^{(1)}_2$ and $\zeta_2^{(2)}$. The return map then takes $\zeta^{(k)}$ to $\zeta^{(k+1)}$ and the flow line $\zeta^{(k)}\to\zeta_2^{(1)}$ (respectively $\zeta^{(k)}\to\zeta_2^{(2)}$) to the flow line $\zeta^{(k+1)}\to\zeta_2^{(2)}$ (respectively to $\zeta^{(k+1)}\to\zeta_2^{(1)}$). 

This then gives a flow with only two flow loops. 
One of these is the original degree-two elliptic flow loop with associated bosonic anti-annulus, corresponding to $\zeta_2$.
The other is a degree $2p+1$ loop corresponding to a saddle that remains at $\zeta$. Since the return map permutes its stable (or unstable) manifolds, this is a negative hyperbolic flow loop.
The corresponding annulus is therefore a fermionic annulus, according to~\eqref{eq:flowloopsandannuli}. 
\end{proof}

An example of the flow loops of Proposition \ref{prp:2,2p+1torusknot} for $T_{2,7}$ is shown in Figure \ref{fig:torus-link-flow-loops}.

\begin{remark}[Coalescence phenomena for flow loops]\label{r : general 2+1=1} 
    The proof of Proposition \ref{prp:2,2p+1torusknot} shows that in general, an elliptic flow loop with a hyperbolic flow loop going twice around it, as in the proof, are deformation equivalent to a single negative hyperbolic flow loop. Translated to holomorphic annuli this means the following. If the orientations and spin structures agree we see that a bosonic anti-annulus with a bosonic annulus going twice around it are equivalent to a single fermionic annulus going once around. 
    Changing orientation and spin structures we get similar relations for other annuli. 
    Reversing orientation gives that a bosonic annulus together with a twice-around bosonic anti-annulus are equivalent to a single fermionic anti-annulus.
    Changing spin structure gives that a fermionic anti-annulus with a twice-around fermionic annulus are equivalent to a single bosonic annulus.
    Changing both orientation and spin structure gives that a fermionic annulus together with a twice-around fermionic anti-annulus are equivalent to a single bosonic anti-annulus.
\end{remark}

\begin{figure}[h!]
\begin{center}
\includegraphics[width=0.45\textwidth]{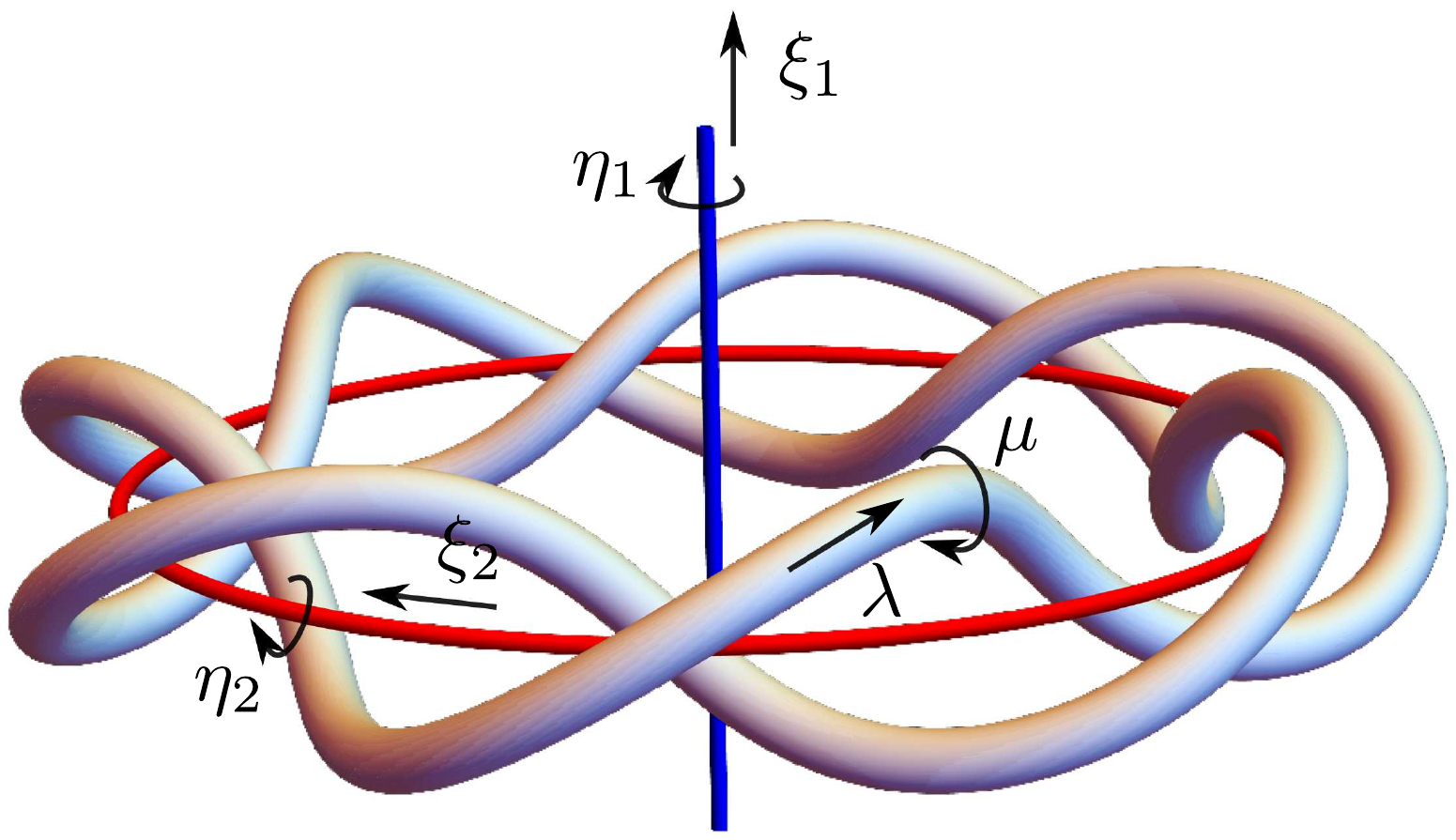}
\caption{The $T_{2,7}$ torus knot with loop flow lines for $\alpha$ forming a negative Hopf link, after the simplification described in Proposition \ref{prp:2,2p+1torusknot}. 
The blue line corresponds to the elliptic flow loop while the red line corresponds to the hyperbolic flow loop.
Note that orientation of $M_K$ here agrees with that of $S^3$. In particular the orientation of $(\lambda,\mu)$ on the boundary torus gives a normal vector that points outward from $M_K$. 
}
\label{fig:torus-link-flow-loops}
\end{center}
\end{figure}

\begin{remark}
One can use the flow loop formula for the Alexander polynomial, see e.g. \cite{dioagoekholm}, to see that $(2,2p+1)$-torus knots are the only torus knots that allows for only two flow loops on their complement. 

Recall that the  Alexander polynomial of $T_{q,p}$ torus knots is given by (where we normalize by dividing by an overall $(t-1)$ factor). 
\be
	\Delta(t) = \frac{(t^{pq}-1)}{(t^p-1)(t^q-1)},
\ee
where the factors in the denominator correspond to the two {bosonic annuli} of degree $p$ and $q$ and that in the numerator is the {bosonic anti-annulus} of degree $pq$. 
Conventions used here involve a change of orientation on one of the Lagrangians compared to our previous discussion, which swaps annuli and anti-annuli according to \eqref{eq:flowloopsandannuli}.

For $(p,q)=(2,2p+1)$, the polynomial simplifies to 
\be\label{eq:Alexander-reduced-annuli}
	\Delta(t) = \frac{(t^{2(2p+1)}-1)}{(t^2-1)(t^{2p+1}-1)} = \frac{(t^{(2p+1)}+1)}{(t^2-1)}\,,
\ee
with contributions from an elliptic flow loop of degree $2$ and a negative hyperbolic flow loop of degree $2p+1$, as shown in Proposition \ref{prp:2,2p+1torusknot}.

Analogous simplifications for other $(p,q)$ torus knots can be found. However since none of these has only two factors, it follows that there must be more than two flow loops in other cases.
\end{remark}

\section{Partition functions in the $\mathfrak{gl}_1$-skein }\label{sec:evaluation}
In this section we specialize the results of Section \ref{sec:knot-complements} to the $\mathfrak{gl}_1$-skein and obtain explicit formulas for torus knot complements. In the case of $(2,2p+1)$-torus knot we furthermore obtain a quiver-like expansion of the partition function and use the geometry of flow loops to derive the augmentation polynomial in a new way.

\subsection{The complement $\mathfrak{gl}_1$-skein partition function for fibered knots}\label{sec:cobordism}
In this section we prove Corollary \ref{c:U(1) knot complement}. Before going into the actual proof we review general properties of the $\mathfrak{gl}_1$-skein of knot complements.

Evaluating \eqref{eq:localization-formula-skein} requires two further steps. 
First, we identify $(S^1\times D^2)_j$ with neighborhoods $N_j$ of $\gamma_j$ in $M_K$. The complement of these neighborhoods then gives a (symplectic) cobordism $\CL_\alpha\approx M_K\setminus (\cup_jN_j)$ with 
\be\label{eq:cobordism-L-MK}
	\partial_+\CL_\alpha=\Lambda_K,\quad \partial_-\CL_\alpha=\cup_j\partial N_j\,,
\ee
comparing Figures \ref{fig:MK-pushoff} and \ref{fig:link-conormal-pushoff}.
Second, for general knots $K$ it is an elaborate task to find a presentation of the HOMFLYPT skein-module of $M_K$ and to get to something directly computable we specialize to the $\mathfrak{gl}_1$ skein. 
Here we need to compute the $\mathfrak{gl}_1$-skein operator $\hat\CO$ induced by the cobordism $\CL_\alpha$.

 \begin{figure*}[h!]
     \centering
     \begin{subfigure}[t]{0.5\textwidth}
\begin{center}
\includegraphics[width=\textwidth]{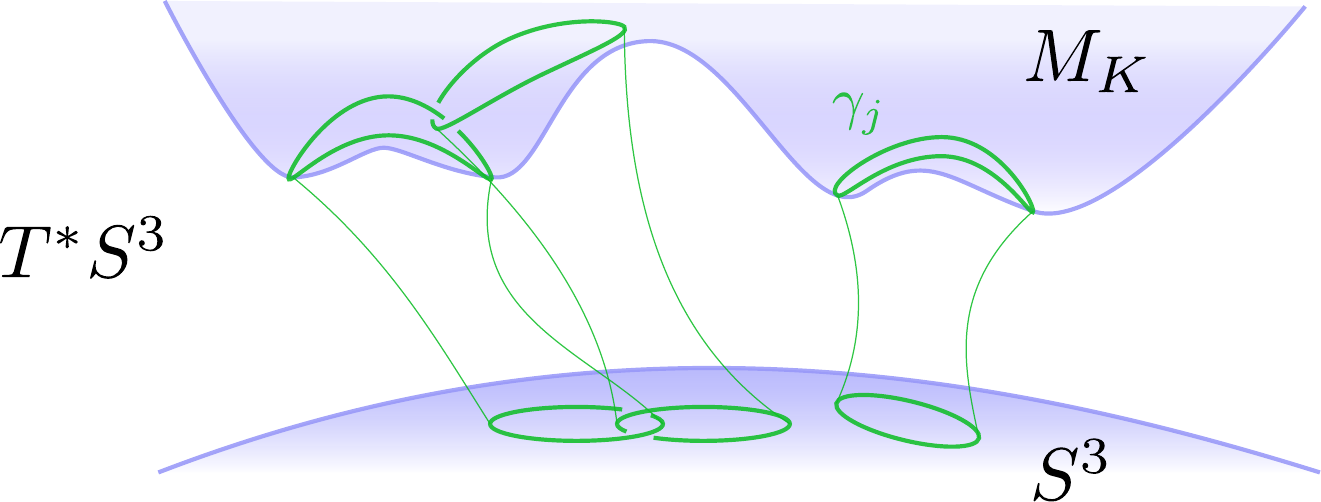}
\caption{}
\label{fig:MK-pushoff}
\end{center}
     \end{subfigure}%
     ~ 
     \begin{subfigure}[t]{0.5\textwidth}
\begin{center}
\includegraphics[width=\textwidth]{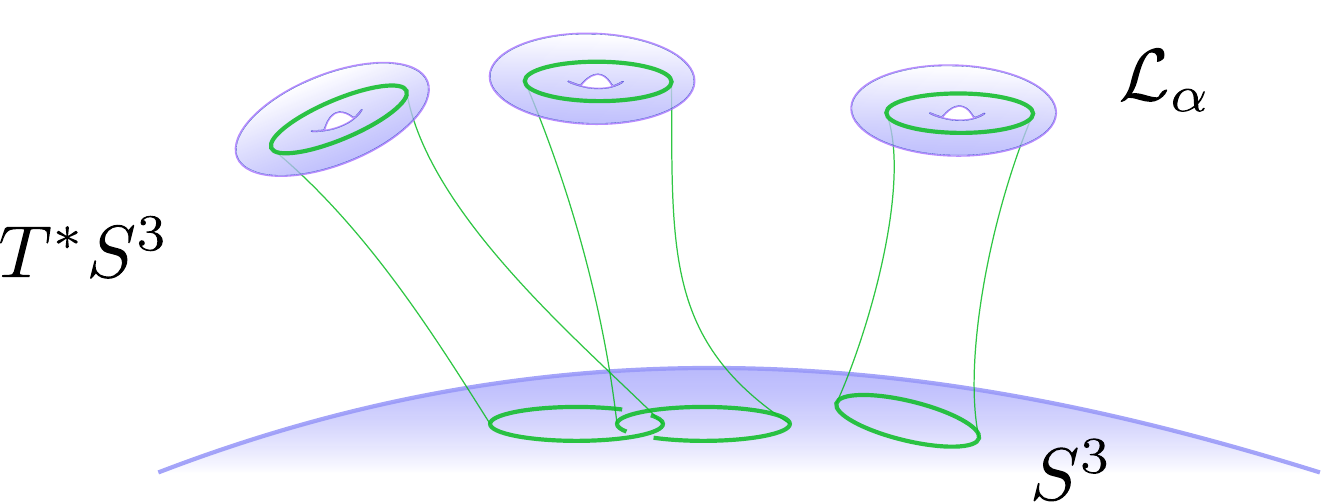}
\caption{}
\label{fig:link-conormal-pushoff}
\end{center}
     \end{subfigure}
     \caption{Left: Pushoff of $M_K$ with holomorphic annuli stretching to the zero section. Right: Pushoff of the link conormal for $\gamma_1,\dots,\gamma_k$ with holomorphic annuli stretching to the zero section.}
     \label{fig:MK-S3-L}
 \end{figure*}

\subsubsection{The $\mathfrak{gl}_1$ skein of $M_K$}\label{ssec:gl1skein}
The $\mathfrak{gl}_1$-skein of $M_K$ encodes `homology and linking' and is given by
\be
	\Sk_{\mathfrak{gl}_1}(M_K) \simeq \Q[q^\pm]\otimes H_1(M_K;\Q)\approx \Q[q^\pm]\otimes \Q[\mu^\pm]\,,
\ee
where $\mu$ is the homology generator corresponding to a 
meridian cycle linking the knot $K$, see figure~\ref{fig:skein-MK}, see Corollary \ref{c:U(1) knot complement} for the definition of linking.
The $\mathfrak{gl}_1$-skein algebra associated with $\Lambda_K=\partial M_K\approx T^2$ is a quantum torus algebra generated by $\mu^\pm$ and $\lambda^\pm$ with commutation relation
\be
	\hat \lambda \hat\mu = q^{-2} \hat \mu\hat\lambda\,,
\ee
where $\lambda$ is the longitude of the knot.
These operators correspond to the quantization of the adapted coordinates for the complement branch
\be
	\mu := x_{m} = y\,,
	\qquad
	\lambda := y_{m} = x^{-1}\,,
\ee 
defined earlier in Seifert framing. 
\begin{figure}[h!]
\begin{center}
\includegraphics[width=0.3\textwidth]{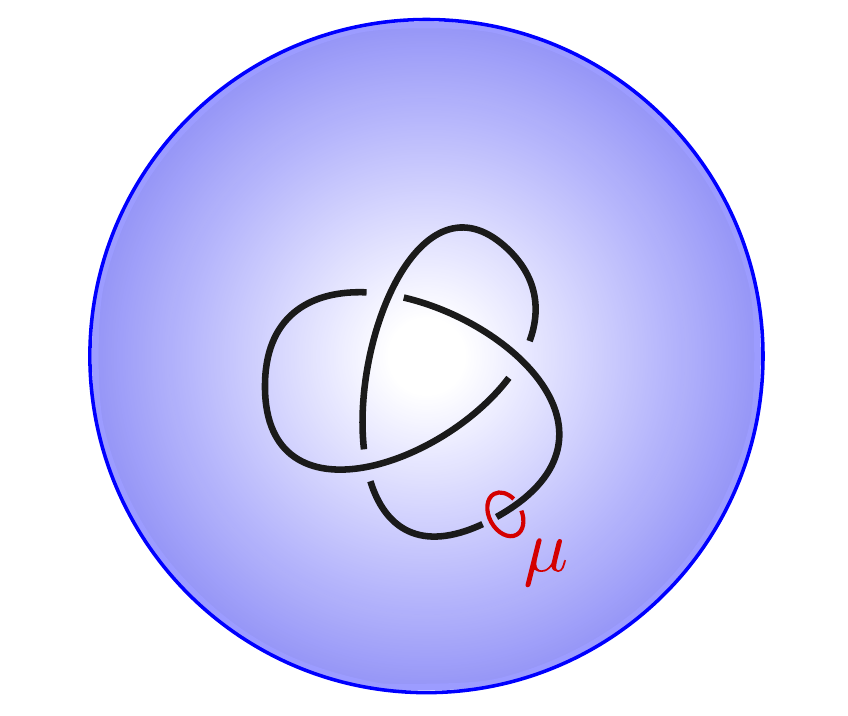}
\caption{The generator of $\Sk_{\mathfrak{gl}_1}(M_K)$ for $K=3_1$.}
\label{fig:skein-MK}
\end{center}
\end{figure}

If $N\subset M_K$ is a solid torus and $\gamma\subset M_K$ is a link then, since skein relations are local, the evaluation of $\gamma$ in the $\mathfrak{gl}_1$-skein of $M_K$ factors through the $\mathfrak{gl}_1$-skein of $N$. 

For the standard basis $\sfW_\lambda$ of the skein of $N$, $\sfW_\lambda$ maps to $0$ in the $\mathfrak{gl}_1$-skein except for symmetric partitions $\lambda=(d)$, in which case it maps to the corresponding power of the homology generator $\xi$
\be
	\sfW_\lambda \ \mapsto\ \left\{\begin{array}{lr}
	\xi^{d} \qquad & \lambda = (d) \\
	0 \qquad &\text{otherwise}
	\end{array}\right.
\ee
The contribution from basic annuli stretching between two solid tori $N_1$ and $N_2$ with homology generators $W_{\ydiagram{1}}^{(k)} = \xi_k$ for $k=1,2$, are then given by the following specialization of \eqref{eq:annuli-explicit-skein} to the $\mathfrak{gl}_1$ skein
\be\label{eq:annuli-explicit}
\begin{array}{|c|c|c|} 
\hline
(\sigma,\epsilon) & \text{terminology} & \text{$\mathfrak{gl}_1$ partition function $\psi^{(\sigma,\epsilon)}_{\text{ann}}$} \\
\hline\hline
    (1,+1) &\text{bosonic annulus} &  \frac{1}{1-\xi_1 \, \xi_2^{-1}}\\
    \hline
    (1,-1) &\text{fermionic annulus}& \frac{1}{1+\xi_1 \, \xi_2^{-1}}\\
    \hline
    (-1,+1) &\text{bosonic anti-annulus}& 1-\xi_1 \, \xi_2^{-1}\\
    \hline
    (-1,-1) &\text{fermionic anti-annulus} & 1+\xi_1 \, \xi_2^{-1}\\
    \hline
\end{array}
\ee

\subsubsection{The cobordism}\label{ssec:cobordism} 
We next consider inclusions of several solid tori $N_j\subset M_K$ and $T_j\subset S^3$ with basic annuli of types $(\sigma_{j},\epsilon_{j})$ stretched between them, ending on flow loop $\gamma_j$ at both ends. 
The boundary components in $T_j$ give HOMFLYPT polynomials. 
The boundary components in $M_K$ can be evaluated in terms of the generator $\mu$ of the $\mathfrak{gl}_1$-skein:  
\be
    \sfW^{(\gamma_j)}_{(d_j)} \ \mapsto \ (q^{\beta_j}\mu^{m_j})^{d_j}\,,
\ee
where $(\beta_j,m_j)\in\IZ\times  \IZ$ are determined by the framing and homology class of the curve $\sfW^{(\gamma_j)}_{(1)}$ in $M_K$.

When there are several flow loops $\gamma_j\subset N_j$ included there is also mutual linking between the Wilson lines $\gamma_j$. These are determined by the topology of the cobordism $\CL_\alpha$ in \eqref{eq:cobordism-L-MK} and the effect in the $\mathfrak{gl}_1$-partition is obtained by twisting the ordinary (unlinked) product partition function by 
a power of $q$ determined by overall linking number
\cite{Ekholm:2018eee, Ekholm:2019lmb}
\be\label{eq:cobordism-link-Gamma}
	\sfW^{(\gamma_1)}_{(d_1)} \dots \sfW^{(\gamma_m)}_{(d_m)}
	\ \mapsto \ 
	q^{\bd^t \cdot\Gamma\cdot \bd} \,
	\prod_{j=1}^{m} (q^{\beta_j} \mu^{m_j})^{d_j}
\ee
Here $\Gamma \in \IZ^{m\times m}$ is a symmetric quadratic form 
\be
	\Gamma_{ij} = \Gamma_{ji} = \mathrm{lk}(\gamma_i,\gamma_j)\,,
\ee
see Corollary \ref{c:U(1) knot complement} for the definition of linking number.

\subsubsection{The general $\mathfrak{gl}_1$-formula }
Collecting the above observations leads to the following specialization of the skein-valued curve counting formula \eqref{eq:localization-formula-final-skein}. 

\begin{proof}[Proof of Corollary \ref{c:U(1) knot complement}]
Section \ref{ssec:gl1skein} and \ref{ssec:cobordism} gives the following formula for the evaluation of curve counts on $M_K$ in $\mathfrak{gl}_1(M_K)$, by \eqref{eq:localization-formula-skein}  
\be\label{eq:localization-formula-final}
	\psi_{M_K}(a,q,\mu) = \sum_{d_1,\dots, d_m\geq 0}
	H_{(d_1)^\vee,\dots, (d_m)^{\vee}}(a,q)
	\,
	q^{\bd^t \cdot\Gamma\cdot \bd}
	\,
	\prod_{j=1}^{m}
	(\epsilon_j \,\sigma_j
	\,
	q^{\beta_j} \mu^{m_j})^{d_j}
\ee
where the HOMFLYPT polynomial of the link $\gamma_1\,\dots,\gamma_m$ is colored by either symmetric or antisymmetric partitions along the $j$-th component according to the type of flow loops
\be\label{eq:d-gamma-vee}
	\begin{array}{|c|c|c|}
		\hline
		\text{annulus type} & (\sigma_j,\epsilon_j) & (d_j)^\vee \\
		\hline\hline
		\text{bosonic annulus} & (+1,+1) & (d_j)\\
		\hline
		\text{femionic annulus} & (+1,-1) & (d_j)\\
		\hline
		\text{bosonic anti-annulus} & (-1,+1) & (d_j)^t= (1)^{d_j} \\
		\hline
		\text{fermionic anti-annulus} & (-1,-1) & (d_j)^t= (1)^{d_j} \\
		\hline
	\end{array}
\ee   
The assignment of an anti-annulus to elliptic flow loops follows from Section \ref{ssec:annuli}, keeping in mind that we match orientations and spin structures of $S^3$ and $M_K$. 
\end{proof}

The partition function of Corollary \ref{c:U(1) knot complement} determines a quantization $\hat A_{M_K}$ of the augmentation curve for the complement branch, by demanding that it annihilates \eqref{eq:localization-formula-final}. 

\begin{remark}
In fact, the quantum curve $\hat A_{M_K}$ can be obtained from the quantum operators that annihilate the \emph{conormal} partition function of the link $\{\gamma_1,\dots,\gamma_m\}$, i.e. the operator equations that encode recursion relations for $H_{(d_1)^\vee,\dots, (d_m)^{\vee}}(a,q)$.
The derivation involves introducing corrections due to linking, and then elimination of variables in the $\mathfrak{gl_1}$-skein of $\partial \CL_\alpha$. 
The cobordism operator that introduces linking is described in Appendix \ref{app:cobordism-quivers}. 
Below we illustrate, for simplicity, the classical version of this procedure in the specific case of $T_{2,2p+1}$ torus knots.
\end{remark}

\subsection{Explicit formula for $(2,2p+1)$-torus knots}\label{sec:torus-knots}
Recall from Proposition \ref{prp:2,2p+1torusknot} that $(2,2p+1)$-torus knots admit a 1-form on the complement with only two simple flow loops forming a negative Hopf link. 
In particular, we find an elliptic flow loop supporting a bosonic anti-annulus in class $\mu^2$ 
and a negative hyperbolic flow loop supporting a fermionic annulus in class $\mu^{2p+1}$.
Therefore, the corresponding winding numbers and signs in \eqref{eq:localization-formula-final} are, according to \eqref{eq:annuli-explicit}
\be\label{eq:2-torus-data-m-sigma}
	m_1=2,
	\quad 
	m_2=2p+1,
	\quad 
	\sigma_1=-1,
	\quad 
	\sigma_2=1,
	\quad
	\epsilon_1 = +1,\quad \epsilon_2=-1.
\ee
This leads to the following specialization of that general formula
\begin{proposition}\label{prop:2-torus-formula}
    The $\mathfrak{gl}_1$-partition of the complement $M_{T_{(2,2p+1)}}$ of a $(2,2p+1)$-torus knot is 
\be\label{eq:gl1-formula-2-2p+1-knots}
    	\psi_{(X,M_{K})}
	=\sum_{n_1, n_2 \geq 0}
	H_{(1)^{n_1},(n_2)}(a,q)
	\,
	q^{\Gamma_{11} n_1^2 + \Gamma_{22} n_2^2 +2\Gamma_{12} \, n_1 n_2}
	\,
	(- q^{\beta_1} \,\mu^2)^{n_1}
	(- q^{\beta_2}\,\mu^{2p+1})^{ n_2}.
\ee
    Here $H_{(1)^{n_1},(n_2)}(a,q)$ denotes the HOMFLYPT polynomial of the negative Hopf link in $S^3$ in representation $((n_1)^t,(n_2))$,   
\be\label{eqGamma-matrix-2p-1}
	\Gamma = \left(\begin{array}{cc}
	0 & -1 \\
	-1 & -2p-1
	\end{array}\right),
\ee
and $\beta_1$ and $\beta_2$ depends on the choice of basis $\mu$. 
Note that under a change of basis $\mu\mapsto q^{s}\mu$, $\beta_1$ and $\beta_2$ satisfy
\be\label{eq:beta-i-relation}
	\beta_1\mapsto \beta_1 + 2s,\quad \beta_2\mapsto \beta_2 + (2p+1)s\,.
\ee
There is a 4-chain for $M_K$ and a choice of basis element $\mu$ such that 
\be\label{eq:betas}
	\beta_1=1\,,
	\qquad
	\beta_2=2p+2\,.
\ee
\end{proposition}

\begin{proof}
Proposition \ref{prp:2,2p+1torusknot} gives two flow loops that form a {negative} Hopf link. It then remains only to determine the linking matrix $\Gamma$ and the $q$-powers of the basic annuli. 

Consider the planar disks $D_k$ bounded by $\gamma_k$. This disk intersects $T_{2,2p+1}$ transversely in $2$ points for $k=1$ and $2p+1$ points for $k=2$. If we take the standard meridian of $T_{2,2p+1}$ as a generator of the $\mathfrak{gl}_1$-skein of $M_{T_{2,2p+1}}$ then these disks shows that $\gamma_1$ is in class $\mu^2$ and $\gamma_2$ is in class $\mu^{2p+1}$. 

The disks $D_1$ and $D_2$ also intersects $\gamma_2$ and $\gamma_1$ with intersection sign $-1$. 
Note here it is crucial that the orientations on $M_K$ and $S^3$ agree, so that the Hopf link has signature $-1$ in both.
This gives the off diagonal elements in $\Gamma$. In order to determine the diagonal elements of $\Gamma$ and the $q$-powers of the annuli we must specify a $4$-chain. 

We will take a $4$-chain that is a small perturbation of the $4$-chain given by the image $\xi$ under the shift map $S^3\setminus K\to M_K$ of the shifting vector field. We take the perturbing vector field to be a small shift of the gradient Morse function on the Seifert surface with a critical point moved slightly off of the flow loop. If $v$ is the perturbing vector field, $\xi'=\xi+v$, and $C$ is the 4-chain of $\xi'$ then $C\cap S^3=\gamma_1'\cup\gamma_2'$, where $\gamma_1$ and $\gamma_2$ are nearby parallels of $\gamma_1$ and $\gamma_2$ that lie at the location of the critical points of $v$. 

Note next that the inward normal to the disk $D_2$ is homotopic to $\xi'$ along $\gamma_1$ and that therefore $\gamma_1$ framed by $\xi'$ represents $\mu$ in the $\mathfrak{gl}_1$-skein of $M_K$. The annulus then contributes $q\mu$, where the power of $q$ is determined by $\mathrm{lk}(\gamma_1,\gamma_1'\cup\gamma_2')=1$, i.e., $\beta_1=1$.

The inward normal of $D_2$, in contrast, rotates with respect to $\xi$. There are local intersections with $C$ each time the inward normal is parallel to $\pm v$, which gives a contribution of $q^{2p+1}$ taking into account also the contribution from $\mathrm{lk}(\gamma_2,\gamma_1'\cup\gamma_2')=1$ gives $\beta_2=2p+2$.

The calculation of the diagonal elements of $\Gamma$ follows from the same calculations, we need to compute the intersection of $\gamma_j$ shifted by $\xi'$ with the deformation surface $D_j$. This is the same calculation as above save for the linking in $S^3$, and with our choice of orientation we get $\Gamma_{11}=0$ and $\Gamma_{22}=-2p-1$.  

\end{proof}

\subsection{Quiver structure}\label{sec:quiver-structure}

Counts of holomorphic curves on rank-1 knot conormal $A$-branes are known to admit, in many examples, a quiver structure \cite{Kucharski:2017ogk}. Concretely, this means that the rank-1 specialization of \eqref{eq : open string part function} takes the form
\be
	\psi_{(X,L_K)}(x,a,q) = \sum_{d_1,\dots, d_m} q^{C_{ij}d_i d_j} \prod_{j=1}^{m} \frac{x_j^{d_j}}{(q^2;q^2)_{d_j}}
\ee
where $C_{ij}=C_{ji}$ is a symmetric matrix with integer entries and $x_j = (-1)^{C_{jj}} q^{q_j} a^{a_j} x^{k_j}$ are determined by certain integers $q_j,a_j,k_j$, which are part of the data in the knots-quivers correspondence.

In \cite{Ekholm:2018eee, Ekholm:2019lmb} it was shown that the quiver structure expresses $\psi_{L_K}$ as a $\Sk_{\mathfrak{gl}_1}(L_K)$-valued partition function of $m$ formal basic disks $\{D_1,\dots, D_m\}$ whose boundaries have linking 
\be
	\mathrm{lk}(\partial D_i, \partial D_j) = C_{ij}\,.
\ee
Each disk contribution is given by the multi-covering formula
\be
	\psi_{\disk}(X_k) = (X_k;q^2)_\infty^{-1}
	=
	\exp\left(\sum_{d\geq 1} \frac{X_k^d}{d(1-q^{2d})} \right)
	\,,\qquad
	X_k = (-1)^{C_{jj}} q^{C_{jj}-1}\hat x_j \prod_{k}\hat y_k^{C_{jk}}\,,
\ee
where
$\hat y_j\hat x_i = q^{2\delta_{ij}}\hat x_i \hat y_j$ is the quantum torus algebra of holonomies on tubular neighbourhoods of basic disks $\partial D_j$ in $L_K$.
The partition function can be written in terms of these as follows
\be\label{eq:quiverdiskformula}
	\psi_{(X,L_K)} = \langle\psi_{\disk}(X_m)\dots \psi_{\disk}(X_1)\rangle_{\Sk_{\mathfrak{gl_1}}(L_K)}
\ee
where evaluation in the skein of $L_K$ takes $\langle \hat x_k\rangle = x_k$. We point out that in most examples the basic disks are formal and does not correspond to actual holomorphic disks. (The notation in \eqref{eq:quiverdiskformula} differs slightly from \cite[eq. (6.9)]{Ekholm:2019lmb}. The meaning of $X_k$ is slightly different, this alters the commutation relation, but the partition function is recovered in the same way, see \cite{Kucharski:2025lcr} for a related discussion. Also `normal ordering' there here simply corresponds to evaluation in the $\mathfrak{gl}_1$ skein of $L_K$.)

Quiver descriptions of $\mathfrak{gl}_1$ skein-valued curve counts have been found beyond the context of knot conormals in resolved conifold, including for Lagrangian $A$-branes corresponding to other branches of augmentation varieties \cite{Ekholm:2021irc}, as well as toric branes in certain toric Calabi-Yau threefolds \cite{Panfil:2018faz}. See \cite{Kucharski:2025tqb} for a recent overview of results.

Here we will use not only basic disks but also basic annuli. 
We consider generalized quivers that includes also annuli,
whose $\mathfrak{gl}_1$ skein-valued partition functions are given in 
\eqref{eq:annuli-explicit}.

The HOMFLYPT polynomials $H_{(1)^{n_1},(n_2)}(a,q)$ of the (negative) Hopf link are well-known, 
and we show in Appendix \ref{app:Hopf} that they admit a quiver description involving six basic curves: two disks on each conormal, an annulus and an anti-annulus stretching between the two conormal components.
The overall linking matrix is
\be
	C_{\widetilde{\text{Hopf}}} = \left(\begin{array}{cc|cc|cc}
		1 & 0 & 0 & 0 & 0 & 1 \\
		0 & 0 & 0 & 0 & 0 & 1 \\
		\hline
		0 & 0 & 0 & 0 & 0 & -1 \\
		0 & 0 & 0 & 1 & 0 & -1 \\
		\hline
		0 & 0 & 0 & 0 & 0 & 0 \\
		1 & 1 & -1 & -1 & 0 & 0 \\
	\end{array}\right)
\ee
where the first (resp. second) two rows/colums correspond to the disks on the first (resp. second) component of the Hopf conormal, and the last two rows/colums correspond to the annulus and anti-annulus.
The cobordism of the Hopf link conormal into the complement $M_K$ maps the boundaries of disks and annuli to the flow lines in homology classes $\mu^2$ and $\mu^{2p+1}$. In addition, it further modifies linking among the six basic curves by $\Gamma$ from \eqref{eqGamma-matrix-2p-1}, see Appendix \ref{app:Hopf} for details.
Overall, this leads to a quiver description for curve counts on the knot complement of $(2,2p+1)$-torus knots with six basic curves (four disks, an annulus and an anti-annulus) with the following linking
\be
\begin{split}
	C_{M_{K}}
	&=
	\left(\begin{array}{cc|cc|cc}
		1 & 0 & -1 & -1 & -1 & 0 \\
		0 & 0 & -1 & -1 & -1 & 0 \\
		\hline
		-1 & -1 & -2p-1 & -2p-1 & -2p-2 & -2p-3 \\
		-1 & -1 & -2p-1 & -2p & -2p-2 & -2p-3 \\
		\hline
		-1 & -1 & -2p-2 & -2p-2 & -2p-3 & -2p-3 \\
		0 & 0 & -2p-3 & -2p-3 & -2p-3 & -2p-3 \\
	\end{array}\right)
\end{split}
\ee
The open string partition function for the $T_{2,2p+1}$ knot complement in Proposition \ref{prp:2,2p+1torusknot} then admits the following presentation as a quiver-like partition function
\be\label{eq:ZMK-final-torus-knots}
\begin{split}
	\psi_{(X,M_{K})}
	& = \sum_{d_1, d_2, d_3, d_4, d_5\geq 0}\sum_{d_6=0}^{1}
	q^{\bd^t\cdot C_{M_{K}}\cdot \bd}
	\, 
	q^{\beta_1 (d_1+d_2+d_5+d_6)+\beta_2 (d_3+d_4+d_5+d_6)}\,
	\\
	&\qquad\qquad\qquad\times
	\frac{(- \mu^{2})^{ d_1}}{(q^2;q^2)_{d_1}} 
	\frac{(a^2 q \mu^2)^{d_2}}{(q^2;q^2)_{d_2}}
	\frac{(-q \mu^{2p+1})^{d_3}}{(q^2;q^2)_{d_3}} 
	\frac{(a^2 \mu^{2p+1})^{d_4}}{(q^2;q^2)_{d_4}}
	(-a^2 \mu^{2p+3})^{d_5} (\mu^{2p+3})^{d_6}
\end{split}
\ee
where $\bd^t = (d_1,\dots, d_6)$ is the quiver charge vector, and $\beta_i$ are given by \eqref{eq:betas}.

\begin{remark}[Quivers and integrality of M2 brane counts for complements of $(2,2p+1)$ torus knots]\label{rmk:integrality}
Our results directly show that curve counts on knot complements of $(2,2p+1)$ torus knots are organized by a quiver structure.
We recall that partition functions of symmetric quivers admit a factorization in terms of integer-valued motivic Donaldson-Thomas invariants \cite{Kontsevich:2010px, efimov2012cohomological}.
Our quiver formula \eqref{eq:ZMK-final-torus-knots} therefore implies integrality counts of M2 branes with boundaries on a single M5 brane wrapping the knot complement Lagrangian $M_K$. In \cite{Marino:2023eto} an analogous integrality structure was conjectured for complex Chern-Simons partition functions on hyperbolic knot complements. It appears to us that this conjecture is related to our statement and it would be interesting to explore this connection.
Our statement is a counterpart of the conjectural knots-quivers correspondence \cite{Kucharski:2017ogk} which pertains to curve counts on knot conormal Lagrangians. 
The corresponding integrality structures for knot conormals were conjectured long ago in \cite{Ooguri:1999bv, Labastida:2000zp,  Labastida:2000yw}. 
\end{remark}

\subsection{Augmentation curves and Lagrangian cobordisms}\label{sec:aug-curves}

The cobordism \eqref{eq:cobordism-L-MK} gives a new way of computing the augmentation curve for $(2,2p+1)$-torus knots as follows. The twisted Hopf link augmentation variety is a 2-complex dimensional subvariety of $(\IC^*)^4$, which has a branch described by the vanishing locus of
\be
\begin{split}
	B_1(\xi_i, \eta_i) & = 1-\eta_1 - \xi_1\eta_1+a^2 \xi_1 \\
	B_2(\xi_i, \eta_i) & = 1-\eta_2 - \xi_2 + a^2 \xi_2 \eta_2 \\
\end{split}
\ee
where $\xi_i$ denote longitudes and $\eta_i$ meridians on the Hopf link components, see Figure \ref{fig:torus-link-flow-loops}. Since there are no holomorphic disks with boundary in $\CL_\alpha$, the bounding chains of the disks near one of the Hopf link can be continued to a bounding chain in $M_K$. It intersects the central curve of the other Hopf link component according to its linking number. This leads to the following modification in the augmentation equations, compare \eqref{eq:Hopf-Gamma-link}, 
\be
	\tilde B_k(\xi_i,\eta_i):= B_k(\xi_i \prod_j \eta_j^{\Gamma_{ij}}, \eta_i) = 0\,,\qquad k=1,2\,.
\ee 
Then
\be\label{eq:tilde-B}
\begin{split}
	\tilde B_1 & = 1-\eta_1 - \xi_1\eta_1\eta_2^{-1}+a^2 \xi_1\eta_2^{-1} \\
	\tilde B_2 & = 1-\eta_2 - \xi_2\eta_1^{-1}\eta_2^{-(2p+1)} + a^2 \xi_2  \eta_1^{-1}\eta_2^{-2p} \\
\end{split}
\ee
(These relations admit a fairly straightforward quantum lift, see Appendix \ref{app:cobordism-quivers} for a discussion and their relation to the quantum cobordism operator $\hat\CO$.)
Now, specializing variables 
\be\label{eq:xi-mu}
	\xi_i = \epsilon_i \sigma_i \mu^{m_i}  = -\mu^{m_i}
\ee 
where we used $\sigma_i$ and $\epsilon_i$ given in \eqref{eq:2-torus-data-m-sigma},
gives (The negative power in $\lambda^{-1}$ appears to preserve the overall sign of the symplectic structure and comes from the different orientation of boundaries along $K$ and along the Hopf link, see Figure \ref{fig:torus-link-flow-loops}.)
\be\label{eq:lambda-eta}
	\log \lambda^{-1} = \frac{\partial W_{M_K}}{{\partial \log \mu}} = \sum_{j=1}^{2}\frac{\partial \tilde W}{\partial\log \xi_i} \frac{\partial \log \xi_i}{\partial\log \mu}\Big|_{\xi_i=-\mu^{m_i}} =\log  \eta_1 ^{m_1} \eta_2 ^{m_2}\Big|_{\xi_i=-\mu^{m_i}}
\ee
where $W_{M_K}$ is the disk potential for $M_K$ given by the leading asymptotics of \eqref{eq:gl1-formula-2-2p+1-knots}
 \be\label{eq:W-complement-general}
 	\lim_{q\to 1} (q^2-1) \log \psi_{(X,M_K)} 
 \ee
and $\tilde W$ is the superpotential whose critical locus is defined by $\{\tilde B_k=0\}$, which means $\log \eta_i = \frac{\partial \tilde W}{\partial\log\xi_i}$. Winding numbers are $m_1=2, m_2=2p+1$.

Relations \eqref{eq:xi-mu} can also be understood by studying homology generators in $M_K$ after removing a neighbourhood of the Hopf link as in Figure \ref{fig:torus-link-flow-loops}, see Figure \ref{fig:homology-map}. 
Similarly \eqref{eq:lambda-eta} is evident by following once along the longitude of $K$ as shown in Figure \ref{fig:torus-link-flow-loops}, and noting that it winds twice around $\eta_1$ in the negative direction, and $2p+1$ times around $\eta_2$ also in the negative direction.
Taken together, equations \eqref{eq:tilde-B}, \eqref{eq:xi-mu} and \eqref{eq:lambda-eta} yield a relation between $\lambda$ and $\mu$ which contains $A_{K}(\lambda,\mu)$ as one of its factors. An illustration of this will be given in sections \ref{sec:unknot-aug-curve} and \ref{sec:trefoil-aug-curve}.

\begin{figure}[h!]
\begin{center}
\includegraphics[width=0.5\textwidth]{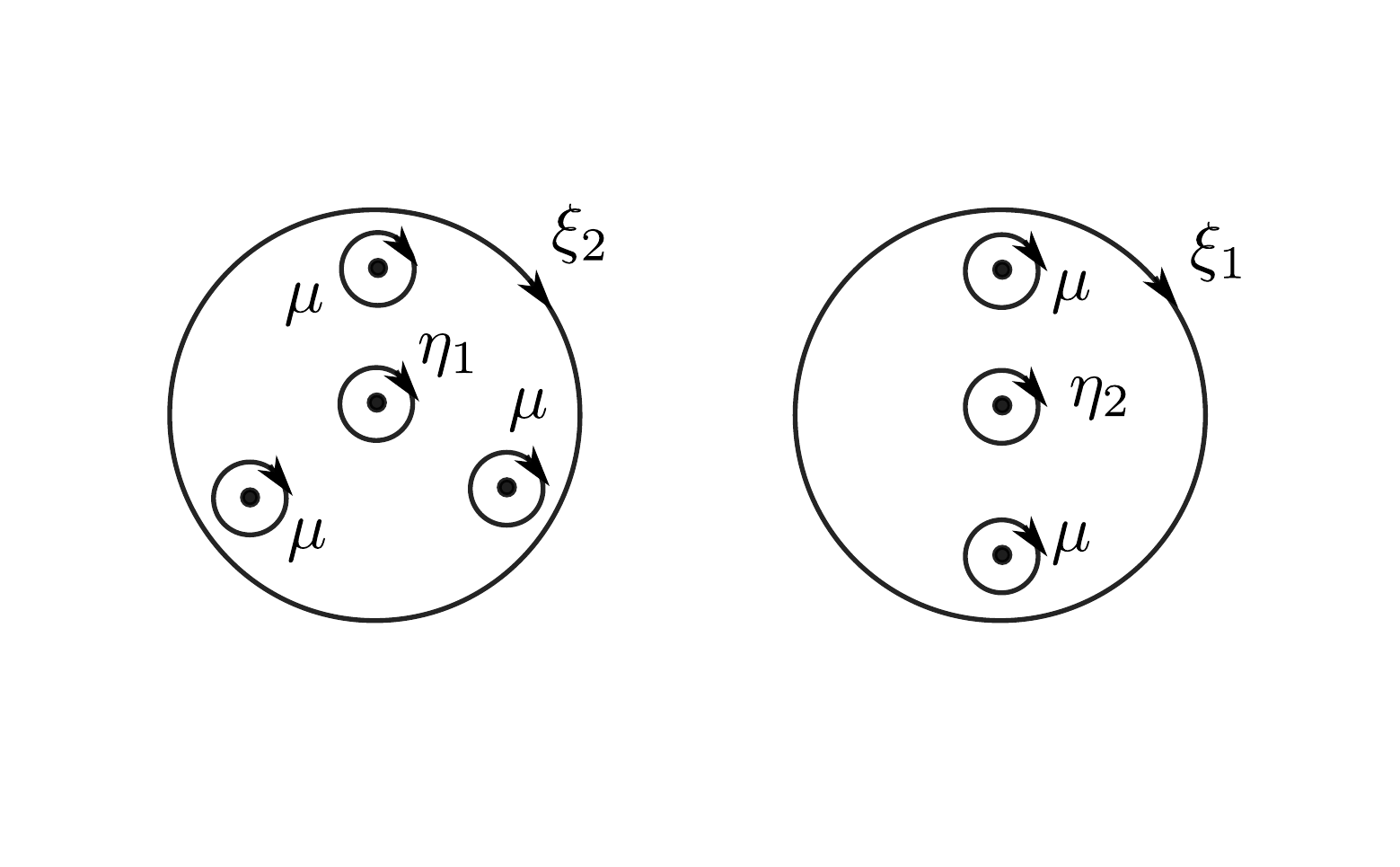}
\caption{Illustration of \eqref{eq:xi-mu} and \eqref{eq:lambda-eta} for the case $p=1$. If the trefoil is removed by setting $\mu=1$, we recover $\xi_1=\eta_2$ and $\xi_2=\eta_1$ (up to signs), as appropriate for the Negative Hopf link complement, see \cite[Section 7.2.1]{Ekholm:2024ceb}. On the other hand, trivializing the Hopf longitudinal holonomies $\eta_2=\eta_1=1$ gives $\xi_1=\mu^2,\xi_2=\mu^3$ (up to signs).}
\label{fig:homology-map}
\end{center}
\end{figure}

\section{Quantizations of augmentation curves and connections to M theory}\label{eq:quantum-branches}

Knowledge of the open string partition function on fibered knot complements allows for concrete investigations of quantizations of augmentation varieties. There are two standard branches in the augmentation curve of any knot: the conormal and the complement branch.
The open string partition function of the conormal branch is given by the generating series of HOMFLYPT polynomials, while for fibered knots, that of the complement branch is given by \eqref{eq:localization-formula-skein}.
Each of these is annihilated by a specific quantization of the classical augmentation curve. In this section we ask whether these quantizations agree or not.

We develop a general definition of phases of augmentation curves motivated by physics: lifting topological string theory to M-theory and taking the M-circle to infinite radius leads to a tropical limit of the augmentation curve which defines natural `phases' for a brane. These include the complement and the conormal phase, as well as additional phases without straightforward geometric interpretation (except when $a=1$, see \cite{cornwell}). Nevertheless, for each phase we define local canonical coordinates on the quantum torus of the Legendrian torus of the knot and a quantization scheme defining formal counts that are direct counterparts of the $\mathfrak{gl}_1$-skein curve counts for phases with Lagrangians fillings.

\subsection{Quantization scheme}\label{ssec: quantization scheme}
Quantizing a mirror curve can often be thought of as constructing a $D$-module associated to a Lagrangian submanifold of a symplectic manifold. Many approaches to this problem have been developed from a physical/string theoretic point of view, see e.g. \cite{Eynard:2007kz, Dijkgraaf:2007sw, Gukov:2008ve, Dimofte:2011gm, Gukov:2011qp, Grassi:2014zfa, Marino:2023eto}. 

Here we consider the problem of quantizing augmentation curves of knots in this way. This means that we wish to promote a Lagrangian submanifold of $(\IC^*)^2$ defined by a two-variable polynomial $A(x,y)$ (depending on the additional parameter $a$) to an element of the $\mathfrak{gl}_1$ skein module of $T^2$ 
\be
	A(x,y) \quad\to\quad \hat A(\hat x,\hat y) \ \in\  \Sk_{\mathfrak{gl}_1}(T^2)\,.
\ee
Quantization of augmentation curves is expected to yield a $D$-module both from a physical point of view, involving fermions on the curve arising on intersecting branes \cite{Dijkgraaf:2007sw} after a chain of dualities, and from a mathematical point of view by higher genus knot contact homology \cite{Ekholm:2018iso}, the relation between the two frameworks has been discussed in \cite{Aganagic:2013jpa}. One approach to constructing $D$-modules for general knots would be to generalize the elimination theory used to produce the augmentation polynomial to elimination theory in the $\mathfrak{gl}_1$-skein of $T^2$. However, at present no such general result is known. Another approach is to use fermions on the augmentation curve itself, which require taking into account singularities of the augmentation curve that appear for generic knots, and again general results in this are lacking.

Parts of the classical augmentation curve are known to parameterize moduli spaces of Lagrangians $L_\alpha$ with rank-1 local systems that fill the Legendrian knot conormal in the resolved conifold. 
Topological string theory (or equivalently skein valued curve counts) associates to each such Lagrangian a partition function, which takes values in the $\mathfrak{gl}_1$ skein module 
\be
	\psi_\alpha \ \in\  \Sk_{\mathfrak{gl}_1}(L_\alpha)\,.
\ee
In local coordinates (defined below), each $\psi_\alpha$ defines a $D$-module, the corresponding ideal consists of the operators that annihilate $\psi_\alpha$. Below we will check whether such $D$-modules corresponding to spaces of different Lagrangians (e.g. of different topology) are compatible, in the sense that they generated  the `same' ideal $\hat A = 0$.

\begin{remark}
As mentioned above, when elimination works for higher-genus knot contact homology \cite{Ekholm:2018iso}, a universal operator $\hat A$ is constructed from data at infinity, and that operator would annihilates all $\psi_\alpha$. It turns out that $\hat A$ is not entirely independent of branch. It was observed in \cite{Ekholm:2024ceb} that $\hat A$ depends on auxiliary geometric data associated to the Lagrangian fillings on different branches.  
\end{remark}

\subsection{Decompactification is tropicalization}

The augmentation curve of a knot is an algebraic variety $\Sigma_K\subset (\IC^*)^2$. 
Similar curves arise in physics as the moduli space of vacua of a 3d $\CN=2$ QFT $T_{3d}$ on $S^1\times \IR^2$ \cite{Ooguri:1999bv, Dimofte:2010tz, Aganagic:2013jpa}, that also admits several other dual descriptions.
Generally speaking, there is an $SL(2,\IZ)$ duality group that acts on 3d $\CN=2$ QFTs with $U(1)$ gauge group, where the standard generators $S$ and $T$ of $SL(2,\IZ)$ correspond respectively to gauging a $U(1)$ flavor symmetry, and to shifting the Chern-Simons level \cite{Witten:2003ya, Dimofte:2011ju}.
Another duality, specific to augmentation curves, arises from M-theory.
The 3d-3d correspondence provides a realization of $T_{3d}$ by considering the twisted compactification of a single M5 brane (the generalization to a stack of M5 branes gives rise to skein-valued augmentation curves \cite{Ekholm:2019yqp}). on a 3-manifold $M_\alpha$ times $S^1\times_q \IR^2$. 
We denote by $T_\alpha$ the twisted compactification of the Abelian 6d $(2,0)$ theory on $M_\alpha$, possibly specified by additional data (e.g., data encoding a choice of topologically inequivalent classes of flat connections on the same 3-manifold $M_\alpha$).
Different `branches' of $\Sigma_K$ (a precise definition will be given below) correspond to different 3d-3d realizations of $T_{3d}$, related by a duality web shown in Figure \ref{fig:dual-branches}.
\begin{figure}[h!]
\begin{center}
\includegraphics[width=0.6\textwidth]{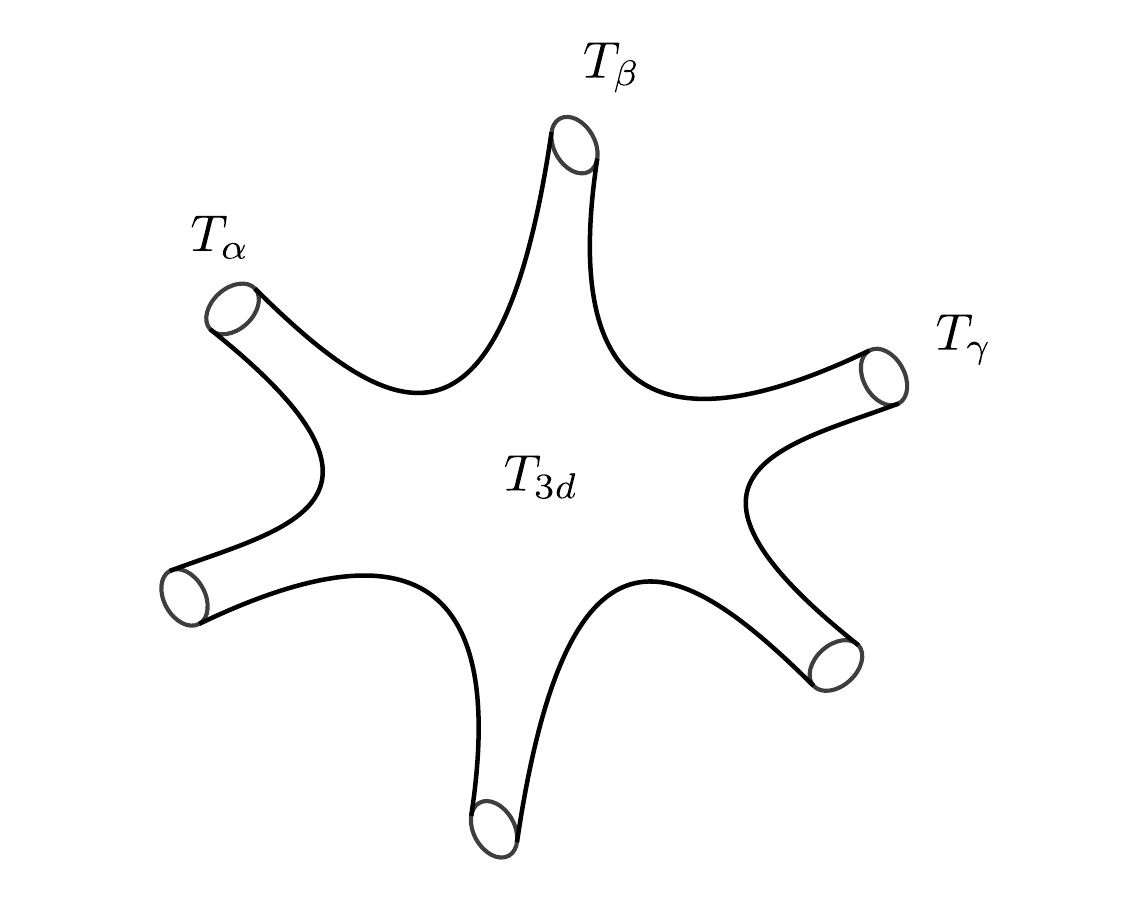}
\caption{Dual Lagrangian descriptions of $T_{3d}$ on $S^1\times \IR^2$ associated with branches of $\Sigma_K$.}
\label{fig:dual-branches}
\end{center}
\end{figure}

In all examples considered, each branch of $\Sigma_K$ leads to a Lagrangian description  $T_\alpha$ which is an Abelian Chern-Simons matter theory encoded by a symmetric quiver $Q$, introduced in \cite{Ekholm:2018eee}. 
From the point of view of a given Lagrangian description $T_\alpha$, branches of $\Sigma_K$ correspond to Coulomb, Higgs, or topological vacua of the 3d $\CN=2$ QFT. 
A duality relating different Lagrangian descriptions $T_\alpha\simeq T_\beta$ then interchanges these interpretations, identifying the Higgs vacua of $T_\alpha$ with, say, topological vacua of $T_\beta$, etc.

The geometric interpretation of phases (and dualities) is less clear. Some branches $T_\alpha$ are associated with a 3-manifold $M_\alpha$, but the topology of $M_\alpha$ is not fixed, and may change from branch to branch
$M_\alpha\not\approx M_\beta$ for $\alpha\ne \beta$.
This then means that QFT dualities translate into topological transitions of 3-manifolds. 
Other branches $T_\beta$ does not have a well-understood geometric interpretation in terms of 3-manifolds.

For branches $T_\alpha$ that do admit a 3-manifold interpretation, cylindrical region $\IC^*_\alpha\subset \Sigma_K$ can be regarded as the moduli space $\IR^{b_1}\times (S^1)^{b_1}$ of a Lagrangian $L_\alpha$ with a rank-1 abelian local system, if $b_1 \equiv (L_{\alpha}) = 1$. 
When such patches $\C_\alpha^\ast$ and $\C_\beta^\ast$ overlap in $\Sigma_K$, there are simultaneous dual interpretations as moduli spaces of both $L_\alpha$ and $L_\beta$.

The change in topology of $M_\alpha$ over a smooth moduli space $\Sigma$ is in string theory and M-theory an effect of open M2 branes ending on the M5 brane \cite{Aganagic:2001nx}. 
From the viewpoint of QFT, M2 branes that wrap a cycle in $L$ times the circle descend to BPS vortices.%
(In fact, transitions among different branches are often identified (in a suitable Lagrangian description) with points where the FI parameter vanishes. Here vortex contributions to the partition functions formally diverge. A dual Lagrangian description related by particle-vortex duality (the $S$ transform) becomes more appropriate to describe the local dynamics in this case.)
A 3d theory $T_{3d}$ compactified on $S^1\times \IR^2$ may be regarded as a 2d $(2,2)$ theory $T_{2d}$ on $\IR^2$, whose fields coincide with Fourier modes of 3d fields complexified by gauge and flavour $S^1$-holonomies
\be
	S^1\ \text{reduction:} \quad T^{\IR^3}_{3d} \to T^{S^1\times \IR^2}_{3d}\approx T_{2d}\,.
\ee
The complexified (Fayet-Ilioupoulos) couplings and (vectormultiplet scalar) expectation values provide local coordinates on $\IC^*\times \IC^*$, where $\Sigma_K$ traces the space of vacua $y(x)$ for any given coupling $x$.

From the point of view of $S^1$ compactification, $\Sigma_K$ is presented as a (dual) circle fibration over the moduli space of $T_{3d}$ on $\IR^3$. 
Singularities of the moduli space, where certain BPS states would become massless, are smoothed out in $T_{3d}^{S^1\times \IR^2}$, because BPS particles with worldline wrapping $S^1$ contribute finite-size corrections to the low energy effective action \cite{Aharony:1997bx}.
There is a different circle fibration for each duality frame $T_\alpha$, corresponding to different ways in which 3d dualities of $T_{3d}^{\mathbb{\IR}^{3}}$ descend to 2d dualities of $T_{3d}^{S^1\times \IR^2}\approx T_{2d}$ \cite{Aharony:2017adm}.

We next consider the connection to tropical geometry.
Taking the radius $R$ of $S^1$ to infinity $R\to\infty$ gives $\IR^3$ from $\IR^2\times S^1$, all $S^1$-holonomies of the 3d QFT are forced to be trivial. On the 3d-3d dual side this turns off the moduli for the flat connection on $M_\alpha$, since the 2-form on the M5 brane cannot be reduced to a connection 1-form along $S^1$ anymore.
As holonomies are frozen to zero, only geometric moduli of the undelying 3-manifold $L_\alpha$ vary and only half of the dimension of the space of vacua remain.
Moreover, from a dynamical perspective the $R\to\infty$ limit suppresses contributions to the low energy effective action from finite-mass BPS states wrapping $S^1$. These BPS states descended from M2 branes responsible for smoothing the junctions among cylindrical regions, hence in the limit the space of vacua is expected to become singular.
The result of taking $R\to\infty$ is expressed mathematically as taking the tropical limit
\be
	S^1\ \text{decompactification:}\quad \Sigma_K \to \Trop_K\,.
\ee
An example is given in Figure~\ref{fig:tropical-phases}, see \cite[Section 5.3]{Ekholm:2019lmb} for further details.

\begin{figure}[h!]
\begin{center}
\includegraphics[width=0.6\textwidth]{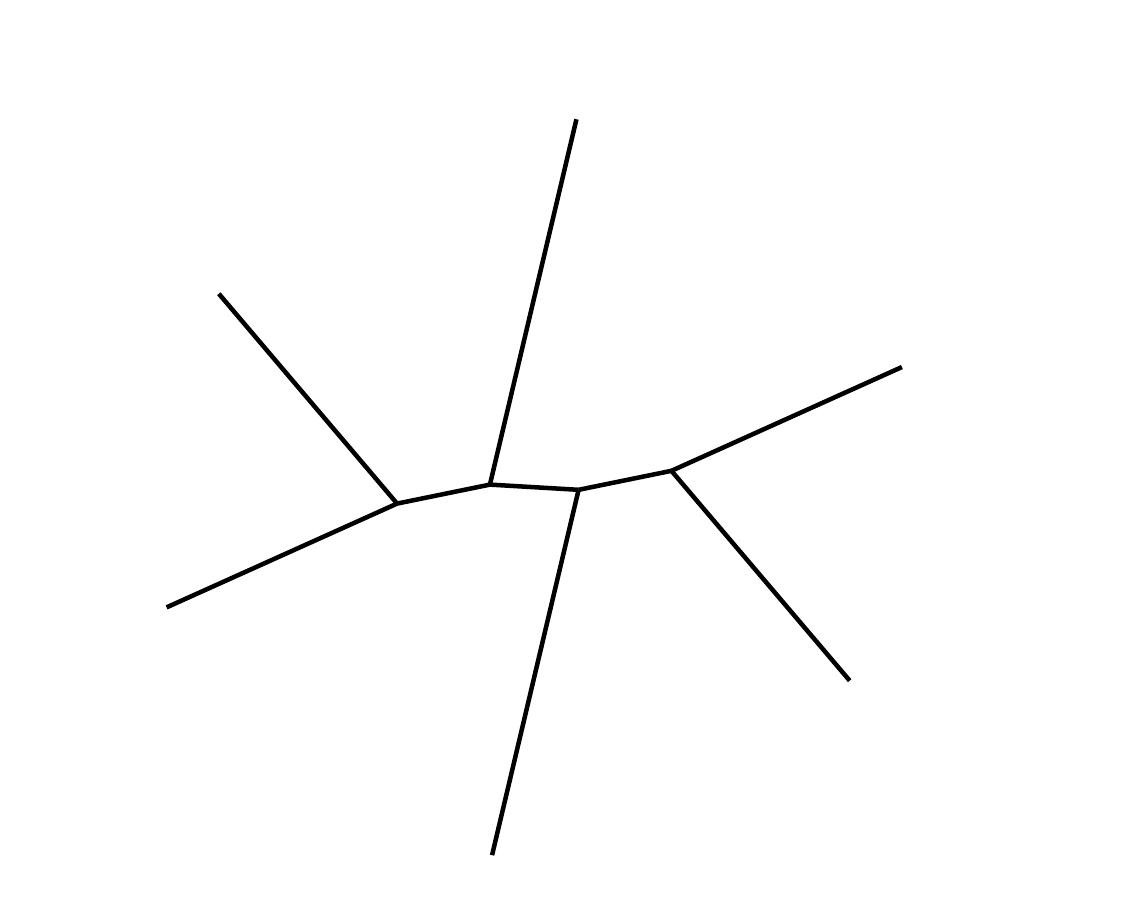}
\caption{The $R\to\infty$ limit of the curve in Figure \ref{fig:dual-branches}, the tropical curve $\Trop_K$.}
\label{fig:tropical-phases}
\end{center}
\end{figure}

In general, the tropical limit of an augmentation curve
\be\label{eq:augmentation-curve}
	\Sigma_K \subset(\IC^*)^2  : \qquad
	\sum_{i,j\in\mathbb{Z}}c_{i,j}a^{v_{i,j}}x^iy^j = 0 
\ee
where $a,x,y\in\mathbb{C}^*$, $v_{i,j}\in\mathbb{Z}$, and $c_{i,j}\in\mathbb{Z}$,
is defined as the set of edges of the graph of the piecewise-linear function, see e.g., \cite{brugalle2015brief} for notions of tropical geometry.
\be\label{eq:trop-curve-def}
	\Trop_K\subset\IR^2:
	\qquad
        \text{Max}\{v_{i,j}*\log|a|+i*\log|x|+j*\log|y|:\forall i,j\in\mathbb{Z}\}\,.
\ee
We point out that in the tropical limit, numerical coefficients are suppressed by $R\to\infty$, however closed string moduli such as $a$ are retained and become real-valued (like ambient coordinates $x, y$). This is because K\"ahler moduli are complexified by the $B$-field, which arises from the circle reduction of the M-theory 3-form, and this complexification is turned off in the decompactification limit.

Tropical curves have both internal edges attached to two vertices, and external edges attached to only one vertex. 
As stratified spaces, tropical curves are characterized by incidence relations among edges.
This structure is piecewise constant as one varies the continuous moduli, which in the case of $\Sigma_K$ include $|a|$ the (exponentiated) volume of $\IP^1$ in the resolved conifold. However at the real-codimension 1 locus defined by $|a|=1$, the structure of a tropical curve can change significantly. 
This leads to two distinct phases for tropical augmentation curves, corresponding to $|a|<1$ and $|a|>1$, each with its own collection of edges and vertices.
In addition, in the limits $|a| \to 0$ or $|a| \to \infty$ all internal branches become external. We illustrate these phenomena with a few examples in Section \ref{sec:tropical-examples}.
%

\subsection{Local open string moduli}\label{sec:local-moduli}
The open string partition function (open curve counts) involves the choice of open string moduli that measure the area of holomorphic curves. That is, the open topological string partition (or its logarithm, the free energy) function should be expressed as a formal power series in the (exponentiated, complexified) area of worldsheet instantons $\xi^\beta = e^{-\text{area}_{\mathbb{C}}(\beta)}$. Here $\beta\in H_2(X,L)$ is the relative homology class of a curve, and $|\xi^\beta|<1$ is contained within the unit disk for all $\beta$. Famously, the choice of such a parametrization is not unique due to a choice of framing \cite{Aganagic:2001nx}.

\subsubsection{Seifert global coordinates.}\label{eq:Seifert-global-coordinates}
Consider now the augmentation curve $\Sigma_K\subset (\C^\ast)^2$ of a knot $K$. There is a canonical choice of coordinates $(x,y)$ on $(C^\ast)^2$ associated to two external edges of $\Sigma_K$ associated to the 
conormal Lagrangian filling of the Legendrian conormal torus $\Lambda_K\subset U^\ast S^3$, and the complement filling $M_K\approx S^3\setminus K$. 

The conormal branch of the classical augmentation variety encodes holomorphic disks with boundary on the conormal Lagrangian $L_K$ of the knot (shifted off of $S^3$) in the resolved conifold $X= \CO(-1)\oplus \CO(-1)$. 
For fibered knots, the complement branch similarly encodes counts of holomorphic disks with boundary on $M_K$ (see \cite{Ekholm:2020lqy} for a discussion of the non-fibered case).

The coordinates on $(\C^\ast)^2$ are $y$ corresponding to the meridian in $\Lambda_K$, the cycle that bounds in $L_K$, and $x$ corresponding to the longitude, the cycle that bounds in $M_K$ which gives the Seifert framing.  
In these coordinates, the conormal branch of the augmentation curve is asymptotic to a puncture at $(x,y)=(0,1)$ and the complement branch is asymptotic to $(x,y)=(1,0)$.

\subsubsection{Definition of branches}
A branch of the augmentation curve corresponds to a pair $(e,v)$, where $e$ is an edge of the tropical diagram attached to a vertex $v$.
The set of branches is then $\{(e_\alpha,v_\alpha)\}_{\alpha}$, pairs of edges and vertices such that $v_\alpha\in \partial e_\alpha$.
Each internal edge defines two branches, while each external edge defines one. 

Any edge $e$ of the tropical diagram is characterized by a vector $(r_1, r_2) \in \IZ^2 / \IZ_2$. Alternatively we assign to each edge 
\be\label{eq:edge-def}
	\text{slope}: \ \ \frac{r_2}{r_1}\,,
	\qquad
	\text{orbifold index}: \ \ \gcd(|r_1|, |r_2|)\,.
\ee
The edge associated with the conormal branch has type $(1,0)$, while that of complement branch has type $(0,1)$, both have orbifold index $1$. The $(r_1, r_2)$-types of other edges can be determined by resolving a balancing condition at each vertex (with orientation of all edges flowing towards the vertex, one imposes that the sum over incident edges $\sum_{i} (r_{1,i}, r_{2,i}) = (0,0)$).

\subsubsection{Coordinates on branches}
We assign local coordinates to each branch $(e_\alpha, v_\alpha)$ in such a way that $v_\alpha$ is located at $(1,1)$.
If the edge $e_\alpha$ is of type $(r_1, r_2)$, we define the primitive vector $(r_1',r_2')$ with
\be\label{eq:sign-xy-def}
	r_i' := \sigma_{\alpha} \frac{r_i}{\gcd(|r_1|, |r_2|)}
\ee
where $\sigma_\alpha =\pm 1$ is chosen in such a way that $(r_1', r_2')$ points \emph{towards} $v_\alpha$.
Then the longitude and meridian coordinates of the branch are defined as follows
\be\label{eq:local-coordinates}
\begin{split}
	x_\alpha &= \xi_x \, x^{s_1'}y^{s_2'}\,,\\
	y_\alpha &= \xi_y \, x^{-r_2'}y^{r_1'}\,, 
\end{split}
\ee
where $s_1', s_2'$ are coprime integers obeying
\be
	r_1' s_1'  + r_2' s_2' = 1
\ee
and $\xi_x, \xi_y \in \IC^*$ are normalization coefficients chosen to ensure that $v_\alpha$ is at $(|x_\alpha|,|y_\alpha|)=(1,1)$.
Note that this definition does not uniquely fix $s_1', s_2'$, but leaves the freedom to change framing
\be
	(s_1', s_2') \sim (s_1' - f r_2', s_2' + f r_1')\,, \quad f\in \IZ\,. 
\ee

In particular we have for the conormal $\alpha=l$ and complement $\alpha=m$ branches in Seifert framing
\be\label{eq:conormal-complement-duality}
	\left\{\begin{split}
	x_l  &= x \\
	y_l  &= y \\
	\end{split}\right.
	\qquad\qquad
	\left\{\begin{split}
	x_m  &= y \\
	y_m  &= x^{-1} \\
	\end{split}\right.
\ee
The fact that these are related to each other by $S\in SL(2,\IZ)$ is a special feature of Seifert framing\,.
A change of framing for the conormal branch, i.e. taking $x\to x(-y)^f$ leads to a `tilting' of any branch that is not horizontal. This applies in particular to the complement branch, see Figure \ref{fig:framing-example} for an example.
The integer $f$, normalized so that $f=0$ corresponds to Seifert framing, corresponds to the self-linking number of $K\subset S^3$, as defined by the push-off corresponding to the framing vector field.

\begin{figure}[h!]
\begin{center}
\includegraphics[width=0.3\textwidth]{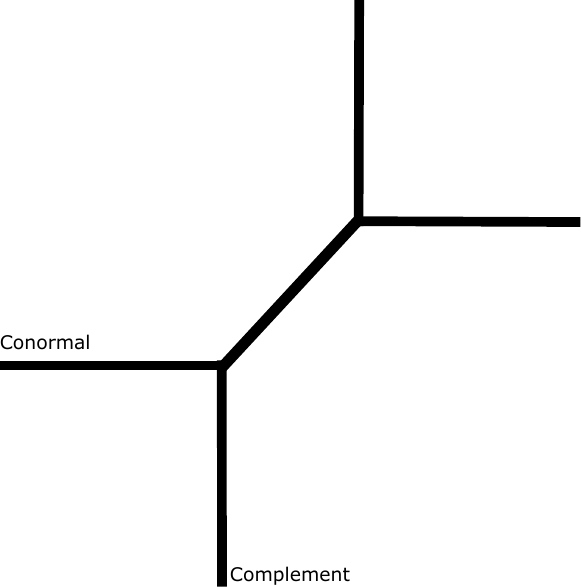}~~~~~~~~~~~~~~~~~~~~~~~~~~~~~~
\includegraphics[width=0.3\textwidth]{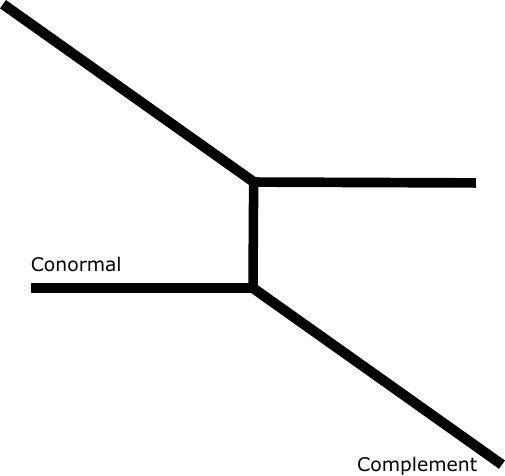}
\caption{
Left: tropical diagram in Seifert framing. Right: a different choice of (conormal) framing changes the slope of the complement branch.
}
\label{fig:framing-example}
\end{center}
\end{figure}

\subsubsection{Large volume open string coordinates.}
If $(x_\alpha, y_\alpha)$ are coordinates on an external branch, as the area of worldsheet instantons becomes large (hence $x_\alpha \sim e^{-\text{area}}\to 0$), their contributions to the meridian holonomy must be trivial and therefore $y_\alpha \sim e^{\oint_{\text{meridian}}{\nabla}}\to 1$. For internal branches the limit works the same after a suitable rescaling of the closed string moduli that effectively decompactifies the internal edge into an half infinite external edge. 

Note that $|x_\alpha|<1$ ($v_{\alpha}$ is placed at $x_\alpha=1$ and the sign in \eqref{eq:sign-xy-def} is such that $|x_\alpha|$ decreases along the edge $e_\alpha$). 
For internal branches $|x_\alpha|$ is also bounded from below by the length $\ell(e_\alpha)$ of the internal edge $e_\alpha$ $e^{-\ell(e_\alpha)} <|x_\alpha | <1$.
From \eqref{eq:trop-curve-def} we have $\ell(e_\alpha) = 2k \log |a|$ for some $k\in \IQ$, where $a$ is the closed string modulus of the resolved conifold. 

Taking this into account, the open string partition function can be organized as powers series:
\begin{itemize}
\item On external branches:
\be
	\psi_\alpha = \sum_{m,n\geq 0} c_{m,n}(q)  a^{2m} x_\alpha^{n}\,.
\ee
\item On internal branches:
\be
	\psi_\alpha = \sum_{l,m,n\geq 0} c_{l,m,n}(q) a^{2m} x_\alpha^{n}  (a^{2k} x_{\alpha}^{-1})^{l}\,.
\ee
\end{itemize}
In both cases, the expansion parameters lie within the unit disk.

\begin{remark}[Irrational open Gromov-Witten invariants]\label{rmk:irrational-OGW}
The formal disk potential (and its higher-genus counterpart) sometimes have \emph{irrational} coefficients, i.e., the formal disk potential has the form
\be
	W^{\mathrm{formal}}_{\alpha} = \sum_{l,m,n\geq 0} w_{l,m,n} a^{2m} x_{\alpha}^{n} a^{2k} (x_\alpha^{-1})^l
	\quad\text{with}\quad
	w_{l,m,n}  \in \IR \setminus \IQ\,.
\ee
see Section \ref{trefoil-int-bra}. In the examples below, such phenomena appear for branches of orbifold index $>1$ and  introducing a generic resolution of the orbifold restores rationality of curve counts. 
Thus, here irrational open Gromov-Witten invariants can be understood as the orbifold limit of rational curve counts where infinite towers of contributions are summed together in the limit. 

\end{remark}

\subsection{Consistency of quantization schemes}\label{sec:globalp-consistency}
Recall from Section \ref{ssec: quantization scheme} of quantization of augmentation curves $A_K(x,y)=0$. In this section we will present general formulas that check how the quantization $\hat A_K(\hat x,\hat y)$ transform between coordinates $(\C^\ast)_\alpha$ along $\Sigma_K$. 

Consider a branch $T_\alpha$ of $\Sigma_K$ with local coordinates $(x_\alpha,y_\alpha)$ in $(\C^\ast)^2$. Using the coordinate change in \eqref{eq:local-coordinates} we obtain a local equation 
\be
A_{K,\alpha}(x_\alpha,y_\alpha) =  \sum_{i,j} c^{(\alpha)}_{i,j}(a) \, x_\alpha^i y_\alpha^j = 0.
\ee
for $\Sigma_K$. Assume now further that we knot the wave function $\psi_\alpha(x_\alpha)$ along $T_\alpha$. Then $\psi_\alpha(x_\alpha)$ defines a quantization of $A_{K,\alpha}$:
\be\label{eq:quantum-curve-alpha}
	\hat A_{K,\alpha} =  \sum_{i,j} c^{(\alpha)}_{i,j}(q,a) \, \hat x_\alpha^i \hat y_\alpha^j\qquad \in \, \Sk_{\mathfrak{gl}_1}(T^2)
\ee
where $\hat x_\alpha,\hat y_\alpha$ are generators of the quantum torus algebra, $\hat y_\alpha \hat x_\alpha = q^{2} \hat x_\alpha\hat y_\alpha$ that quantizes the Poisson algebra of functions on $(\IC^*)^2$, by the requirement that
\be
\hat A_{K,\alpha}\psi_\alpha =0.
\ee 
The coefficients $c^{(\alpha)}_{i,j}(q,a) $ are polynomials in $q,a$ with integer coefficients
with classical limit the augmentation curve \eqref{eq:augmentation-curve}: 
\be
	\lim_{q\to 1} c^{(\alpha)}_{i,j}(q,a) = c_{i,j}(a)\,.
\ee

Since the coordinates $(x_\alpha,y_\alpha)$ correspond to curves on $\Lambda_K$ with intersection number $1$ there is an $SL(2,\mathbb{Z})$ transformation to Seifert coordinates $(x,y)$ and we get on the $\mathfrak{gl}_1$ quantum torus algebra a change of coordinates 
\be\label{eq:local-quantum-coordinates}
\begin{split}
	\hat x_\alpha &= \tilde \xi_x \, \hat x^{s_1'}\hat y^{s_2'}\,,\\
	\hat y_\alpha &= \tilde \xi_y \, \hat x^{-r_2'}\hat y^{r_1'}\,, 
\end{split}
\ee
with $\lim_{q\to 1} \tilde \xi_i = \xi_i$ in \eqref{eq:local-coordinates}.
This allows us to express $\hat A_{K,\alpha}$ in terms of global quantum torus variables $\hat x, \hat y$, and in particular to compare $\hat A_{K,\alpha}(\hat x,\hat y)$ to $\hat A_{K,\beta}(\hat x,\hat y)$ for $\alpha\ne \beta$.
As we shall see, quantizations on different branches generally give \emph{different} operators, but the differences can be understood and the operators should regarded as equivalent in an appropriate sense
\begin{remark}
In previous work, the quantized operator $\hat A_K$ of $A_K$ of one branch has been used unaltered for all other branches, and the corresponding partition functions $\psi_\alpha$ have been obtained by solving the corresponding (adapted) $q$-difference equation. Here we have explicit knowledge of partition functions on two branches and we find that such calculations sometimes involves additional change of coordinates (also for the recusrsion relation) between branches.
\end{remark}

For the conormal branch of $\Sigma_K$ the partition function is given by HOMFLYPT polynomials colored by symmetric partitions
\be\label{eq:conormal-psi-general}
	\psi_{L_K} = \sum_{n\geq 0} H_{(n)}(a,q) \, (x_{l})^n\,.
\ee
The corresponding quantization $\hat A_{K,l}$ therefore coincides with quantum curves studied commonly in the literature on quantum knot invariants \cite{Garoufalidis:2016zhf, 2013NuPhB.867..506F}.
Less is known about curve counts for other Lagrangian fillings. Our main Theorem \ref{t:skeinvalued knot complement} gives $\psi_{M_K}$ and correspondingly $\hat A_{K,m}$ for torus knots.

\subsection{Examples of tropical augmentation curves}\label{sec:tropical-examples}

In this section we give examples of tropical limits of augmentation curves.

\subsubsection{Unknot}
The augmentation curve and its tropicalization are
\be
	\Sigma_{0_1} :\qquad 1-x-y+a^2xy = 0
\ee
\be
	\Trop_{0_1}:\qquad \text{Max}\{0,u,v,t +u+v\}
\ee
where $u=\log|x|,v=\log|y|$ and $t=\log |a^2|$.
The tropical diagram is shown in Figure \ref{fig:unknot-trop}
\begin{figure}[h!]
\begin{center}
\includegraphics[width=0.2\textwidth]{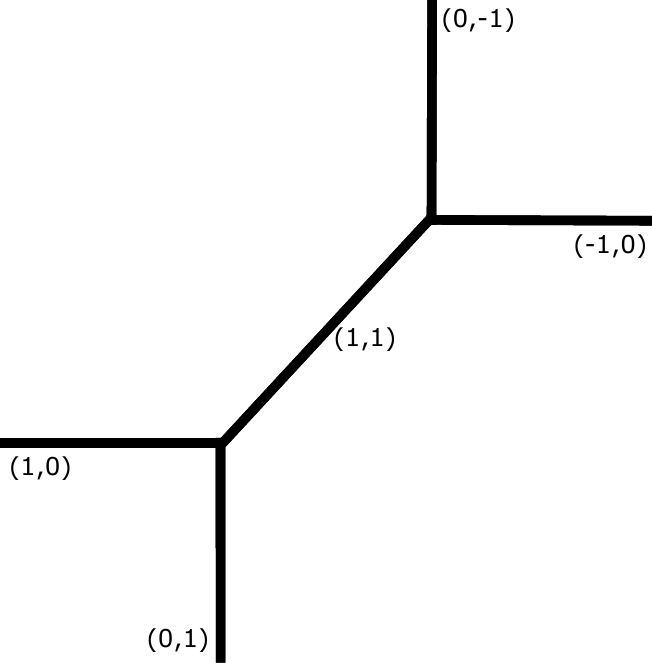}~~~~~~~~~~~~~~~~~~~~~~~~~~~~~~
\includegraphics[width=0.2\textwidth]{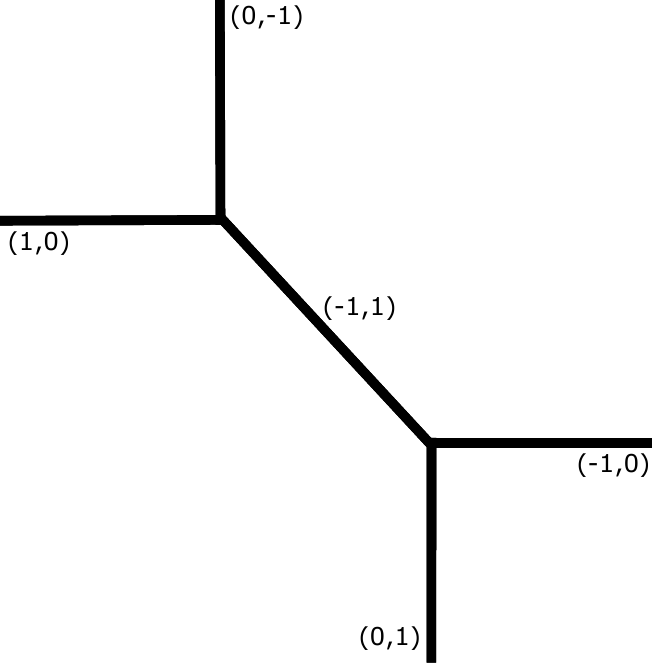}
\caption{Tropical augmentation curve of unknot. The figure on the left is for $|a|<1$, whereas the figure on the right is for $|a|>1$.}
\label{fig:unknot-trop}
\end{center}
\end{figure}

\subsubsection{Reduced trefoil}
The \emph{reduced} augmentation curve for the trefoil knot is (see e.g. \cite{ng2014topological}),
\[
	\Sigma_{3_1} :\qquad  
	1-x-y+xy-2xy^2+2a^2xy^2+a^2xy^3-x^2y^3-a^4xy^4+a^2x^2y^4 = 0.
\]
The corresponding tropical diagram is
\[
	\Trop_{3_1} :\qquad  \text{Max}\{0,u,v,u+v,u+2v,t +u+2v,t+u+3v,2u+3v,2t+u+4v,t+2u+4v\},
\]
where $u=\log|x|,v=\log|y|$ and $t=\log |a^2|$. Its branches and phases are shown in Figure \ref{trefoil-trop-fig-1}.

\begin{figure}[h!]
    \centering
    \includegraphics[scale=0.5]{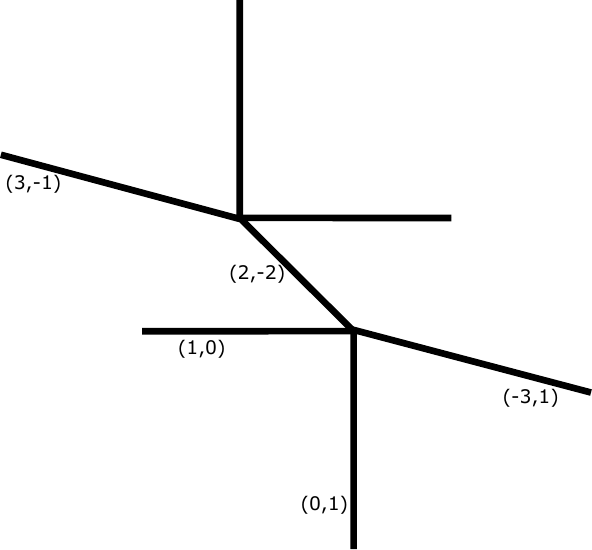}~~~~~~~~~~~~~~~
    \includegraphics[scale=0.5]{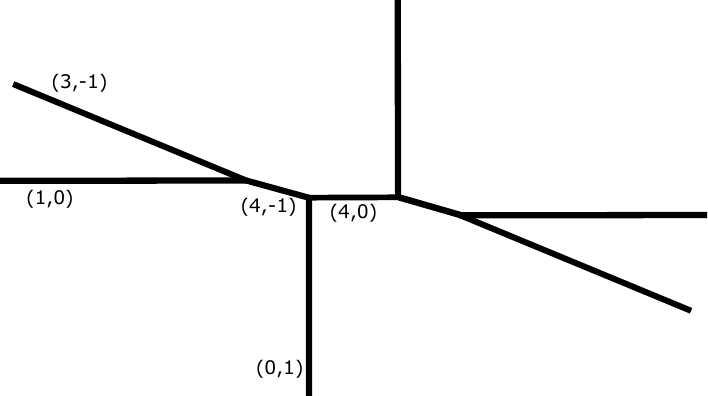}
    \caption{Tropical augmentation curve of trefoil. The figure on the left corresponds to $|a|<1$, and the figure on the right corresponds to $|a|>1$. }\label{trefoil-trop-fig-1}
\end{figure}

\subsubsection{Figure-eight}
The augmentation curve of the figure-eight knot is (see e.g. \cite{klemmetal}), 
\[
	\Sigma_{4_1} :\qquad  
	x-x^2-2a^2xy+2x^2y-a^4y^2+x^3y^2+a^4y^3-a^2x^3y^3+2a^4xy^4-2a^4x^2y^4-a^6xy^5+a^4x^2y^5.
\]
The tropical curve is
\begin{align*}
	\Trop_{4_1}:\qquad \text{Max}\{u,2u,&t+u+v,2u+v,2t+2v, \\
    &3u+2v,2t+3v,t+3u+3v,2t+u+4v,2t+2u+4v,3t+u+5v,2t+2u+5v\},
\end{align*}
where $u=\log|x|,v=\log|y|$ and $t=\log |a^2|$. Its branches and phases are shown in Figure \ref{figure-eight-fig-1}.
\begin{figure}[h!]
    \centering
    \includegraphics[scale=0.45]{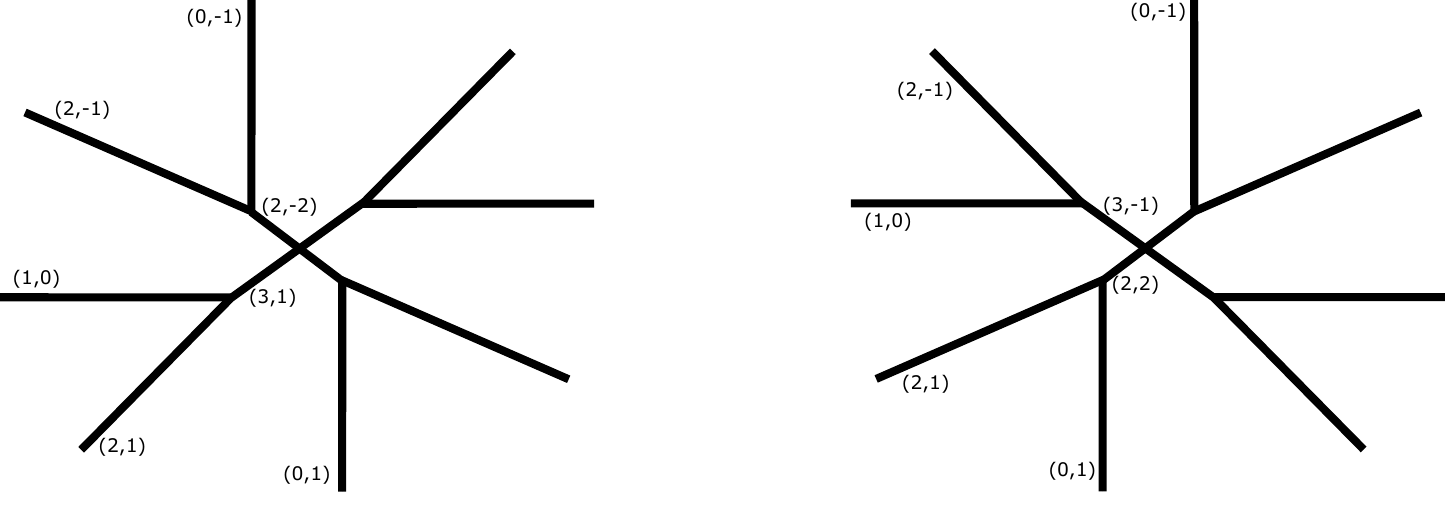}
    \caption{Tropical augmentation curve of $4_1$ knot. The figure on the left shows $|a|<1$, while the figure on the right shows $|a|>1$.}\label{figure-eight-fig-1}
\end{figure}

\section{The $(2,1)$ torus knot -- the twisted unknot}\label{sec:unknot}
In this section we illustrate our formulas for the simplest $(2,2p+1)$-torus knot, setting $p=0$ which gives the twisted unknot.

The augmentation polynomial of the untwisted unknot in Seifert framing and in global Seifert coordinates defined in Section \ref{eq:Seifert-global-coordinates} is
\be\label{eq:unknot-A}
	A_{U,l} = 1-y- x+a^2 xy\,.
\ee
In logarithmic coordinates $u=\log|x|$ and $v=\log|y|$,  the tropical augmentation curve is
\be\label{unknot-trop}
	\Trop_{0_1}: \quad   \text{Max}\{0,u,v,2\log|a|+u+v\}.
\ee
This is shown in Figure \ref{fig:unknot-trop}, with conormal branch labeled by $(1,0)$ and complement branch by $(0,1)$.

\subsection{Conormal branch}

On the conormal branch, $(x_l,y_l)=(x,y)$. 
The conormal partition function is the generating series of symmetrically colored HOMFLYPT polynomials.
We recall that for a generic partition $\nu$ the colored HOMFLYPT of the \emph{untwisted} unknot is
\be
	H_\nu(a,q) = \prod_{\ydiagram{1}\in \nu} \frac{a q^{c(\ydiagram{1})} - a^{-1} q^{-c(\ydiagram{1})} }{q^{h(\ydiagram{1})}-q^{-h(\ydiagram{1})}}
\ee
with $c(\ydiagram{1})$ and $h(\ydiagram{1})$ respectively denoting the content and hooklength of a box in the Young diagram representing $\lambda$.
Specializing to symmetric partitions gives $H_{(r)}(a,q) = a^{-r}q^r (a^2;q^2)_r / (q^2;q^2)_r$. 
Passing to the unknot with a positive or negative twist amounts to rescaling $H_{(r)}$ by $a^{\pm 1}$, which gives the following conormal partition function
\be
	\psi_{L_{U^\pm}}(x,q,a)
	= 
	\sum_{r\geq0}\frac{(a^2;q^2)_r}{(q^2;q^2)_r} \, (a^{-1\pm 1} x)^r
	=
	\frac{(a^{1\pm 1}x;q^2)_\infty}{(a^{-1\pm 1}x;q^2)_\infty}
\ee
For convenience we have rescaled $x$ by $q^{-1}$ and absorbed the factor $q^{r}$.
This defines the following quantizations of the augmentation curve for the unknot
\be
	\hat A_{U^\pm,l} = 1-\hat{y}- a^{-1\pm 1} \hat{x}+a^{1\pm 1}\hat{x}\hat{y}\,.
\ee
Observe that choosing the positive twist corresponds to Seifert global coordinates defined in Section \ref{eq:Seifert-global-coordinates}.
In keeping with our general framework we will instead be interested in the negatively twisted unknot.
The resulting curve is however not expressed in global Seifert coordinates due to the fact that $y\to 0$ would give $x\to a^2$. Therefore Seifert coordinates are obtained by a rescaling of $x$ by $a^2$, which effectively givesthe same curve as for positive twisting. Therefore, although we work with the negatively twisted unknot, our choice of coordinates will force us to consider the curve $\hat A_{U^+,l}$ in either case.

We recall that the conormal partition function admits a quiver presentation involving a disk and an anti-disk (framed disk) without mutual linking \cite{Kucharski:2017ogk}
\be
	\psi_{L_{U^{\pm}}}(x,q,a)=
	\sum_{d_1,d_2\geq 0}
	q^{d_1^2} \frac{(-a^{1\pm 1} q^{-1} x)^{d_1}}{(q^2;q^2)_{d_1}} \frac{(a^{-1\pm 1}x)^{d_2}}{(q^2;q^2)_{d_2}}
\ee

\begin{figure}[h!]
\begin{center}
\includegraphics[width=0.4\textwidth]{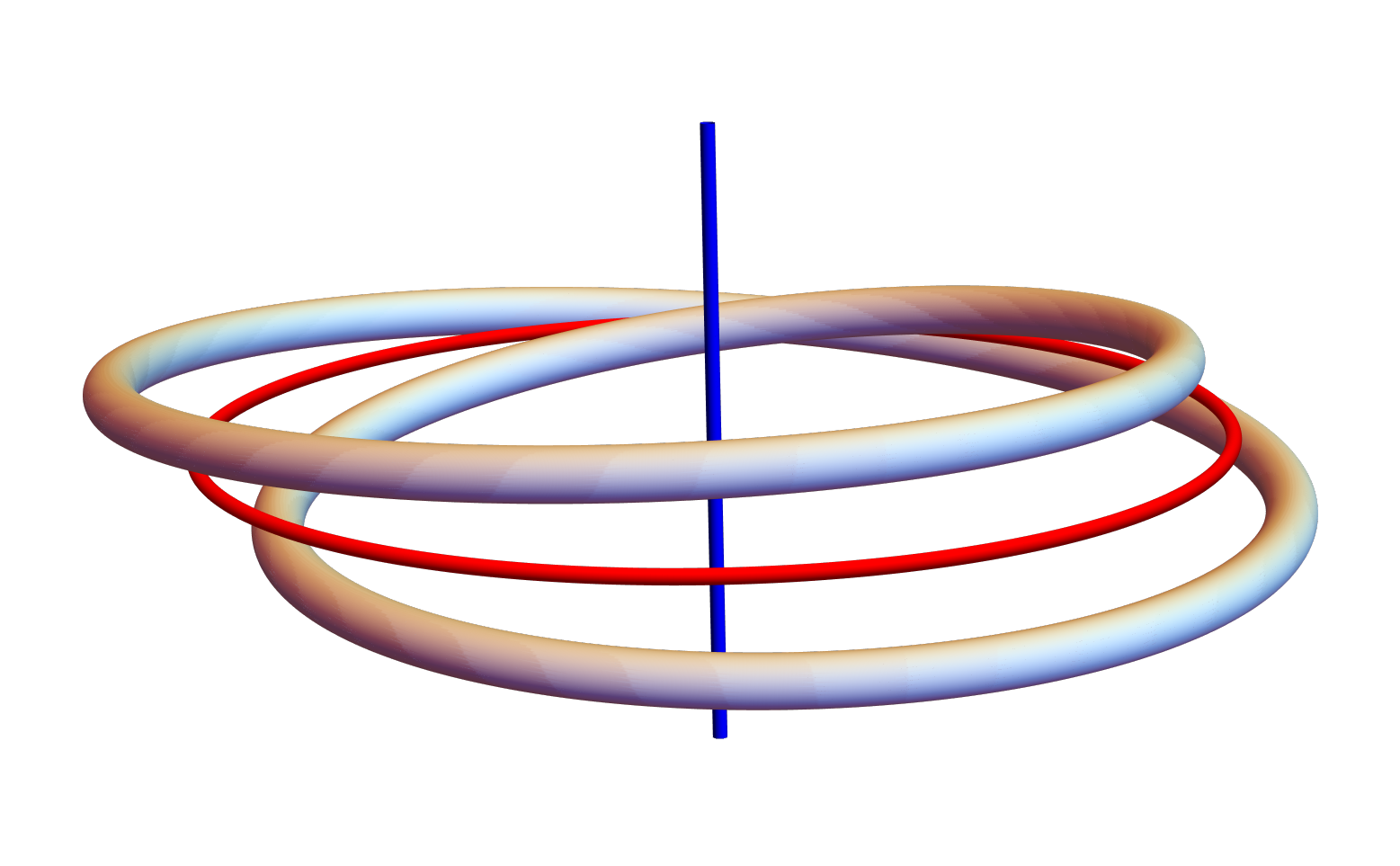}
\caption{The Unknot-Hopf cobordism (negative twist).}
\label{fig:unknot-cobordism}
\end{center}
\end{figure}

\subsection{Complement branch}\label{sec:unknot-complement-branch}
From the above discussion we see that switching between positive and negative twist corresponds to rescaling $x$. 
When working on the complement this rescaling becomes inessential, because oure definition of Seifert global coordinates from Section \ref{eq:Seifert-global-coordinates} anyways requires $x\to 1$ as $y\to 0$.  Thus in both cases, a suitable rescaling would lead to the same classical augmentation curve and same local coordinates. Without loss of generality we will henceforth focus on the positive twist only.

The cobordism between the twisted unknot Legendrian and the Hopf link Legendrian is shown in Figure \ref{fig:unknot-cobordism}.
We consider the partition function \eqref{eq:gl1-formula-2-2p+1-knots} specialized to $p=0$. 
This reads
\be
\begin{split}
	\psi_{M_U}
	& =
	\sum_{n_1, n_2 \geq 0}
	H_{(1)^{n_1},(n_2)}(a,q)
	\,
	q^{- n_2^2 - 2 \, n_1 n_2}
	\,
	(- q \,\mu^2)^{n_1}
	(-  q^{2}\,\mu)^{ n_2}
	\\
	& \mathop{=}^{\eqref{eq:ZMK-final-torus-knots}}
	\sum_{d_1, d_2, d_3, d_4, d_5\geq 0}\sum_{d_6=0}^{1}
	q^{d_1^2
	-d_3^2
	-3 d_5^2
	-2 d_3 d_1
	-2 d_4 d_1
	-2 d_5 d_1
	-3 d_6^2
	-2 d_2 d_3
	-2 d_2 d_4
	-2 d_3 d_4
	-2 d_2 d_5
	-4 d_3 d_5
	-4 d_4 d_5
	-6 d_3 d_6
	-6 d_4 d_6
	-6 d_5 d_6}
	\,
	\\&\qquad
	\qquad
	\qquad\times
	\frac{(-q\mu)^{2 d_1}}{(q^2;q^2)_{d_1}} 
	\frac{(a^2 q^2 \mu^2)^{d_2}}{(q^2;q^2)_{d_2}}
	\frac{(-q^3 \mu)^{d_3}}{(q^2;q^2)_{d_3}} 
	\frac{(a^2 q^2 \mu)^{d_4}}{(q^2;q^2)_{d_4}}
	(-a^2 q^3 \mu^{3})^{d_5} (q^3 \mu^{3})^{d_6}
	\,.
\end{split}
\ee

We have found numerically that this admits the much simpler form
\be
	\psi_{M_U} = \frac{(q^2 y;q^2)_\infty}{(q^2 a^2 y;q^2)_\infty}\,,
\ee
which moreover has a simple quiver description based on a disk and an antidisk
\be
	\psi_{M_U} = \sum_{d_1, d_2\geq 0} q^{d_2^2} \frac{(q^2a^2y)^{d_1}}{(q^2;q^2)_{d_1}}\frac{(-qy)^{d_2}}{(q^2;q^2)_{d_2}}\,.
\ee
From these expressions it is easy to see that we have the following quantization of \eqref{eq:unknot-A} for the unknot complement $M_U$
\be
	\hat A_{U,m} = 1 -  \hat y -   \hat x + q^2a^2 \hat x \hat y\,.
\ee
Comparing with the conormal quantum curve, we find that in the case of the unknot, the complement branch quantization of the augmentation curve agrees with the conormal branch quantization after the change of variables $a\to a q$
\be\label{eq:unknot-conormal-complement-relation}
	\hat A_{U,m}(\hat x, \hat y, a) = \hat A_{U,l}(\hat x,\hat y, qa)\,.
\ee
Note that this is compatible with the $\mathfrak{gl}_1$ skein specialization of the skein-valued unknot recursion for conormal and complement obtained in \cite[Section 6.3]{Ekholm:2024ceb}. The uplift of this formula to the worldsheet skein module would involve shifting $a^2\to a^2 a_M a_L$, corresponding to the fact that any curve that intersects $S^3$ will also intersect $L$ and $M$, since their connected sum is homeomorphic to $S^3$, and since we are counting curves with $M$ very close to the zero section.

Correspondingly, we also find that our result agrees with the $F_K(a,q,x)$ invariant computed in \cite[equation (54)]{Ekholm:2020lqy}, upon making the following variable identifications
\be
	\psi_{M_U}(a,q,y) = F_K(a^2q^2,q^2,y)\,.
\ee

\begin{remark}[$q$-shift from brane recombination] \label{rmk:4-chain-shift}
The shift $a\to aq$ in \eqref{eq:unknot-conormal-complement-relation} has a geometric origin. It arises by viewing the complement Lagrangian as obtained by performing surgery on the conormal and the zero-section of $T^*S^3$, see \cite{Ekholm:2024ceb}.
Recall that the power of $a$ counts intersections with the 4-chain of $S^3$ in $T^\ast S^3$, and that in the $\mathfrak{gl}_1$-skein on $L$ intersections with the 4-chain on $L$ gives powers of $q$ ($a=q^N$ in the $\mathfrak{gl}_N$-specialization of the skein), the shift $a\to q a$ corresponds to the $4$-chain of $M_K$ being a smoothing of the sum of 4-chains of $L_K$ and $S^3$.
\end{remark}

\subsection{The augmentation curve}\label{sec:unknot-aug-curve}
We illustrate the new derivation of the augmentation curve from the cobordism to the Hopf link following Section \ref{sec:aug-curves}.

The twisted Hopf link augmentation variety gets modified by the cobordism to \eqref{eq:tilde-B}, which in this case evaluates to
\be
\begin{split}
	\tilde B_1 & = 1-\eta_1 - \xi_1\eta_1\eta_2^{-1}+a^2 \xi_1\eta_2^{-1} \\
	\tilde B_2 & = 1-\eta_2 - \xi_2\eta_1^{-1}\eta_2^{-1} + a^2 \xi_2  \eta_1^{-1} \\
\end{split}
\ee
while coordinates $(\lambda,\mu)$ for the unknot Legendrian torus holonomies are related to the Hopf holonomies by the specialization of \eqref{eq:xi-mu} and \eqref{eq:lambda-eta}
\be
	\xi_1 = -\mu^2\,,\qquad
	\xi_2 = -\mu\,,\qquad
	\lambda^{-1} = \eta_1^2\eta_2\,.
\ee 
Taken together these reduce to a subvariety of $(\IC^*)^2$ with coordinates $\lambda,\mu$ which contains \eqref{eq:unknot-A} as a factor
\be
	(1 - \mu-\lambda + a^2 \lambda\mu) (\dots) = 0\,.
\ee

\subsection{Internal branch}

The internal diagonal edge shown in Figure \ref{fig:unknot-trop} is of type $(r_1,r_2) = (1,1)\sim (-1,-1)$. There are two branches associated with this edge: one corresponding to the left-bottom vertex (at $x=y=1$) and one associated with the top-right vertex (at $x=y=a^{-2}$). Working with the former, we identify $(r_1', r_2')=(-1,-1)$ and choose $s_1'=-1, s_2'=0$. From \eqref{eq:local-coordinates} we deduce that the appropriate local coordinates are
\be
	x_\alpha = x^{-1}\,, \qquad y_\alpha = -x y^{-1}\,.
\ee

To quantize this branch we only require that the partition function admits a quiver description. The quiver matrix can be deduced from the disk potential. To obtain this we write the classical curve in local coordinates
\be
	A_{U,\alpha} = 1-y_\alpha+x_\alpha y_\alpha -a^2 x_\alpha^{-1}
\ee
which gives 
\be
	W^{U,\alpha} = \int \log y_\alpha\, d\log x_\alpha
	= 
	\Li_{2}(x) + \Li_2(a^2 x^{-1}) \,.
\ee
The same disk potential is given by the semiclassical limit of the following quiver partition function
\be
	\psi_{\text{diagonal}}
	=
	\sum_{d_1, d_2\geq 0} q^{d_1^2 +d_2^2} \frac{(- q^{\beta_1} x_\alpha)^{d_1}}{(q^2;q^2)_{d_1}}\frac{(-a^2 q^{\beta_2} x_\alpha^{-1})^{d_2}}{(q^2;q^2)_{d_2}} = (q^{\beta_1+1} x_\alpha ;q^2)_\infty(q^{\beta_2+1} a^2x_\alpha^{-1} ;q^2)_\infty\,,
\ee
with $\beta_1, \beta_2$ unknowns to be determined (one of them can be absorbed by rescaling $x_\alpha$).
Noting that
\be
	\hat y_\alpha\cdot \psi_{\alpha} = \frac{1-q^{\beta_2-1}a^2 x_\alpha^{-1}}{1-q^{\beta_1+1}x_\alpha}\psi_{\alpha},
\ee
we deduce that the quantum curve on the branch under study is
\be
	\hat A_{U,\alpha} = 1 -  \hat y_\alpha -q^{\beta_2-1}a^2 \hat x_\alpha^{-1} +q^{\beta_1+1}\hat x_\alpha \hat y_\alpha.
\ee
As before, we might absorb one of the $\beta_i$ into a rescaling of $x$. The other remains as an undetermined parameter.

\section{The $(2,3)$ torus knot -- the trefoil}\label{sec:trefoil}

Next we consider the negative trefoil knot $3_1$, corresponding to $p=1$, see Figure \ref{fig:trefoil-cobordism}. 
The full augmentation curve for the negative trefoil knot in Seifert framing is known to be, see e.g. \cite[Section 5]{ng2014topological} (with identifications $\lambda= x,\mu = y a^{2}, U = a^{2}$).
\be\label{eq:31-Aug-curve}
	1-y-x+x y-2 x y^2+2 a^2 x y^2+a^2 x y^3-a^4 x y^4-x^2 y^3+a^2 x^2 y^4 = 0.
\ee
Note that $x,y$ are normalized according to our definition of Seifert global coordinates.
Let us define the curve $A_{3_1}$ by introducing an additional branch as follows
\be\label{eq:31-full-A-class}
	A_{3_1} = (1-a^2 y^2) 
	\left(
	1-y-x+x y-2 x y^2+2 a^2 x y^2+a^2 x y^3-a^4 x y^4-x^2 y^3+a^2 x^2 y^4
	\right),
\ee
motivated by the fact that the quantum curve will necessarily feature the additional terms, as we will see shortly. 

\begin{figure}[h!]
\begin{center}
\includegraphics[width=0.4\textwidth]{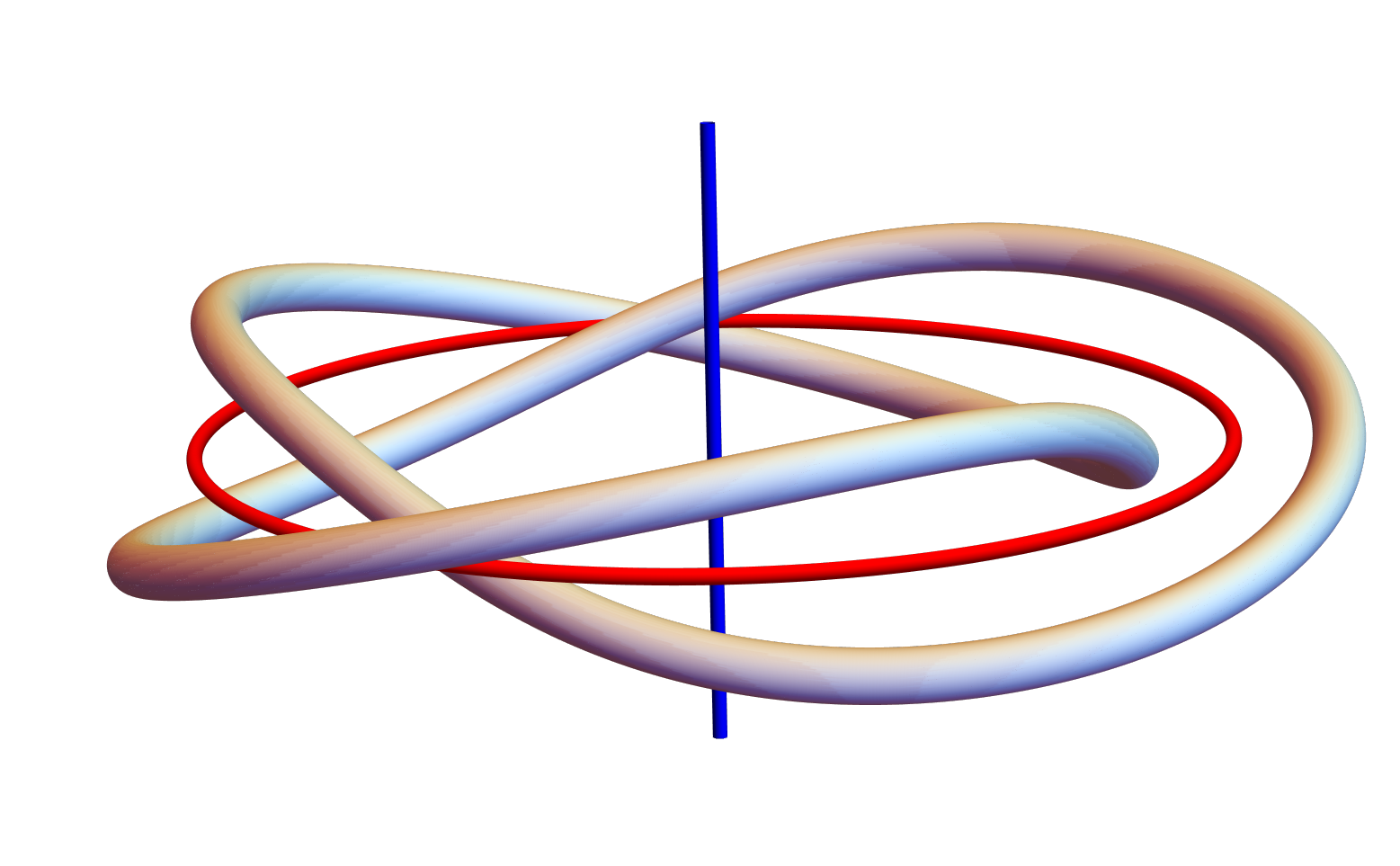}
\caption{The (negative) trefoil knot.}
\label{fig:trefoil-cobordism}
\end{center}
\end{figure}

In logarithmic coordinates $u=\log|x|$ and $v=\log|y|$, the tropical curve is given by
\begin{equation}
	\Trop_{3_1}:\quad \text{Max}\{0,u,v,u+v,u+2v,2\log a+u+2v,2\log a+u+3v,2u+3v,4\log a+u+4v,2\log a+2u+4v\}.
\end{equation}
The corresponding diagram is shown in Figure \ref{trefoil-trop-fig-1}.
There are now three types of external edges: horizontal edges corresponding to the conormal filling, vertical edges corresponding to the complement filling, and a new kind of diagonal edges that in the limit $a=1$ correspond to nontrivial flat connections on the complement branch, see \cite{cornwell}.
Note also that the number of internal edges changes depending on whether $|a|<1$ or $|a|>1$.

\subsection{Conormal branch}
In blackboard framing the symmetrically colored HOMFLYPT polynomials for the negative trefoil are
\be
	H^{\mathrm{bb}}_{(r)}(a,q) = 
	q^{-3r^2} 
	\left(\frac{q^2}{a^2}\right)^r
	\frac{(a^2;q^2)_r}{(q^2;q^2)_r} \, \left[ \sum_{k=0}^{r} \frac{(q^2;q^2)_r}{(q^2;q^2)_k (q^2;q^2)_{r-k}} q^{2k(r+1)} (a^2q^{-2};q^2)_k \right]
\ee

We will work in Seifert framing, which amounts to multiplying $H_r$ by $q^{3 r(r-1)}$ and for convenience we introduce an additional linear rescaling by $(q a^2)^r$
\be
	H_{(r)}(a,q) := q^{3 r(r-1)} \left({q a^2}\right)^r \, H^{\mathrm{bb}}_{(r)}(a,q)
\ee 
The partition function is the generating series of symmetrically colored (unreduced) HOMFLYPT polynomials
\be
\begin{split}
	\psi_{L_K}(x,q,a)
	& = \sum_{r\geq0} H_{(r)}(a,q) x^r \\
	& = \sum_{r\geq 0} \frac{(a^2;q^2)_r}{(q^2;q^2)_r} \, \left[ \sum_{k=0}^{r} \frac{(q^2;q^2)_r}{(q^2;q^2)_k (q^2;q^2)_{r-k}} q^{2k(r+1)} (a^2q^{-2};q^2)_k \right]
	x^r\,.
\end{split}
\ee
This is annihilated by the following quantization of the augmentation curve 
\be\label{tref-con-recur-oper}
\begin{split}
	\hat A_{L_K}
	& =
	1-\hat y-\frac{a^2 \hat y^2}{q^6}+\frac{a^2 \hat y^3}{q^6}
	\\
	&
	+\hat x \left(a^6 \hat y^6-a^4 \left(q^2+1+q^{-2}\right) \hat y^4-a^4 q^2 \hat y^5+a^2  \left(q^2+q^4\right) \hat y^4+a^2 \left(q^2+1+q^{-2}\right) \hat y^2- \left(q^2+q^4\right) \hat y^2+q^2\hat y-1 \right)
	\\
	& 
	+\hat x^2
	\left(- q^8 \hat y^3 + a^2 q^8 \hat y^4 + a^2 q^{14} \hat y^5 - a^4 q^{14} \hat y^6\right)
	\\
\end{split}
\ee
Note that in the classical limit this curve recovers the classical curve \eqref{eq:31-full-A-class}, which includes additional factors compared to the augmentation curve \eqref{eq:31-Aug-curve}.
The conormal partition function admits a quiver presentation involving six disks \cite{Kucharski:2017ogk}
\be
	\psi_{L_K}(x,q,a)=
	\sum_{d_1,\dots, d_6\geq 0}
	q^{
	(C_{L_{3_1}})_{ij} d_i d_j
	}
	\prod_{i=1}^{6} \frac{x_i^{d_i}}{(q^2;q^2)_{d_i}}
\ee
with quiver adjacency matrix
\be
	C_{L_{3_1}} = \left(
\begin{array}{cccccc}
 0 & 0 & 1 & 2 & 1 & 2 \\
 0 & 1 & 1 & 2 & 1 & 2 \\
 1 & 1 & 2 & 2 & 2 & 3 \\
 2 & 2 & 2 & 3 & 2 & 3 \\
 1 & 1 & 2 & 2 & 3 & 3 \\
 2 & 2 & 3 & 3 & 3 & 4 \\
\end{array}
\right)
\ee
and variables
\be
	x_1 =  x\,,\quad
	x_2 = -\frac{a^2}{q}x\,,\quad
	x_3 = q^2 x\,,\quad
	x_4 = -a^2 qx\,,\quad
	x_5 = -\frac{a^2}{q}x\,,\quad
	x_6 = \frac{a^4}{q^2}\,.
\ee

\subsection{Complement branch}\label{sec:trefoil-complement-branch}
The partition function of holomorphic curves is given by \eqref{eq:gl1-formula-2-2p+1-knots} specialized to $p=1$
\be
\begin{split}
	\psi_{M_{3_1}}
	& =
	\sum_{n_1, n_2 \geq 0}
	H_{(1)^{n_1},(n_2)}(a,q)
	\,
	q^{-3 n_2^2 -2 \, n_1 n_2}
	\,
	(-q \,\mu^2)^{n_1}
	(- q^{4}\,\mu^{3})^{ n_2}
\,.
\end{split}
\ee
The corresponding quiver description is given by \eqref{eq:ZMK-final-torus-knots}.
Basic curves include two disks wrapping twice around the meridian of the knot, two disks wrapping three times the meridian, an annulus and an anti-annulus whose boundaries wrap overall five times. Therefore, the complement branch partition function in quiver representation is
\be\label{eq:complement-psi-31}
\begin{split}
	\psi_{M_{3_1}}
	&=
	\sum_{d_1,\dots, d_5\geq 0}\sum_{d_6=0}^{1}
	q^{
	(C_{M_{3_1}})_{ij} d_i d_j
	}\frac{\prod_{i=1}^{6}y_i^{d_i}}{\prod_{i=1}^{4}(q^2;q^2)_{d_i}}
\end{split}
\ee
with quiver adjacency matrix
\be
	C_{M_{3_1}} =
	\left(
\begin{array}{cccccc}
 1 & 0 & -1 & -1 & -1 & 0 \\
 0 & 0 & -1 & -1 & -1 & 0 \\
 -1 & -1 & -3 & -3 & -4 & -5 \\
 -1 & -1 & -3 & -2 & -4 & -5 \\
 -1 & -1 & -4 & -4 & -5 & -5 \\
 0 & 0 & -5 & -5 & -5 & -5 \\
\end{array}
\right)
\ee
and variables
\be
	y_1 = -q y^2\,,\quad
	y_2 = a^2 q^2 y^2\,,\quad
	y_3 = -q^5 y^3\,,\quad
	y_4 = a^2 q^4 y^3\,,\quad
	y_5 = -a^2 q^5 y^5\,,\quad
	y_6 = q^5y^5\,.
\ee
The partition function determines (up to the freedom to rescale $\mu$) the following quantization of $A_{3_1}$
\be\label{eq:complement-curve-result}
\begin{split}
	\hat A_{M_K}
	& = 
	\frac{(\hat y -1) \left(a^2 \hat y ^2-q^4\right) }{q^4}
	\\
	&
	+\left(-\frac{a^2 \hat y^2 \left(a^2 \hat y^3+\left(a^2-1\right) \hat y^2-1\right)}{q^4}-\frac{\left(a^2-1\right) \hat y^2 \left(a^2 \hat y^2-1\right)}{q^2}+\left(a^2-1\right) \hat y^2+\frac{a^4 \hat y^4 \left(a^2 \hat y^2-1\right)}{q^6}+\hat y-1 
	\right)\hat x
	\\
	&+ 
	\frac{\hat y ^3 \left(a^2 \hat y ^2-1\right) \left(q^2-a^2 \hat y \right)}{q^6}
	\hat x^2
\end{split}
\ee
where operators $\hat x$, $\hat y$ act by
\be\label{eq:trefoil-complement-polarization}
    \hat xf(y)=f(q^{-2} y)\,,\qquad
    \hat yf(y)=yf(y)\,.
\ee
Comparing the conormal and complement recursion operators, namely equations (\ref{tref-con-recur-oper}) and (\ref{eq:complement-curve-result}), we see that $\hat A_{M_K}$ is obtained from $\hat A_{L_K}$ by simultaneously rescaling $\hat x$ and $a$ 
\be\label{eq:conormal-complement-shifts-trefoil}
	\hat A_{M_{3_1}}(\hat x,\hat y,a) = \hat A_{L_{3_1}} (\hat x,\hat y,qa)\,.
\ee
We observe the appearance of a shift of $a$ by $q$, just as in the case of the unknot \eqref{eq:unknot-conormal-complement-relation}. Its origin can be understood by brane recombination as explained in Remark \ref{rmk:4-chain-shift}.

\subsection{The augmentation curve}\label{sec:trefoil-aug-curve}

Before proceeding to other branches, we take a break to discuss how the augmentation curve of the trefoil can be obtained from that of the twisted Hopf link via the cobordism argument of Section \ref{sec:aug-curves}. 
With $p=1$ the twisted Hopf link augmentation variety \eqref{eq:tilde-B} is described by
\be
\begin{split}
	\tilde B_1 & = 1-\eta_1 - \xi_1\eta_1\eta_2^{-1} + a^2 \xi_1\eta_2^{-1} \\
	\tilde B_2 & = 1-\eta_2 - \xi_2\eta_1^{-1}\eta_2^{-3} + a^2 \xi_2  \eta_1^{-1}\eta_2^{-2} \\
\end{split}
\ee
while coordinates $(\lambda,\mu)$ for the unknot Legendrian torus holonomies are now related to the Hopf holonomies by the specialization of \eqref{eq:xi-mu}-\eqref{eq:lambda-eta}
\be\label{eq:xi-mu-trefoil}
	\xi_1 = -\mu^2\,,\qquad
	\xi_2 = -\mu^3\,,\qquad
	\lambda^{-1} = \eta_1^2\eta_2^3\,.
\ee 
Taken together these reduce to a subvariety of $(\IC^*)^2$ with coordinates $\lambda,\mu$ which contains \eqref{eq:31-Aug-curve} as a factor
\be
	(1 -\mu-\lambda+\lambda  \mu-2 \lambda  \mu ^2  
	+2 a^2 \lambda  \mu ^2
	+a^2 \lambda  \mu ^3
	-a^4 \lambda  \mu ^4 -\lambda ^2 \mu ^3 +a^2 \lambda ^2 \mu ^4
	) (\dots) = 0\,.
\ee

\subsection{External diagonal branch}

We now consider the diagonal branch, corresponding to the lower-right external diagonal leg of the tropical diagram in Figure \ref{trefoil-trop-fig-1}.
Here we do not have any geometric interpretation as the moduli space of $A$-branes and cannot compute the partition function $\psi_{\delta}$ directly. We can derive possible partition function as follows:

We first compute the classical local coordinates $(x_\delta, y_\delta)$.
Starting from $(x,y) = (x_{m}, y_{m})$ we change framing to $x=x'(-y')^{-3}, y=y'$. At this point the diagonal branch becomes of slope $(0,1)$, and we can identify $x_{m}=y'$, $y_{m}=(x')^{-1}$.
At the quantum level, the change of framing is implemented by $\hat x = (-q^{-1} \hat y')^{-3} \hat x'$, $\hat y = \hat y'$ while the Fourier transform is  $\hat x''=\hat y'$, $\hat y''=(\hat x')^{-1}$. Finally, we define $\hat x_{\delta} = \hat x''$ and $\hat y_{\delta} = q^{-5} \hat y''$.
Applying these to the complement quantum curve \eqref{eq:complement-curve-result}
\be
	\hat A_{M_K} = f_0(a,q,\hat y) + f_1(a,q,\hat y) \hat x +  f_2(a,q,\hat y)\hat x^2
\ee
we obtain
\be\label{eq:A-hat-diagonal}
\begin{split}
	\hat A_{\delta} 
	& =  f_2(a,q,q^4 \hat x_\alpha) q^{-8}  \hat x_\alpha^{-3}
	+ f_1(a,q,q^4 \hat x_\alpha) \hat y_\alpha
	+ f_0(a,q,q^4 \hat x_\alpha) q^{14}  \hat x_\alpha^{3}\hat y_\alpha^2
	\\
	& = 1-a^2 q^2 {\hat x_\alpha}-a^2 q^8 {\hat x_\alpha}^2+ a^4 q^{10} {\hat x_\alpha}^3
	\\
	& +\big[
	a^6 q^{18} {\hat x_\alpha}^6-a^4 q^{16} {\hat x_\alpha}^5-a^4 \left(q^4+q^2+1\right) q^{10} {\hat x_\alpha}^4+a^2 \left(q^4+q^2+1\right) q^4 {\hat x_\alpha}^2\\
	& +a^2 \left(q^6+q^4\right) q^8 {\hat x_\alpha}^4+q^4 {\hat x_\alpha}-\left(q^2+1\right) q^6 {\hat x_\alpha}^2-1
	\big]\hat y_\alpha
	\\
	& +\left(
	-a^2 q^{22} {\hat x_\alpha}^6+a^2 q^{18} {\hat x_\alpha}^5+q^{18} {\hat x_\alpha}^4-q^{14} {\hat x_\alpha}^3
	\right)\hat y_\alpha^2
\end{split}
\ee
where $\hat x_\delta,\hat y_\delta$ denote generators of the quantum torus adapted to the diagonal branch.

The partition function annihilated by this curve has the following quiver structure
\be\label{eq:Z-diagonal}
	\psi_{\text{diagonal}}(x_\alpha)
	=
	\sum_{d_1,\dots, d_5\geq 0}
	q^{(C_{\text{diagonal}})_{ij} d_i d_j}
	\prod_{i=1}^{5} \frac{x_i^{d_i}}{(q^2;q^2)_{d_i}}
\ee
where the quiver adjacency matrix is
\be
	C_{\text{diagonal}}=\left(
\begin{array}{ccccc}
 0 & 0 & 1 & 1 & 0 \\
 0 & 1 & 1 & 0 & 0 \\
 1 & 1 & 2 & 0 & 0 \\
 1 & 0 & 0 & 0 & 0 \\
 0 & 0 & 0 & 0 & 1 \\
\end{array}
\right)
\ee
and quiver variables are
\be
	x_1 = a^2 q^2 x_\alpha \,,\quad
	x_2 =  -q^3 x_\alpha\,,\quad
	x_3 = q^4 x_\alpha^2 \,,\quad
	x_4 = q^2 x_\alpha \,,\quad
	x_5 = -q x_\alpha \,.
\ee
This quiver corresponds formally to five basic disks, four of which are degree 1 in $x_\alpha$, and one (the third) of degree 2 in the homology generator variable $x_\alpha$. 
The linking of disk 3 with the others is not an integer multiple of two, suggesting a (formal) separtion of its two local boundary strands.
Also worth noting is that disks 4 and 5 form a quasi-cancelling pair: they would cancel exactly except for the fact that disk 4 links once with disk 1, which disk 5 does not.

\subsection{Internal diagonal branch}\label{trefoil-int-bra}

Finally we consider the internal diagonal leg of the tropical diagram in Figure \ref{trefoil-trop-fig-1} for $|a|<1$.
In fact we will restrict to the extremal curve obtained by setting $a=0$ in \eqref{eq:31-full-A-class}, which gives
\be
	1-y-x+xy-2xy^2-x^2y^3 =0 \,,
\ee
whose tropical limit is described by
\be
	\text{Max}\{0,u,v,u+v,u+2v,2u+3v\}.
\ee
The tropical diagram diagram is shown in Figure \ref{fig:trefoil-phases-extremal}, where the internal diagonal edge is now the top-left external diagonal edge, which is of type $(r_1, r_2)=(2,-2)$. This is therefore an example of \emph{orbifold edge} of order $2$. 

\begin{figure}[h!]
    \centering
    \includegraphics[scale=0.6]{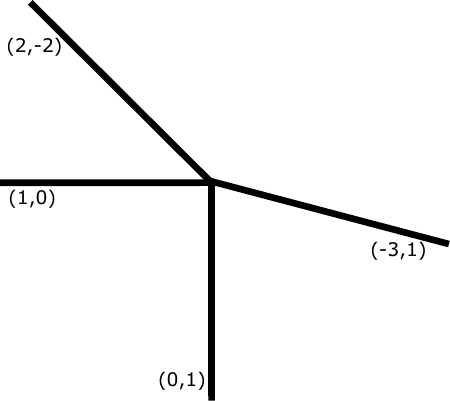}
    \caption{Tropical curve of trefoil in extremal limit}\label{fig:trefoil-phases-extremal}
    \label{fig:trefoil-bottom-phases}
\end{figure}

The classical local coordinates are obtained from \eqref{eq:local-coordinates} with $(r_1', r_2') = (1,-1)$ and $(s_1', s_2')=(0,-1)$ which gives
\be
	x_{\alpha} = y^{-1}\,,\qquad y_\alpha = -xy\,.
\ee
In terms of these, the classical (extremal) curve takes the following form
\be\label{mid-branch-curve}
	(1-y_\alpha)^2=x_\alpha(1-y_\alpha+x_\alpha y_\alpha).
\ee

The disk potential computed from (\ref{mid-branch-curve}) is then
\be
\begin{split}
	W_{\pm}(x_\alpha) 
	& = 
	\int\log y_{\alpha} \, d\log x_\alpha
	\\
	& =\frac{1}{2} \left(-1\pm\sqrt{5}\right) x_\alpha
	+\frac{1}{40} \left(-5\pm3 \sqrt{5}\right) x_\alpha^2
	+\frac{1}{450} \left(-25\pm11 \sqrt{5}\right) x_\alpha^3+O\left(x_\alpha^4\right)\,.
\end{split}
\ee
There are two interesting properties of these expressions, both arising from the fact that the internal branch has orbifold index $2$.
On the one hand, denoting by $y_\pm(x_\alpha)$ the two roots of \eqref{mid-branch-curve}, their product is
\be
	y_+ y_- = 1-x_\alpha
\ee
and therefore
\be\label{eq:orbifold-disk-potentials}
	W_++W_-=-\text{Li}_2(x_\alpha)\,.
\ee
On the other hand, the orbifold potentials $W_\pm$ have irrational coefficients. This is in contrast to what would be expected from usual (open) Gromov-Witten theory, as anticipated in Remark \ref{rmk:irrational-OGW}.

To see the origin of irrationality, we consider a one-parameter deformation defined by
\be
	(1-y_\alpha)^2=x_\alpha(1-y_\alpha+\beta x_\alpha y_\alpha)\,,
\ee
whose tropicalization is
\begin{equation}\label{eq:tropical-bottom-trefoil-deformed}
    \text{Max}\{0,u,v,u+v,2v,2u+v+\log\beta\}\,.
\end{equation}
The tropical diagram for this tropical variety is shown in Figure \ref{Fig-6}.

\begin{figure}[h!]
    \centering
    \includegraphics[scale=0.5]{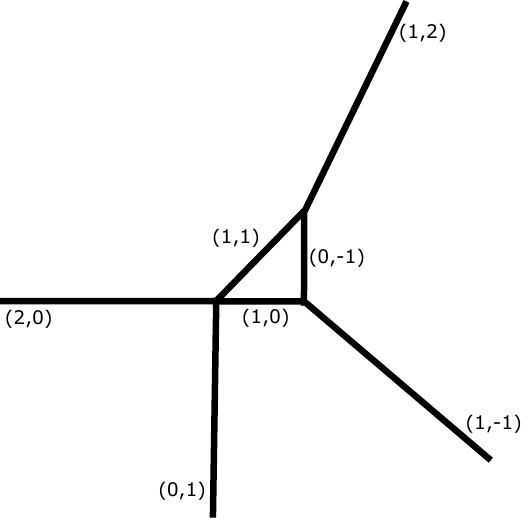}~~~~~~~~~~~~~~~~~~~~~~~~~~~~
    \includegraphics[scale=0.5]{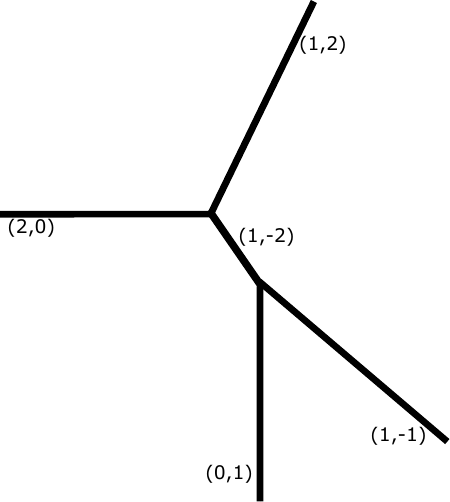}
    \caption{Tropical diagram for \eqref{eq:tropical-bottom-trefoil-deformed}. The horizontal branch is the internal diagonal $(2,-2)$ branch of the extemal trefoil tropical diagram in Figure \ref{fig:trefoil-bottom-phases}. 
    On the left: deformation with $\beta < 1$. On the right: $\beta > 1$.}\label{Fig-6}
\end{figure}

We then consider the $\beta$ deformation of the disk potential
\be
\begin{split}
	W_{\pm}
	&=\frac{1}{2}\left(-1\mp\sqrt{1+4\beta}\right)x_\alpha
	+\left(\frac{\mp(1+2\beta)-\sqrt{1+4\beta}}{8\sqrt{1+4\beta}}\right)x_\alpha^2\nonumber\\ \nonumber\\
	&+\left(\frac{\mp(-1-6\beta-6\beta^2+2\beta^3)-\sqrt{1+4\beta}-4\beta\sqrt{1+4\beta}}{18(1+4\beta)^{3/2}}\right)x_\alpha^3+\mathcal{O}(x_\alpha)^4\,.\\
\end{split}
\ee
This is still an irrational function of moduli. However, viewing $\beta$ as a closed string parameter, the corresponding large-volume expansion near $\beta=0$ features rational coefficients
\be
\begin{split}
    	W_{+}
	&=
	\left(-1-\beta +\beta ^2-2 \beta ^3+5 \beta ^4-14 \beta ^5+42 \beta ^6-132 \beta ^7+429 \beta ^8-1430 \beta ^9+4862 \beta ^{10}+O\left(\beta ^{11}\right)\right)x_\alpha\\
	&+\left(-\frac{1}{4}-\frac{\beta ^2}{4}+\beta ^3-\frac{15 \beta ^4}{4}+14 \beta ^5-\frac{105 \beta ^6}{2}+198 \beta ^7-\frac{3003 \beta ^8}{4}+2860 \beta ^9-\frac{21879 \beta ^{10}}{2}+O\left(\beta ^{11}\right)\right)x_\alpha^2\\
	&+\left(-\frac{1}{9}-\frac{\beta ^3}{9}+\beta ^4-6 \beta ^5+\frac{280 \beta ^6}{9}-150 \beta ^7+693 \beta ^8-\frac{28028 \beta ^9}{9}+13728 \beta ^{10}+O\left(\beta ^{11}\right)\right)x_\alpha^3\\
	& + O(x_\alpha^4)
\end{split}
\ee
and with $W_-$ again given by \eqref{eq:orbifold-disk-potentials}.

With these rational-valued coefficients of the disk potential we get the following open string interpretation. Observe that $W_+$ admits a quiver description with corresponding partition function
\be
	\psi_{\text{internal}}^{a=0}
	= 
	\sum_{d_1,d_2,d_3\geq 0} 
	q^{2d_2^2-d_3^2+2d_1d_3}
	\frac{(q^{\delta_1}x_\alpha)^{d_1}}{(q^2;q^2)_{r_1}}\frac{(q^{\delta_2}\beta)^{d_2}}{(q^2;q^2)_{d_2}}\frac{(-q^{\delta_3}\beta)^{d_3}}{(q^2;q^2)_{d_3}}
\ee
where $\delta_i$ are $q$-charges of the three basic disks that cannot be determined by the disk potential.
This quiver consists of three disks without any linking. Two of the disks have no $x_\alpha$-dependence, and as a consequence the coefficients of $x_\alpha^{r}$ correspond to infinite power series in $\beta$.
 
At the quantum level, the partition function is annihilated by the following operator
\be
	\hat A_{\text{internal}}^{a=0}
	= 1- (1+q^{-2})\hat y_\alpha+ q^{-2} \hat y_\alpha^2 - q^{\delta_1} \hat x_\alpha (1-\hat y_\alpha + q^{1+\delta_1+\delta_3}\hat x_\alpha \hat y_\alpha)
\ee
Note that this holds for arbitrary $\delta_1,\delta_2,\delta_3$. Moreover it does not depend on $\delta_2$. In fact the partition function is not entirely determined by this, but one also needs to impose the closed-string moduli $q$-difference equations, i.e. appropriate $q$-difference equations involving shifts of $\beta$.

\section{Comments on the relation to $\widehat Z$-invariants for knot complements}\label{eq:relation-to-Zhat}

In this concluding section we remark on relations to work on $\widehat Z$ invariants for knot complements, with particular attention to the conjectural interpretation of the latter in terms of curve counts on $M_K$.

\subsection{Curve counting on geometric branches and Chern-Simons theory}

The augmentation curve of any fibered knot admits at least two `geometric' branches, i.e. patches that admit a geometric interpretation as moduli spaces of some Lagrangian $A$-brane: the conormal branch parameterizing the conormal Lagrangian $L_K$ and the complement branch parameterizing the knot complement Lagrangian $M_K$.

The $\mathfrak{gl}_1$-skein valued partition function $\psi_\alpha$ for a geometric branch counts holomorphic curves with boundary on the corresponding Lagrangian $L_\alpha$ in the resolved conifold.
Moreover, by invariance under conifold transition, this also coincides
with the open string partition function in $T^*S^3$ 
\be\label{eq:MK-TS}
	\psi_\alpha = \psi_{\text{top}}(L_\alpha, T^*S^3)\,.
\ee

When $L_\alpha$ is the knot complement Lagrangian $M_K$, the partition function $\psi_{M_K}$ counts worldsheet instantons with boundaries on a single brane on $M_K$ in the background of $N = \log_q a$ branes on the zero section $S^3$.
Moreover, by the topological gauge/string duality of \cite{Witten:1992fb}, $\psi_{M_K}$ corresponds to a partition function of $U(1)$ Chern-Simons QFT on $M_K$ with Wilson lines insertions in correspondence of worldsheet boundaries.

Chern-Simons theory on $M_K$ plays a prominent role in the construction of a three-variable series $F_K(a,q,y)$ for knot complements originally introduced for $a=q^2$ in \cite{gukov2021two} as a $\widehat Z$ invariant for manifolds with boundary \cite{Gukov:2016njj, Gukov:2016gkn, Gukov:2017kmk}, and later generalized to arbitrary $N$ in \cite{Park:2019xey, Park:2020edg, Kucharski:2020rsp,Ekholm:2020lqy, Park:2021ufu,Ekholm:2021irc}.

\subsection{Fourier-Laplace transform approach to computing $F_K$}

The working definition of $F_K$ (with a suitable normalization) adopted in the literature is the Fourier-Laplace transform of the conormal partition function \eqref{eq:conormal-psi-general}
\be
	\psi_{L_K}(a,q,x) \mathop{\longleftrightarrow}^{\text{Fourier}} F_K(a,q,y)\,.
\ee
This is usually defined by writing $H_{(r)}(a,q)$ as a function of $a,q,q^{2r}$ and by replacing $q^{2r} \to y$.
By construction, the resulting power series in $y$ is annihilated by the \emph{conormal} quantum curve 
\be\label{eq:A-FK}
	\hat A_{\text{conormal}}\cdot F_K(a,q,y) = 0\,,
\ee
with the proviso that the $\mathfrak{gl_1}$ torus skein algebra acts by the Fourier-dual representation
\be
	\hat x f(y) = f(q^{-2} y)\,,\qquad
	\hat y \, f(y) = y\, f(y)\,.
\ee

A geometric interpretation of $\widehat Z$ was conjectured in \cite[Section 2.9]{Gukov:2017kmk}, where it was proposed that it should be related to the partition function of open Gromov-Witten invariants of a 3-manifold in its own cotangent bundle $M\subset T^*M$.
If we take $M=M_K$ but replace $T^*M$ with $T^*S^3$, by viewing $M_K\subset S^3$ as a submanifold,
we arrive at the following conjectural interpretation for $F_K$
\be\label{eq:FK-interpretation}
	F_K \mathop{=}^{?}  \psi_{\text{top}}(M_K, T^*S^3)\,.
\ee
By virtue of \eqref{eq:MK-TS}, the conjecture holds if $F_K = \psi_{M_K}$. Moreover, due to \eqref{eq:A-FK} this will be true when the worldsheet quantizations of the agumentation curve on complement and conormal branches agree 
\be
	\hat A_{l} \dot{=} \hat A_{m}\,.
\ee
Here $\dot{=}$ denotes equality up to certain $q$-shifts arising in the context of worldsheet skein theory \cite{Ekholm:2024ceb}, such as those encountered in Remark \ref{rmk:4-chain-shift}.
For this reason, the question of global consistency of quantization discussed in Section \ref{sec:globalp-consistency} is directly related to providing an enumerative interpretation of the $F_K$ invariant.

Our results for the unknot and the trefoil, and more generally, $T_{2,2p+1}$ torus knots indicate that \eqref{eq:FK-interpretation} indeed holds after small adjustments. We next discuss these adjustments.

The first concerns the usual definition of $F_K$ via Fourier-Laplace transform on the colored HOMFLYPT generating series. Observe that this procedure is sensitive to a choice of framing for the conormal branch. A change of framing takes $x\to x (-y)^f$, thereby changing $H_{(r)}$ as a function of $q^{r}$.
From the viewpoint of the tropical diagram, a change of framing can take any non-horizontal edge (i.e. type $(r_1, r_2)$ with $r_2\neq 0$) into a vertical edge (i.e. of type $(0,r_2)$) for a suitable choice of $f$, see Figure \ref{fig:framing-example}.
Since the notion of vertical branch is framing-dependent, the definition of $F_K$ via Fourier-Laplace transform, which relates a horizontal branch to a vertical one, is framing-dependent. As explained in Section \ref{sec:local-moduli}, Seifert framing ensures that the complement branch is vertical, and therefore provides the appropriate framing for computing $F_K$ via Fourier-Laplace transform on HOMFLYPT generating series.

Next, we discuss how the partition function of curve counts on the trefoil complement \eqref{eq:complement-psi-31} compares to $F_{3_1}$ invariant computed in \cite[Section 5.2]{Ekholm:2020lqy}. 
These two are clearly different. Part of this difference stems from the reasons discussed above. Another important difference lies in the fact that \cite{Ekholm:2020lqy} uses the quantum curve from \cite[eq. (2.24)]{2013NuPhB.867..506F}, which is adapted (after setting $t=-1$) to \emph{reduced} normalization of the HOMFLYPT polynomial. At the classical level, the relation between our curve \eqref{eq:31-full-A-class} and the one of \cite{2013NuPhB.867..506F} is obtained by substituting $x\to (1-a x) (1-x)^{-1} a^{-1} y$, $y\to a^{-1} x^{-1}$ and $a\to a^{1/2}$ and multiplying by an overall factor $a^4 x^6(x-1)(1-ax)^{-1}(1-ax^2)^{-1}$. This gives in fact the curve of \cite[eq. (2.25)]{2013NuPhB.867..506F} with $t=-1$. The rational factors in the substitution of the $x$-variable implement the change from unreduced to reduced normalization. The partition function of holomorphic curves $\psi_{m_{3_1}}$ corresponds to the generating series annihilated by the unreduced curve, while $F_{3_1}$ computed in \cite{Ekholm:2020lqy} was annihilated by the reduced curve.

Similar remarks can be applied to the comparison between the partition function of curve counts on $T_{2,2p+1}$ torus knots \eqref{eq:ZMK-final-torus-knots}, and the $F_K$ invariants computed in \cite[eq. (78)]{Ekholm:2020lqy}.

\section*{Acknowledgements}
The work of S.C. is supported by the project “Quiver Structures in Knot Invariants from String Theory and
Enumerative Geometry”, funded by Vergstiftelsen and the Swedish Research Council (VR 2022-06593),
through the Centre of Excellence in Geometry and Physics at Uppsala University. T.E. is supported by the Knut and Alice Wallenberg Foundation, KAW2020.0307 Wallenberg Scholar and
by the Swedish Research Council, VR 2022-06593, Centre of Excellence in Geometry and Physics at Uppsala
University and VR 2024-04417, project grant.
The work of P.L. is supported by the Knut and Alice Wallenberg Foundation, KAW2020.0307 Wallenberg
Scholar and by the Swedish Research Council, VR 2022-06593, Centre of Excellence in Geometry and Physics
at Uppsala University. S.C. is grateful to Lukas Nakamura for discussions.

\appendix

\section{Quiver linking cobordism operator}\label{app:cobordism-quivers}

In this appendix we provide a general description of a cobordism operator $\hat\CO$, relating curve counts encoded by two quivers with the same set of basic curves but different linking among their boundaries.

\subsection{General construction}\label{app:}

We consider a quiver $Q$ with $K$ vertices and adjacency matrix $C$. The partition function defines an element of the skein module of a collection of $K$ solid tori
\be
	Z = \sum_{\bd\in \IZ_{\geq 0}^{\times K}} q^{\bd^t\cdot C\cdot \bd} \prod_{i=1}^{K} \frac{x_i^{d_i}}{(q^2;q^2)_{d_i}}
	\quad\in\quad \Sk_{\mathfrak{gl}_1}(\bigsqcup_{j=1}^{K}S^1\times D^2).
\ee
We wish to build a cobordism operator $\hat\CO$ that relates $Q$ to another quiver $Q$' with the same set of basic curves, but different matrix $C'$
\be\label{eq:cobord-Z}
	Z = \hat\CO\cdot Z'
\ee
Here both $Z, Z'$ belong to $ \Sk_{\mathfrak{gl}_1}(\bigsqcup_{j=1}^{K} S^1\times D^2)$ and $\hat\CO$ belongs to the corresponding skein algebra
\be
	\hat\CO \in \Sk_{\mathfrak{gl}_1}(\bigsqcup_{i=1}^K T^2)\,.
\ee

A systematic construction of $\hat\CO$ can be obtained by recalling that $Z$ (and similarly $Z'$) can be expressed as an operator in the $\mathfrak{gl}_1$ skein algebra of $\bigsqcup_{i=1}^K T^2$ as follows \cite{Ekholm:2019lmb}. Introduce the basis of longitudes and meridians
\be
	\hat x_i, \hat y_i \in \Sk_{\mathfrak{gl}_1}(\bigsqcup_{i=1}^K T^2)
\ee
obeying the quantum torus algebra relations
\be
	\hat y_i \hat x_j = q^{2\delta_{ij}} \hat x_j\hat y_i \qquad
	[\hat x_i,\hat x_j]=[\hat y_i,\hat y_j]=0\qquad i,j\in\{1,\dots,K\}\,.
\ee
Define the variables
\be
	X_i = (-1)^{C_{ii}} q^{C_{ii}-1}\hat x_i\hat y_i^{C_{ii}} \prod_{j<i} \hat y_j^{C_{ij}}
\ee
and the operator
\be\label{eq:quiver-operator}
	\hat \psi := (q X_K;q^2)^{-1}_{\infty}\dots (q X_1;q^2)^{-1}_{\infty}.
\ee
The partition function is obtained by acting with $\hat \psi$ on the (normalized) vacuum of the skein module of $\bigsqcup_{i=1}^K (S^1\times D^2)$, or equivalently by \emph{normal ordering} (see \cite{Ekholm:2019lmb})
\be
	Z =\ : \hat\psi :
\ee

This acquires a geometric interpretation if the different components of $\bigsqcup_{i=1}^K (S^1\times D^2)$ are embedded into a common 3-manifold, such as a conormal solid torus \cite{Ekholm:2018eee}, where longitudes of $S^1\times D^2$ components are mutually linking as encoded by the quiver matrix $C$.
The operator $\hat\psi$ can then be viewed as $K$ insertions of multicoverings of basic holomorphic disks in the empty skein, creating the partition function one basic curve at a time, with linking data encoded by the quiver adjacency matrix $C$ taken into account via the quantum torus algebra relations.

\begin{remark}
This definition of $\hat\psi$ extends in a straightforward way to generalized quivers with disks, annuli and other types of generators associated to vertices, by replacing each of the factors $(q X;q^2)_\infty^{-1}$ with the appropriate multi-covering partition function for the given curve. For example, the multi-covering partition function of annuli of various types are given in \eqref{eq:annuli-explicit}.
\end{remark}

Let $\hat \psi$ and $\hat \psi'$ be the operators corresponding respectively to $Q, Q'$. 
It is clear from \eqref{eq:quiver-operator} that each of them is invertible. 
We define the cobordism operator as
\be\label{eq:cobordism-def}
	\hat\CO := \hat\psi \cdot (\hat\psi')^{-1}\,.
\ee
By definition, this takes $\hat\psi'$ into $\hat\psi$, i.e. $\hat\psi = \hat\CO \cdot \hat\psi'$, 
which after normal ordering (or acting on the vacuum skein element) gives the desired property \eqref{eq:cobord-Z}.

An interesting application of the cobordism operator is that it maps recursion relations of one quiver into those of the other quiver
\be
	\hat A\cdot \psi = 0\qquad \Rightarrow \qquad
	\hat A'\cdot \psi' = 0
\ee
with $\hat A,\hat A' \in \Sk(\bigsqcup_{j=1}^K T^2)$ and
\be
	\hat A' = \hat A\cdot \hat\CO\,.
\ee

\subsection{Example: linking basic disks}\label{sec:toy-example}
As a simple example of cobordism, we consider the linking operation of two basic disks.
Consider a quiver with two vertices and adjacency matrix
\be
	C = \left(\begin{array}{cc}
	0 & 0 \\
	0 & 0 
	\end{array}\right).
\ee
The partition function and the corresponding recursion relations are
\be
	\psi = (x_1;q^2)^{-1} (x_2;q^2)^{-1} 
	\qquad
	\left\{\begin{split}
	\hat A_1 = 1-\hat y_1 - \hat x_1 = 0\\
	\hat A_2 = 1-\hat y_2 - \hat x_2 = 0
	\end{split}\right.
\ee
We will build the cobordism to the quiver with linked vertices
\be
	C' = \left(\begin{array}{cc}
	0 & 1 \\
	1 & 0 
	\end{array}\right)
\ee
whose partition function and recursion relations are
\be
	\psi' = (x_1;q^2)_\infty^{-1} (x_2;q^2)_\infty^{-1} (x_1 x_2;q^2)_\infty
	\qquad
	\left\{\begin{split}
	\hat A_1' = 1-\hat y_1 - \hat x_1 \hat y_2 = 0\\
	\hat A_2' =1-\hat y_2 - \hat x_2 \hat y_1 = 0
	\end{split}\right.
\ee
We define the operator-valued quiver partition functions as follows
\be
	\hat\psi = (\hat x_2;q^2)_\infty^{-1}(\hat x_1;q^2)_\infty^{-1}
	\qquad
	\hat\psi' = (\hat x_2\hat y_1;q^2)_\infty^{-1}(\hat x_1;q^2)_\infty^{-1}
\ee
so that their normal-ordered expressions give back $\psi,\psi'$ respectively.
Then
\be
	\hat \CO := \hat\psi\cdot(\hat \psi')^{-1} =  (\hat x_2;q^2)_\infty^{-1} (\hat x_2\hat y_1;q^2)_\infty
\ee
We can then compute 
\be
\begin{split}
	\hat A_1 \cdot \hat \CO 
	& = (1-\hat y_1 - \hat x_1) \cdot (\hat x_2;q^2)_\infty^{-1} (\hat x_2\hat y_1;q^2)_\infty \\
	& = \CO\cdot (1-\hat y_1) - \hat x_1 \cdot (\hat x_2;q^2)_\infty^{-1} (\hat x_2\hat y_1;q^2)_\infty \\
	& = \CO\cdot (1-\hat y_1) -   (\hat x_2;q^2)_\infty^{-1} (q^{-2} \hat x_2\hat y_1;q^2)_\infty  \cdot \hat x_1\\
	& = \CO\cdot (1-\hat y_1 - (1-q^{-2} \hat x_2\hat y_1)  \hat x_1)\\
	& = \CO\cdot (1-\hat y_1 - \hat x_1 (1- \hat x_2\hat y_1)  )\qquad \text{(use $\hat A_2'$)}\\
	& = \CO\cdot (1-\hat y_1 - \hat x_1 \hat y_2) \\
	& = \CO\cdot \hat A_1'\\
\end{split}
\ee
which holds in the ideal generated by $\hat\psi'$.
A similar computation gives
\be
	\hat A_2 \cdot \hat \CO  = \hat \CO \cdot \hat A_2' 
\ee

\section{Quivers for the twisted Hopf link, and for $(2,2p+1)$-torus knots}\label{app:Hopf}

The HOMFLYPT polynomials $H_{(1)^{n_1},(n_2)}(a,q)$ of the (negative) Hopf link are well-known.
Here we provide a quiver description of these in terms of two disks on each conormal, and two annuli stretching between them and linking with the disks.
The quiver that we obtained is related, but not equivalent to, the conjectural one described in
\cite[Conjecture 2.3]{Ekholm:2024ceb}. 
The difference is due to the twisting of one annulus into an anti-annulus. 

We verified numerically that the (twisted) Hopf partition function admits the following quiver structure.
\be\label{eq:twisted-Hopf-HOMFLYPT}
\begin{split}
	& Z_{\widetilde{\text{Hopf}}}(a,q) 
	:= \sum_{n_1, n_2} H_{(1)^{n_1},(n_2)}(a,q) \xi_1^{n_1} \xi_2^{n_2} 
	\\
	& = 
	\sum_{n_1, n_2}\left\{
	\sum_{{k}=0}^1\left[\sum_{j=0}^{\text{Max}[n_1,n_2]-{k}}\frac{(-1)^{n_1-{k} } a^{2 \left(n_1-{k} \right)} q^{-2 j-2 {k}  \left(n_2+1\right)+2 {k}  n_1+n_1+n_2}(a^{-2};q^2)_{n_1-j-{k}}(a^2;q^2)_{n_2-j-{k}}}{(q^2;q^2)_{n_1-j-{k}}(q^2;q^2)_{n_2-j-{k}}}\right]
	\right\}
	\,
	\xi_1^{n_1} \xi_2^{n_2} 
	\\
	& = \sum_{d_1, d_2, d_3, d_4, d_5\geq 0}\sum_{d_6=0}^{1}
	q^{d_1^2 + d_4^2 + 2 (d_1+d_2) d_6-2 (d_3+d_4) d_6}
	\frac{\xi_1^{d_1}}{(q^2;q^2)_{d_1}} \frac{(-a^2 q \xi_1)^{d_2}}{(q^2;q^2)_{d_2}}
	\frac{(q \xi_2)^{d_3}}{(q^2;q^2)_{d_3}} \frac{(-a^2 \xi_2)^{d_4}}{(q^2;q^2)_{d_4}}
	(-a^2 \xi_1 \xi_2)^{d_5} (\xi_1 \xi_2)^{d_6}
\end{split}
\ee
where we have rescaled $H_{(1)^{n_1}, (n_2)}$ by introducing an overall factor of $a^{n_1+n_2}$ that gives exclusively positive powers of $a$.
Here $\xi_j$ for $j=1,2$ correspond to longitudinal generators of the skein modules $\Sk_{\mathfrak{gl}_1}(S^1\times D^2)_j$, see Figure \ref{fig:torus-link-flow-loops}. 
This is a generalized quiver partition function, where first two vertices of the quiver correspond to unknot disk-antidisk pair on the first conormal, the third and fourth vertices to the unknot disk-antidisk pair on the second conormal, and the last two vertices correspond respectively to an annulus and an anti-annulus.
This expression also encodes the $(a,q)$-degrees of each basic curve, as well as the quiver matrix that encodes linking on each conormal component
\be\label{eq:twisted-Hopf-quiver-matrix}
	C_{\widetilde{\text{Hopf}}} = \left(\begin{array}{cc|cc|cc}
		1 & 0 & 0 & 0 & 0 & 1 \\
		0 & 0 & 0 & 0 & 0 & 1 \\
		\hline
		0 & 0 & 0 & 0 & 0 & -1 \\
		0 & 0 & 0 & 1 & 0 & -1 \\
		\hline
		0 & 0 & 0 & 0 & 0 & 0 \\
		1 & 1 & -1 & -1 & 0 & 0 \\
	\end{array}\right)
\ee

Using this expression for $ H_{(1)^{n_1},(n_2)}(a,q)$ into \eqref{eq:gl1-formula-2-2p+1-knots}  leads to a generalized quiver description of curve counts on $M_K$.
Indeed observe that, introducing  quantum torus dual variables
\be
	\hat\eta_i\hat \xi_j = q^{2\delta_{ij}} \hat\xi_j\hat\eta_i\,,
\ee
corresponding to quantization of the holonomies in Figure \ref{fig:torus-link-flow-loops}, we can write
\be\label{eq:Hopf-Gamma-link}
\begin{split}
	& H_{(1)^{n_1},(n_2)}(a,q)
	\,
	q^{\Gamma_{11} n_1^2 + \Gamma_{22} n_2^2 +2\Gamma_{12} \, n_1 n_2}
	\,
	\xi_1^{n_1}\xi_2^{n_2}
	\\
	=&
	H_{(1)^{n_1},(n_2)}(a,q)
	\,
	q^{\Gamma_{11} n_1+\Gamma_{22}n_2}
	\,
	:
	(\xi_1 \hat\eta_1^{\Gamma_{11}}\hat\eta_2^{\Gamma_{12}})^{n_1}
	(\xi_2 \hat\eta_1^{\Gamma_{12}}\hat\eta_2^{\Gamma_{22}})^{n_2}
	:
\end{split}
\ee
where $:\ \dots\ :$ denotes normal ordering.
This allows to relate the linking matrix of basic curves in $\mathcal{L}_{\alpha}$ (given by \eqref{eq:twisted-Hopf-quiver-matrix})  to the linking matrix of basic curves in $M_K$ 
as follows \cite{Ekholm:2019lmb}
\be\label{eq:torus-knot-complement-quiver-matrix}
\begin{split}
	C_{M_{K}}
	&=
	C_{\widetilde{\text{Hopf}}} + \left(\begin{array}{c|c|c}
		\Gamma_{11} & \Gamma_{12} & \Gamma_{11}+\Gamma_{12} \\
		\hline
		\Gamma_{12} & \Gamma_{22} & \Gamma_{22}+\Gamma_{12} \\
		\hline
		\Gamma_{11}+\Gamma_{12} & \Gamma_{22}+\Gamma_{12} & \Gamma_{11}+\Gamma_{22}+2\Gamma_{12} \\
	\end{array}\right)
	\\
	&=
	\left(\begin{array}{cc|cc|cc}
		1 & 0 & -1 & -1 & -1 & 0 \\
		0 & 0 & -1 & -1 & -1 & 0 \\
		\hline
		-1 & -1 & -2p-1 & -2p-1 & -2p-2 & -2p-3 \\
		-1 & -1 & -2p-1 & -2p & -2p-2 & -2p-3 \\
		\hline
		-1 & -1 & -2p-2 & -2p-2 & -2p-3 & -2p-3 \\
		0 & 0 & -2p-3 & -2p-3 & -2p-3 & -2p-3 \\
	\end{array}\right)
	\,.
\end{split}
\ee

\newpage
\bibliographystyle{unsrt}
\bibliography{bibliography.bib}

\end{document}